\documentclass[12pt,a4paper]{amsart}
\usepackage{etex}

\usepackage[english]{babel}

\usepackage{mathtools}
\usepackage{cmll}
\usepackage{relsize}
\usepackage[latin1]{inputenc}
\usepackage{amsmath}
\usepackage{amscd}
\usepackage{amstext}
\usepackage{color}
\usepackage{amsbsy}
\usepackage{amsopn}
\usepackage{txfonts}
\usepackage{amsthm}
\usepackage{amsxtra}
\usepackage{upref}
\usepackage{amsfonts}
\usepackage{amssymb}
\usepackage{mathrsfs}
\usepackage{euscript}
\usepackage{array}
\usepackage{stmaryrd}
\usepackage{hyperref}
\usepackage{verbatim}

\usepackage[dvipsnames]{xcolor}

\usepackage{graphicx}
\usepackage{multicol}
\input{diagxy}
\usepackage{bussproofs}
\usepackage{cmll}
\usepackage[all]{xy}

\theoremstyle{plain}
\newcounter{thm} \numberwithin{thm}{section}
 \numberwithin{exa}{section}
 \numberwithin{def}{section}
\newcounter{rmk} \numberwithin{rmk}{section}

      \newtheorem{theorem}[thm]{Theorem}
      \newtheorem{lemma}[thm]{Lemma}
      \newtheorem{corollary}[thm]{Corollary}
      \newtheorem{proposition}[thm]{Proposition}
      
\theoremstyle{definition}
      \newtheorem{definition}[thm]{Definition}

 \theoremstyle{remark}
     \newtheorem{ex}[rmk]{Example}
      \newtheorem{remark}[rmk]{Remark}

      \theoremstyle{proof}

\newcounter{Step}
\newenvironment{step}[0]{\bigskip\addtocounter{Step}{1}\noindent\textbf{Step \theStep :} }{\
  \begin{flushright} \end{flushright}}

\newcommand\ITens{\otimes}
\newcommand\Tens[2]{{#1}\ITens{#2}}
\newcommand\One{1}
\newcommand\WEAK{\mathsf w}
\newcommand\COWEAK{\overline\WEAK}
\newcommand\DER{\mathsf{ d}}
\newcommand\CODER{\overline\DER}
\newcommand\CONTR{\mathsf c}
\newcommand\COCONTR{\overline\CONTR}

\newcommand\Weak[1]{\WEAK_{#1}}
\newcommand\Coweak[1]{\COWEAK_{#1}}
\newcommand\Contr[1]{\CONTR_{#1}}
\newcommand\Cocontr[1]{\COCONTR_{#1}}

\newcommand\Der[1]{\DER_{#1}}
\newcommand\Coder[1]{\CODER_{#1}}
\newcommand\Digg[1]{\operatorname{\mathsf p}_{#1}}

\newcommand\Excl[1]{\oc{#1}}
\def\occ{\oc_\mathscr{C}}

\def\N{\mbox{I\hspace{-.15em}N} }
\def\K{\mathbb{K}}
\def\R{\mbox{I\hspace{-.15em}R} }

\def\C{\hspace{.17em}\mbox{l\hspace{-.47em}C} }

\def\RR{\mathscr{R}}

\def\SS{\mathscr{S}}

\def\o{\otimes}

\newcommand{\kref}{k - \mathbf{Ref}}
\newcommand{\Cref}{\mathscr{C}-\mathbf{Ref}}

\newcommand{\Ref}{-\mathbf{Ref}}

\newcommand{\Kc}{\mathbf{Kc}}
\renewcommand{\Mc}{\mathbf{Mc}}
\newcommand{\Mb}{\mathbf{Mb}}
\newcommand{\McS}{\mathbf{McSch}}
\newcommand{\McMS}{\mathbf{Mc\mu Sch}}
\newcommand{\MbMS}{\mathbf{Mb\mu Sch}}

\newcommand{\MbMcM}{\mathbf{\mu McMb}}

\newcommand{\McMb}{\mathbf{McMb}}

\newcommand{\Sch}{\mathbf{Sch}}
\newcommand{\CSch}{\mathbf{CSch}}
\newcommand{\MS}{\mathbf{\mu Sch}}
\newcommand{\MSch}{\mathbf{\mu Sch}}
\newcommand{\MLCS}{\mathbf{\mu-LCS}}
\newcommand{\LCS}{\mathbf{LCS}}

\newcommand{\FDFS}{\mathbf{F\times DFS}}
\newcommand{\CLCS}{\mathbf{CLCS}}
\newcommand{\rRef}{\mathbf{{\rho}-Ref}}
\newcommand{\M}{\mathbf{\mu LCS}}
\newcommand{\gammaKc}{\gamma-\Kc}
\newcommand{\Conv}{\mathbf{Conv}}

\newcommand{\Cin}{\mathscr{C}^\infty}
\newcommand{\Cinco}{C^\infty_{co}}
\newcommand{\CinC}{C^\infty_{\mathscr{C}}}

\newcommand{\lcs}{lcs }
\newcommand{\mco}[1]{\widehat{#1}^M}

\begin{document}
\title{Models of Linear Logic based on the Schwartz $\varepsilon$-product.}
\author{Yoann Dabrowski}
\address{ Universit\'{e} de Lyon\\ 
Universit\'{e} Lyon 1\\
Institut Camille Jordan UMR 5208\\
43 blvd. du 11 novembre 1918\\
F-69622 Villeurbanne cedex\\
France
}
\email{ dabrowski@math.univ-lyon1.fr}
\author{Marie Kerjean}
\address{Univ Paris Diderot, Sorbonne Paris Cité, IRIF, UMR 7126, CNRS, F-75205 Paris, France}
\email{kerjean@irif.fr}

\begin{abstract}
From the interpretation of Linear Logic multiplicative disjunction as the $\varepsilon$-product defined by Laurent Schwartz, we construct several models of Differential Linear Logic based on usual mathematical notions of smooth maps. This improves on previous results in  \cite{BET12} based on convenient smoothness where only intuitionist models were built. We isolate a completeness condition, called \textbf{k}-quasi-completeness, and an associated notion stable by duality called \textbf{k}-reflexivity, allowing for a $*$-autonomous category of \textbf{k}-reflexive spaces in which the dual of the tensor product is the reflexive version of the $\varepsilon$ product. We adapt Meise's definition of Smooth maps into a first model of Differential Linear Logic, made of \textbf{k}-reflexive spaces. We also build two new models  of Linear Logic with conveniently smooth maps, on categories made respectively of Mackey-complete Schwartz spaces and Mackey-complete Nuclear Spaces (with extra reflexivity conditions). Varying slightly the notion of smoothness, one also recovers models of DiLL on the same $*$-autonomous categories. Throughout the article, we work within the setting of Dialogue categories where the tensor product is exactly the $\varepsilon$-product (without reflexivization).
\end{abstract}

\maketitle

\tableofcontents

\section{Introduction}

Since the discovery of linear logic by Girard \cite{Gir87}, thirty years ago,  many attempts have been made to obtain denotational models of linear logic in the context of some classes of vector spaces with  linear proofs interpreted as linear maps \cite{Blu96,Ehr02,Gir04,Ehr05,BET12}. 
Models of linear logic are often inspired by coherent spaces, or by the relational model of linear logic. Coherent Banach spaces \cite{Gir96}, coherent probabilistic or coherent quantum spaces \cite{Gir04} are Girard's attempts to extend the first model, as finiteness spaces \cite{Ehr05} or K\" othe spaces \cite{Ehr02} were designed by Ehrhard as a vectorial version of the relational model. 

Three difficulties appear in this semantical study of linear logic. The equivalence between a formula and its double negation in linear logic asks for the considered vector spaces to be isomorphic to their double duals. This is constraining in infinite dimension. This infinite dimensionality is strongly needed to interpret exponential connectives. Moreover, extension to natural models of Differential linear logic would aim at getting models of classical proofs by some classes of smooth maps, which should give a Cartesian closed category.  Finally, imposing a reflexivity condition to respect the first requirement usually implies issues of stability by natural tensor products of this condition, needed to model multiplicative connectives. This corresponds to the hard task of finding $*$-autonomous categories \cite{Bar}. As pointed out in \cite{Ehrhard16}, 
the only model of differential Linear logic using smooth maps \cite{BET12} 
misses annoyingly the $*$-autonomous property for classical linear logic.  

Our aim in this paper is to solve all these issues simultaneously and produce several denotational models of classical linear logic with some classes of smooth maps as morphism in the Kleisli category of the monad. We will show that the constraint of finding a $*$-autonomous category in a compatible way with a Cartesian closed category of smooth maps is even relevant to find better  mathematical notions of smooth maps in locally convex spaces. Let us explain this mathematical motivation first.

It seems that, historically, the development of differential calculus beyond normed spaces suffered from the lack of interplay between analytic considerations and  categorical, synthetic or logic ones. Partially as a consequence, analysts often forgot looking for good stability properties by duality and focused on one side of the topological or bornological viewpoint. 

Take one of the analytic summary of the early theory in the form of Keller's book \cite{Keller}. It already gives a unified and simplified approach based on continuity conditions  of derivatives in various senses. But it is well-known that in order to look for good categorical properties such as Cartesian closedness, the category of continuous maps is not a good starting point, the category of maps continuous on compact sets would be better. This appears strongly in all the developments made to recover continuity of evaluation on the topological product (instead of considering the product of a Cartesian closed category), which is unavoidable for full continuity of composition of derivatives in the chain rule. This leads to considering convergence notions beyond topological spaces on spaces of linear maps, but then, no abstract duality theory of those vector convergence spaces or abstract tensor product theory is developed. Either one remains with spaces of smooth maps that have tricky composition (of module type) between different notions of smoothness or 
 composition within the classes involving convergence vector spaces whose general theory remained underdeveloped with respect to locally convex spaces.  At the end, everything goes well only on restricted classes of spaces that lack almost any categorical stability properties, and nobody understands half of the notions introduced. The situation became slightly better when \cite{Meise} considered $k$-space conditions and obtained what analysts call kernel representation theorems (Seely isomorphisms for linear logicians), but still the class of spaces considered and the $k$-space conditions on products limited having a good categorical framework for the hugest classes of spaces: the only classes stable by products were Fr\'echet spaces and (DFM)-spaces, which are by their very nature not stable by duality. 
 
The general lesson here is that, if one wants to stay within better studied and commonly used locally convex spaces, one should better not stick to  functions continuous on products, and the corresponding projective topological tensor product, but always take tensor products that come from a $*$-autonomous category, since one also needs duality, or at least a closed category, to control the spaces of linear maps in which the derivatives take values. $*$-autonomous categories are the better behaved categories having all those data. Ideally, following the development of polarization of Linear logic in \cite{MelliesTabareau} inspired by game semantics, we are  able to get more flexibility and allow larger dialogue categories containing such $*$-autonomous categories as their category of continuation. We will get slightly better categorical properties on those larger categories.

A better categorical framework was later found and summarized in \cite{FrolicherKriegel,KrieglMichor} the so-called convenient smoothness.
A posteriori, as seen \cite{Kock}, the notion is closely related to synthetic differential geometry as diffeological spaces are. It chooses a very liberal notion of smoothness, that does not imply continuity except on very special compact sets, images of finite dimensional compact sets by smooth maps. It gives a nice Cartesian closed category and this enabled \cite{BET12} to obtain a model of intuitionistic differential linear logic. As we will see, this may give the wrong idea that this very liberal notion of smoothness is the only way of getting Cartesian closedness and it also takes the viewpoint of focusing on bornological properties. This is the main reason why, in our view, they don't obtain $*$-autonomous categories  since bornological locally convex spaces have complete duals which gives an asymmetric requirement on duals since they only need a much weaker Mackey-completeness on their spaces to work with their notion of smooth maps. We will obtain in this paper several models of linear logic using conveniently smooth maps, and we will explain logically this Mackey-completeness condition in section 6.2. It is exactly a compatibility condition on $F$ enabling to force our models to satisfy $!E\multimap F=(!E\multimap 1)\parr F$. Of course, as usual for vector spaces, our models will satisfy the mix rule making the unit for multiplicative connectives self-dual and this formula is interpreted mathematically as saying that smooth maps with value in some complete enough space are never a big deal and reduced by duality to the scalar case. But of course, this requires to identify the right completeness notion.

Another insight in our work is that the setting of models of Linear logic with smooth maps offers a decisive interpretation for the multiplicative disjunction $\parr$. In the setting of smooth functions, the epsilon product introduced by Laurent Schwartz is well studied and behave exactly as wanted: under some completeness condition, one indeed has $\Cin(E,\R) \varepsilon F   \simeq \Cin (E,F)$. This required for instance in \cite{Meise}  some restrictive conditions. We reduce these conditions to the definition \ref{k-complete} of $k$-complete spaces, which is also enough to get associativity and commutativity of $\varepsilon$. The interpretation of the tensor product follows as the negation of the $\varepsilon$ product. We would like to point out that plenty of possibilities exists for defining a topological tensor product (see subsection 2.2 for reminders), and that choosing to build our models from the $\varepsilon$ product offers a simplifying and intuitive guideline.
 
 With this background in mind, we can describe in more detail our results and our strategy.
 
 The first part of the paper will focus on building several $*$-autonomous categories. This work started with a negative lesson the first author learned from the second author's results in \cite{Kerjean}. Combining lots of strong properties on concrete spaces as for instance in \cite{BrouderDabrowski,Dab14a} will never be enough, it makes stability of these properties by tensor product and duality too hard. The only way out is to get a duality functor that makes spaces reflexive for this duality in order to correct tensor products by double dualization. The lesson is that identifying a proper notion of duality is therefore crucial if one wants to get an interesting analytic tensor product. From an analytic viewpoint, the inductive tensor product is too weak to deal with extensions to completions and therefore the weak dual or the Mackey dual, shown to work well with this tensor product in  \cite{Kerjean}, and which are the first duality functors implying easy reflexivity properties, are not enough for our purposes.  The insight is given by a result of \cite{Schwartz} that implies that another slightly different dual, the Arens dual always satisfies $((E'_c)'_c)'_c=E'_c$ hence one gets a functor enabling to get reflexive spaces, in some weakened sense of reflexivity. Moreover, Laurent Schwartz also developed there a related tensor product, the so called $\varepsilon$-product which is intimately related. This tensor product is a dual tensor product, generalization of the (dual) injective tensor product of (dual) Banach spaces and logicians would say it is a negative connective (for instance, as seen from its commutation with categorical projective limits) suitable for interpreting $\parr$. Moreover, it is strongly related with Seely-like isomorphisms for various classes of non-linear maps, from continuous maps (see e.g. \cite{Treves}) to smooth maps \cite{Meise}. It is also strongly related with nuclearity and Grothendieck's approximation property. This is thus a well established analytic tool desirable as a connective for a natural model of linear logic. We actually realize that most of the general properties for the Arens dual and the $\varepsilon$-product in \cite{Schwartz} are nicely deduced from a very general $*$-autonomous category we will explain at the end of the preliminary section 2. This first model of MALL that we will obtain takes seriously the lack of self-duality of the notion of locally convex spaces and notices that adjoining a bornology with weak compatibility conditions enables to get a framework where building a $*$-autonomous category is almost tautological. This may probably be related to some kind of Chu construction (cf. \cite{Bar1} and appendix to \cite{Bar}), but we won't investigate this expectation here. This is opposite to the consideration of bornological locally convex vector spaces where bornology and topology are linked to determine one another, here they can be almost independently chosen and correspond to encapsulating on the same space the topology of the space and of its dual (given by the bornology).
 
 Then, the work necessary to obtain a $*$-autonomous category of locally convex spaces is twofold, it requires to impose some completeness condition required to get associativity maps for the $\varepsilon$-product and then make the Arens dual compatible with some completion process to keep a reflexivity condition and get another duality functor with duals isomorphic to triple duals.
 We repeat this general plan twice in section 4 and 5 to obtain two extreme cases where this plan can be carried out. The first version uses the notion of completeness used in \cite{Schwartz}, or rather a slight variant we will call $k$-quasi-completeness and builds a model of MALL without extra requirement than being $k$-quasi-complete and the Arens dual of a $k$-quasi-complete space. This notion is equivalent to a reflexivity property that we call $k$-reflexivity. This first $*$-autonomous category is important because its positive tensor product is a completed variant of an algebraic tensor product $\o_\gamma$ having universal properties for bilinear maps which have a so-called hypocontinuity condition implying continuity on product of compact sets (see section 2.2 for more preliminary background). This suggested us a relation to the  well-known Cartesian closed category (equivalent to $k$-spaces) of topological spaces with maps all maps continuous on compact sets. Using strongly that we obtained a $*$-autonomous category, this enables us to provide the strongest notion of smoothness (on locally convex spaces) that we can imagine having a Cartesian closedness property. Contrary to convenient smoothness, it satisfies a much stronger continuity condition of all derivatives on compacts sets. Here, we thus combine the $*$-autonomous category with a Cartesian closed category in taking inspiration of the former to define the latter. This is developed in subsection 4.2.
  
 Then in section 5, we can turn to the complementary goal of finding a $*$-autonomous framework that will be well-suited for the already known and more liberal notion of smoothness, namely convenient smoothness. Here, we need to combine Mackey-completeness with a Schwartz space property to reach our goals. This is partially based on preliminary work in section 3 that actually makes appear a strong relation with Mackey duals which can actually replace Arens duals in this context, contrary to the first author's original intuition alluded to before. 
Technically, it is convenient to decompose our search for a $*$-autonomous category in two steps. Once identified the right duality notion and the corresponding reflexivity, we produce first a Dialogue category that deduces its structure from a kind of intertwining with the $*$-autonomous category obtained in section 2. Then we use \cite{MelliesTabareau} to recover a $*$-autonomous category in a standard way. This gives us the notion of $\rho$-dual and the $*$-autonomous category of $\rho$-Reflexive spaces. As before, those spaces can be described in saying that they are Mackey-complete with Mackey-complete Mackey dual (coinciding with Arens dual here) and they have the Schwartz topology associated to their Mackey topology. We gave the name $\rho$-dual since this was the first and more fruitful way (as seen its relation developed later with convenient smoothness) of obtaining a reflexive space  by duality, hence the letter $\rho$ for reflexive, while staying close to the letter $\sigma$ that would have remembered  the key Schwartz space property, but which was already taken by weak duals.

At the end of the first part of the paper, we have a kind of generic methodology enabling to produce $*$-autonomous categories of locally convex spaces from a kind of universal one from section 2. We also have obtained two examples that we want to extend to denotational models of full (differential) Linear logic in the second part. 

We start with the convenient smoothness setting in section 6. Actually we work with several topological variants of this setting (all having the same bornologification). To complement our identification of a logical meaning of Mackey-completeness, we also relate the extra Schwartz property condition with the logical 
interpretation of the transpose of the dereliction 
 $\DER^t E^*\multimap (!E)^*$. This asks for the topology on $E^*$ to be finer than the one induced by $(!E)^*$.  If moreover one wants to recover later a model of differential linear logic, we need a morphism : $\CODER : !E \to E$ such that $\DER \circ \CODER = Id_E$. This enforces the fact that the topology on $E^*$ must equal the one induced by $(!E)^*$. In this way, various natural topologies on conveniently smooth maps suggest various topologies on duals. We investigate in more detail the two extreme cases again, corresponding to well-known functional analytic conditions, both invented by Grothendieck, namely Schwartz topologies and the subclass of nuclear topologies. We obtain in that way in section 6 two denotational models of LL on the same $*$-autonomous category (of $\rho$-reflexive spaces), with the same Cartesian closed category of conveniently smooth maps, but with two different comonads. We actually show this difference in remark  \ref{TwoComonads} using Banach spaces without the approximation property. This also gives an insight of the functional analytic significance of the two structures. Technically, we use dialogue categories again, but not trough the models of tensor logic from \cite{MelliesTabareau}, but rather with a variant we introduce to keep  Cartesian closed the category equipped with non-linear maps as morphisms.

Finally, in section 7, we extend our models to models of (full) differential linear logic. In the $k$-reflexive space case, we have already identified the right notion of smooth maps for that in section 4, but in the $\rho$-reflexive case, which generalizes convenient vector spaces, we need to slightly change our notion of smoothness and introduce a corresponding notion of $\rho$-smoothness. Indeed, for the new $\rho$-reflexive spaces which are not bornological, the derivative of conveniently smooth maps are only bounded and need not be in spaces of continuous linear maps which are the maps of our $*$-autonomous categories. Taking inspiration of our use of dialogue categories and its interplay with Cartesian closed categories in section 6, we introduce in section 7.1 a notion merging dialogue categories with differential $\lambda$-categories of \cite{BuccEhrMan} and realize the correction of derivative we need in a general context in section 7.2. This enables us to get a class of models of DiLL with at least 3 new different models in that way, one on $k$-reflexive spaces (section 7.4) and two being on the same category of $\rho$-reflexive spaces with $\rho$-smooth maps (section 7.3). This is done concretely by considering only smooth maps whose derivatives are smooth in their non-linear variable with value in (iterated) spaces of  continuous linear maps.

\subsection{A first look at the interpretation of Linear Logic constructions}
For the reader familiar with other denotational models of Linear Logic, we would like to point out some of the constructions involved in the first model $\kref$.  Our two other main models make use of similar constructions, with a touch of Mackey-completeness.

First, we define a \emph{$k$-quasi-complete} space as a space in which the  closed absolutely convex cover of a compact subset is still compact. We detail a procedure of $k$-quasi-completion, which is done inductively. 

We take as the interpretation $E^\bot$ of the negation the $k$-quasi completion of $E'_c$, the dual of $E$ endowed with the  compact-open topology, at least when $E$ is $k$-quasi-complete. 
We define $!E$ as $\Cin_{co}(E, \K)^\perp$, the $k$-quasicompletion of the dual of the space of scalar smooth functions. This definition is in fact enforced as soon as we have a $*$-autonomous category with a co-Kleisli category of smooth maps.  Here we define the space of smooth functions as the space of infinitely many times G\^ateaux-differentiable  functions with derivatives continuous on compacts, with a good topology (see subsection \ref{sec:ksmooth}). This definition, adapted from the one of Meise, allows for Cartesian closedeness.

We then interpret the $\parr$ as the (double dual of) the $\varepsilon$ product: $E \varepsilon F = \mathcal{L}_\epsilon(E'_c,F)$, the space of all linear continuous functions from $E'_c$ to $F$ endowed with the topology of uniform convergence on equicontinuous subsets of $E'$. The interpretation of $\otimes$ is the dual of $\varepsilon$, and can be seen as the $k$-quasi-completion of a certain topological tensor product $\otimes_\gamma$. 

The additive connectives $\times$ and $\oplus$ are easily interpreted as the product and the co-product. In our vectorial setting, they coincide in finite arity.

In the differential setting, codereliction $\CODER$ is interpreted as usual by the transpose of differentiation at $0$ of scalar smooth maps.

\part{Three Models of MALL}
\section{Preliminaries}

We will be working with \emph{locally convex separated} topological vector spaces. We will write in short \lcs for such spaces, following \cite{Kothe} in that respect. We refer to the book by Jarchow \cite{Jarchow} for basic definitions. We will recall the definitions from Schwartz \cite{Schwartz} concerning the $\varepsilon$ product. We write $E=F$ when two \lcs are equal algebraically and $E \simeq F$ when the \lcs equal topologically as well.

\begin{remark} We will call \textit{embedding} a continuous linear map $E\to F$ which is one-to-one and with the topology of $E$ induced from this inclusion. In the functional analytic literature \cite[p 2]{Kothe2} this is called topological monomorphism and abbreviated monomorphism, this is also the case in \cite{Schwartz}. This disagrees with the categorical terminology, hence our choice of a more consensual  term. A monomorphism in the category of separated locally convex vector spaces is an injective continuous linear map, and a regular monomorphism is a embedding with closed image (a \textit{closed embedding}).  A regular monomorphism in the category of non-separated locally convex spaces coincide with an embedding but we won't use this category.\end{remark}
\begin{remark} We will use projective kernels as in \cite{Kothe}. They are more general than categorical limits, which are more general than projective limits of \cite{Kothe}, which coincide with those categorical limits indexed by directed sets. 
\end{remark}

\subsection{Reminder on topological vector spaces}
\label{subsec:remindertvs}

\begin{definition}
Consider $E$ a vector space. A bornology on $E$ is a collection of sets (the bounded sets of $E$) such that the union of all those sets covers E, and such that the collection is stable under inclusion and finite unions.
\end{definition}

When $E$ is a topological vector space, one defines the Von-Neumann bornology $\beta$ as those sets which are absorbed by any neighbourhood of $0$. Without any other precision, the name bounded set will refer to a bounded set for the Von-Neumann bornology. 
Other examples of bornology are the collections $\gamma$ of all absolutely convex compact subsets of $E$, and $\sigma$ of all bipolars of finite sets. When $E$ is a space of continuous linear maps, one can also consider on $E$ the bornology $\varepsilon$ of all equicontinuous parts of $E$. When $E$ is a \lcs, we only consider saturated bornologies, namely those which contain the subsets of the bipolars of each of its
members.

\begin{definition}
Consider $E$, $F$, $G$ topological vector spaces and $h : E \times F \mapsto G$ a bilinear map. 
\begin{itemize}
\item $h$ is continuous if it is continuous from $E \times F$ endowed with the product topology to $G$.
\item $h$ is separately continuous if for any $x\in E$ and $y \in F$, $h(x,.)$ is continuous from $F$ to $G$ and $h(.,y)$ is continuous from $E$ to $G$.
\item Consider $\mathcal{B}_1$ (resp. $\mathcal{B}_2$) a bornology on $E$ (resp. $F$). Then $h$ is said to be $\mathcal{B}_1$,$\mathcal{B}_2$ hypocontinuous \cite{Schwartz2} if for
every $0$-neighbourhood $W$ in $G$, every bounded set $A_E$ in $E$, and every bounded set
$A_F$ in $F$, there are $0$-neighbourhoods $V_F \subset F$ and $V_E \subset E$ such that $h (A_E \times V_F) \subset W$ and $ h(V_E \times A_F) \subset W$. When no precision is given, an hypocontinuous bilinear map is a map hypocontinuous for both Von-Neumann bornologies. 
\end{itemize}
\end{definition}

Consider $A$ an absolutely convex and bounded subset of a \lcs $E$. We write $E_A$ for the linear span of $A$ in $E$. It is a normed space when endowed with the Minkowski functional 
$$||x||_A\equiv p_A (x) = \inf\ \left\{ \lambda \in \R^+\ |\ x \in \lambda  A \right\}. $$
A \lcs $E$ is said to be \emph{Mackey-complete} (or locally complete \cite[10.2]{Jarchow}) when for every bounded closed and absolutely convex subset $A$, $E_A$ is a Banach space. A sequence is \emph{ Mackey-convergent} if it is convergent in some $E_B$. 
This notion can be generalized for any bornology $\mathcal{B}$ on $E$ : a sequence is said to be \emph{$\mathcal{B}$-convergent} if it is convergent is some $E_B$ for $B \in \mathcal{B}$. \newline

Consider $E$ a \lcs and $\tau$ its topology. Recall that a filter in $E'$ is said to be equicontinuously convergent if it is $\varepsilon$-convergent. $E$ is a \emph{Schwartz space} if it is endowed with a Schwartz topology, that is a space such that every continuously convergent filter in $E'$ converges equicontinuously.
We refer to \cite[chapter 1]{HogbeNlendMoscatelli} and \cite[sections 10.4, 21.1]{Jarchow} for an overview on Schwartz topologies. We recall some facts below. 

The finest Schwartz locally convex topology coarser than $\tau$ is the topology $\tau_0$ of uniform convergence on sequences of $E'$ converging equicontinuously to $0$. We write $\mathscr{S}(E)=\mathscr{S}(E,\tau)=(E,\tau_0)$. We have $\mathscr{S}(E)'=E'$, and $\mathscr{S}(E)$ is always separated. A \lcs $E$ is a Schwartz space if and only if $\mathscr{S}(E)=E$, if and only if the completion $\tilde{E}$ is  a Schwartz space. We do know also that $\mathscr{S}(E)$ is Mackey-complete as soon as $E$ is (as both space have the same dual, they have the same bounded sets by Mackey-Arens Theorem). Any subspace of a Schwartz space is a Schwartz space.

\subsection{Reminder on tensor products and duals of locally convex spaces.}

Several topologies can be associated with the tensor product of two topological vector space. 

\begin{definition}
Consider $E$ and $F$ two lcs. 
\begin{itemize}
\item The projective tensor product $E\otimes_\pi F$ is the finest locally convex topology on $E \otimes F$ making $E \times F \rightarrow E \otimes_{\pi} F$ continuous. 
\item The inductive tensor product $E\otimes_i F$ is the finest locally convex topology on $E \otimes F$ making $E \times F \rightarrow E \otimes_{i} F$ separately continuous. 
\item The hypocontinous tensor product $E\otimes_\beta F$ is the finest locally convex topology on $E \otimes F$ making $E \times F \rightarrow E \otimes_{\beta} F$ hypocontinuous.
\item The $\gamma$ tensor product $E\otimes_\gamma F$ is the finest locally convex topology on $E \otimes F$ making $E \times F \rightarrow E \otimes_{\gamma} F$ $\gamma$-hypocontinuous.
\item Suppose that $E$ and $F$ are duals. The $\varepsilon$-hypocontinous  tensor product $E\otimes_{\beta e} F$ is the finest locally convex topology on $E \otimes F$ making $E \times F \rightarrow E \otimes_{\beta e} F$ $\varepsilon$-hypocontinuous.
\item Consider $\mathcal{B}_1$ (resp. $\mathcal{B}_2$) a bornology on $E$ (resp. $F$).The $\mathcal{B}_1-\mathcal{B}_2$-hypocontinous  tensor product $E\otimes_{\mathcal{B}_1,\mathcal{B}_2} F$ is the finest locally convex topology on $E \otimes F$ making $E \times F \rightarrow E \otimes_{\mathcal{B}_1,\mathcal{B}_2} F$ $\mathcal{B}_1,\mathcal{B}_2$-hypocontinuous.

\end{itemize}
\end{definition}
All the above tensor products, except the last one, are commutative and the ${\otimes}_\pi
$ product is associative. 
With the last generic notation one gets $\o_i=\o_{\sigma,\sigma}, \o_\beta=\o_{\beta,\beta},\o_\gamma=\o_{\gamma,\gamma}, \otimes_{\beta e}=\o_{\varepsilon,\varepsilon}$ and we will sometimes consider during proofs non-symmetric variants such as: $\o_{\varepsilon,\gamma}, \o_{\sigma,\gamma}$ etc. Note that the injective tensor product $\o_\varepsilon\neq \o_{\varepsilon,\varepsilon}$ is a dual version we will discuss later. It does not have the above kind of universal properties.



\begin{definition}
One can define several topologies on the dual $E'$ of a \lcs $E$. We will make use of :
\begin{itemize}
\item The strong dual $E'_{\beta}$, endowed with the strong topology $\beta (E',E)$ of uniform convergence on bounded subsets of $E$.
\item The \emph{Arens dual} $E'_c$ endowed with the topology $\gamma (E',E)$ of uniform convergence on absolutely convex compact subsets of $E$.
\item The \emph{Mackey dual} $E'_\mu$, endowed with the Mackey topology of uniform convergence on absolutely convex weakly compact subsets of $E$. 
\item The weak dual $E'_\sigma$ endowed with the weak topology $\sigma (E',E)$ of simple convergence on points of $E$. 
\item The $\varepsilon$-dual  $E'_\varepsilon$  of a dual $E=F'$  is the dual $E'$ endowed with the topology of uniform convergence on equicontinuous sets in $F'$.
\end{itemize}
Remember that when it is considered as a set of linear forms acting on $E'$, $E$ is always endowed with the topology of uniform convergence on equicontinuous parts of $E'$, equivalent to the original topology of $E$, hence $(E'_\mu)'_\epsilon\simeq (E'_c)'_\epsilon\simeq (E'_\sigma)'_\epsilon\simeq E$. A \lcs is said to be reflexive when it is topologically equal to its strong double dual $(E'_\beta)'_\beta$.
\end{definition}

The \emph{Mackey-Arens theorem} \cite[8.5.5]{Jarchow} states that whenever $E'$ is endowed with a topology finer that the weak topology, and coarser than the Mackey topology, then $E =E''$ algebraically. Thus one has \begin{equation} \label{Ec}
 E = (E'_c)'.
\end{equation} 

As explained by Laurent Schwartz \cite[section 1]{Schwartz}, the equality $E \simeq (E'_c)'_c$ holds as soon as $E$ is endowed with its $\gamma$ topology, i.e. with the topology of uniform convergence on absolutely convex compact subsets of $E'_c$. He proves moreover that an Arens dual is always endowed with its $\gamma$-topology, that is : 
$E'_c \simeq ((E'_c)'_c)'_c $. This fact is the starting point of the construction of a $*$-autonomous category in section \ref{sec:rhorefl}. \newline

The \emph{$\varepsilon$-product} has been extensively used and studied by Laurent Schwartz \cite[section 1]{Schwartz}.
By definition $E\varepsilon F= (E'_c\otimes_{\beta e} F'_c)'$ is the set of $\varepsilon$-hypocontinuous bilinear forms on the duals $E'_c$ and $F'_c$.
When $E,F$ have their $\gamma$ topologies  this is the same as $E\varepsilon F= (E'_c \otimes_{\gamma} F'_c)'.$



The topology on $E\varepsilon F$ is the topology of uniform convergence on products of equicontinuous sets in $E', F'$. If $E,F$ are quasi-complete spaces (resp. complete spaces , resp. complete spaces  with the approximation property) so is $E\varepsilon F$ (see \cite[Prop 3 p29, Corol 1 p 47]{Schwartz}). The $\varepsilon$ tensor product $E\otimes_\varepsilon F$ coincides with the topology on $E\otimes F$ induced by $E\varepsilon F$ (see \cite[Prop 11 p46]{Schwartz}), ${\otimes}_\varepsilon$ is associative, and $E\hat{\otimes}_\varepsilon F\simeq E\varepsilon F$ if $E,F$ are complete and $E$ has the approximation property. 

The $\varepsilon$-product is also defined on any finite number of space as $\varepsilon_i E_i$, the space of $\varepsilon$-equicontinuous multilinear forms on $\prod_i (E_i)'_c$,, endowed the the topology of uniform convergence on equicontinuous sets. Schwartz proves the associativity of the $\varepsilon$-product when the spaces are quasi-complete. We do so when the spaces are Mackey-complete and Schwartz, see lemma \ref{AssocScwhartz}.
\subsection{Dialogue  and $*$-autonomous categories}
It is well known that models of (classical) linear logic requires building $*$-autonomous categories introduced in  \cite{Bar}. If we add categorical completeness, they give models of MALL. We need some background about them, as well as a generalization introduced in \cite{MelliesTabareau}: the notion of Dialogue category that will serve us as an intermediate in between a general $*$-autonomous category we will introduce in the next subsection and more specific ones requiring a kind of reflexivity of locally convex spaces that we will obtain by double dualization, hence in moving to the so-called continuation category of the Dialogue category.

Recall the definition (cf. \cite{Bar}):
\begin{definition}
A \emph{$*$-autonomous category} is a symmetric monoidal closed category $(\mathcal{C},\o_{\mathcal{C}},1_\mathcal{C},[\cdot,\cdot]_{\mathcal{C}})$ with an object $\perp$ giving an equivalence of categories $(\cdot)^*=[\cdot,\perp]_{\mathcal{C}}:\mathcal{C}^{op}\to \mathcal{C}$ and with the canonical map $d_A:A\to  (A^*)^*$ being a natural isomorphism.
\end{definition}
Since our primary data will be functional, based on space of linear maps (and tensorial structure will be deduced since it requires various completions), we will need a consequence of the discussion in \cite[(4.4) (4.5) p 14-15]{Bar}. We outline the proof for the reader's convenience. We refer to \cite[p 25]{Schipper} (see also \cite{DayLaplaza}) for the definition of symmetric closed category.

\begin{lemma}\label{ClosedCatToStarAut}
Let $(\mathcal{C},1_\mathcal{C},[\cdot,\cdot]_{\mathcal{C}})$ a  symmetric closed category, which especially implies there is a natural isomorphism $s_{X,Y,Z}:[X,[Y,Z]_{\mathcal{C}}]_{\mathcal{C}}\to[Y,[X,Z]_{\mathcal{C}}]_{\mathcal{C}}$ and let $\perp=[1_\mathcal{C},1_\mathcal{C}]_{\mathcal{C}}.$ Assume moreover that there is a natural isomorphism, $d_{X}:X\to[[X, \perp]_{\mathcal{C}},\perp]_{\mathcal{C}}$. Define $X^*=[X,\perp]_{\mathcal{C}}$ and $(X\o_\mathcal{C} Y)=([X,Y^*]_{\mathcal{C}})^*$. Then $(\mathcal{C},\o_{\mathcal{C}},1_\mathcal{C},[\cdot,\cdot]_{\mathcal{C}},(\cdot)^*)$ is a $*$-autonomous category.
\end{lemma}
\begin{proof} Recall for instance that $i_X:X\to [1_\mathcal{C},X]_{\mathcal{C}}$ is an available natural isomorphism. Note first that there is a natural isomorphism 
defined by:
\begin{align*}
\xymatrix@C=20pt{
d_{X,Y}:[X,Y]_{\mathcal{C}} \ar[rr]^{[X, d_Y]_{\mathcal{C}}}
&& [X,Y^{**}]_{\mathcal{C}} 
\ar[r]^{s_{X,Y^*,\perp}\qquad}&
[Y^*,[X,\perp]_{\mathcal{C}}]_{\mathcal{C}}.
}
\end{align*}

The assumptions give a natural isomorphism:
\begin{align*}\mathcal{C}(X,[Y,Z^*]_{\mathcal{C}})&\simeq \mathcal{C}(1,[X,[Y,Z^*]_{\mathcal{C}}]_{\mathcal{C}})\simeq \mathcal{C}(1,[X,[Z,Y^*]_{\mathcal{C}}]_{\mathcal{C}})\\&\simeq \mathcal{C}(1,[Z,[X,Y^*]_{\mathcal{C}}]_{\mathcal{C}})\simeq \mathcal{C}(Z,[Y,X^*]_{\mathcal{C}}).\end{align*}
Moreover, we have a bijection $\mathcal{C}(X^*,Y^*)\simeq \mathcal{C}(1,[X^*,Y^*]_{\mathcal{C}})\simeq \mathcal{C}(1,[Y,X]_{\mathcal{C}})\simeq \mathcal{C}(Y,X)$ so that the assumptions in \cite[(4.4)]{Bar} are satisfied. His discussion 
in (4.5) gives a natural isomorphism: 
$\pi_{XYZ}:\mathcal{C}(X\o_\mathcal{C}Y,Z)\to \mathcal{C}(X,[Y,Z]_{\mathcal{C}}]_{\mathcal{C}})$.

We are thus in the third basic situation of \cite[IV .4]{Schipper} which gives (from $s$) a natural transformation 
$p_{XYZ}:[X\o_\mathcal{C}Y,Z]_{\mathcal{C}}\to [X,[Y,Z]_{\mathcal{C}}]_{\mathcal{C}}.$ Then the proof of his Prop VI.4.2 proves his compatibility condition MSCC1 from SCC3, hence we have a monoidal symmetric closed category in the sense of \cite[Def IV.3.1]{Schipper}.

Then \cite[Thm VI.6.2 p 136]{Schipper} gives us a usual symmetric monoidal closed category in the sense of \cite{EilenbergKelly}. This concludes.
\end{proof}

We finally recall the more general definition in \cite{MelliesTabareau}:
\begin{definition}
A \emph{Dialogue category} is a symmetric monoidal category $(\mathcal{C},\o_{\mathcal{C}},1_\mathcal{C})$ with a functor, called tensorial negation: $\neg:\mathcal{C}\to \mathcal{C}^{op}$ which is associated to a natural bijection $\varphi_{A,B,C}:\mathcal{C}(A\o_{\mathcal{C}} B,\neg C)\simeq \mathcal{C}(A,\neg (B\o_{\mathcal{C}} C))$ and satisfying the commutative diagram with associators $Ass^{\o_\mathcal{C}}_{A,B,C}:A\o_{\mathcal{C}} (B\o_{\mathcal{C}} C)\to(A\o_{\mathcal{C}} B)\o_{\mathcal{C}} C$:
\begin{equation}\label{DialogueCompatibility}
\xymatrix@C=40pt{
 \mathcal{C}((A\o_{\mathcal{C}} B)\o_{\mathcal{C}} C,\neg D)\ar[d]|{\mathcal{C}(Ass^{\o_\mathcal{C}}_{A,B,C},\neg D)}\ar[r]^{\varphi_{A\o_{\mathcal{C}} B,C,D}}  & \mathcal{C}(A\o_{\mathcal{C}} B,\neg (C \o_{\mathcal{C}}D)) \ar[r]^{\varphi_{A, B,C\o_{\mathcal{C}} D}\ \ \ }
&\mathcal{C}\Big(A,\neg \big[B\o_{\mathcal{C}} (C\o_{\mathcal{C}} D)\big]\Big)\\
\mathcal{C}(A\o_{\mathcal{C}} (B\o_{\mathcal{C}} C),\neg D)  \ar[rr]^{\varphi_{A,B\o C,D}} & &\mathcal{C}\Big(A,\neg \big[(B\o_{\mathcal{C}} C)\o_{\mathcal{C}} D\big]\Big)\ar[u]|{\mathcal{C}(A,\neg Ass^{\o_\mathcal{C}}_{B,C,D})} }
\end{equation}
\end{definition}

\subsection{A model of MALL making appear the Arens dual and the Schwartz $\varepsilon$-product}

We introduce a first $*$-autonomous category that captures categorically the part of \cite{Schwartz} that does not use quasi-completeness. Since bornological and topological concepts are dual to one another, it is natural to fix a saturated bornology on $E$ in order to create a self-dual concept. Then, if one wants every object to be a dual object as in a $*$-autonomous category, one must consider only bornologies that can arise as the natural bornology on the dual, namely, the equicontinuous bornology. We could take a precompactness condition to ensure that, but to make appear the Arens dual and $\varepsilon$-product (and not the polar topology and Meise's variant of the $\varepsilon$-product), we use instead a compactness condition. A weak-compactness condition would work for the self-duality requirement by Mackey Theorem but not for dealing with tensor products.

We will thus use a (saturated, topological) variant of the notion of compactology used in \cite[p 157]{Jarchow}. We say that a saturated bornology $B_E$ on a \lcs $E$ is a \textit{compactology} if it consists of relatively compact sets. Hence, the bipolar of each bounded set for this bornology is an absolutely convex compact set in $E$, and it is bounded for this bornology. A separated locally convex space with a compactology will be called a \textit{compactological locally convex space}.
\begin{definition}
Let $\LCS$ be the category of separated locally convex spaces with continuous linear maps and $\CLCS$  the category of  compactological locally convex spaces, with maps given by bounded continuous linear maps.  For $E,F\in \textbf{CLCS}$ the internal Hom $L_b(E,F)$ is the above set of maps given the topology of uniform convergence on the bornology of $E$ and the bornology of equibounded equicontinuous sets. We call $E'_b=L_b(E,\K)$ (its bornology is merely the equicontinuous bornology, see step 1 of next proof).
The algebraic tensor product $E\o_H F$ is the algebraic tensor product with the topology having the universal property for  $B_E,B_F$-hypocontinuous maps, and the bornology generated by bipolars of sets $A\o C$ for $A\in B_E,C\in B_F$.
\end{definition}

Note that we didn't claim that $E\o_H F$ is in $\textbf{CLCS}$, it may not be. It gives a generic hypocontinuous tensor product.
Note that composition of bounded continuous linear maps are of the same type, hence $\textbf{CLCS}$ is indeed a category.
 
 Recall also that $\LCS$ is complete and cocomplete since it has small products and coproducts, kernels and cokernels (given by the quotient by the closure $\overline{Im[f-g]}$) \cite[ \S 18.3.(1,2,5), 18.5.(1)]{Kothe}.
 
In order to state simultaneously a variant adapted to Schwartz spaces, we introduce a variant:

 \begin{definition}
Let $\Sch\subset\LCS$ be the full subcategory of Schwartz spaces and $\CSch\subset \CLCS$  the full subcategory of Schwartz compactological lcs, namely those spaces which are Schwartz as locally convex spaces and for which $E'_b$ is a Schwartz \lcs too.
\end{definition}
This second condition is well-known to be equivalent to the bornology being a Schwartz bornology \cite{HogbeNlendMoscatelli}, and to a more concrete one:
\begin{lemma}\label{SchwartzBorno}
For $E\in \CLCS$, $E'_b$ is a Schwartz \lcs if and only if every bounded set in $B_E$ is included in the closed absolutely convex cover of a $B_E$-null sequence. 
\end{lemma}
\begin{proof}
$E'_b$ is Schwartz if and only if  $E'_b=\mathscr{S}(E'_b).$ But $\mathscr{S}(E'_b)$ is known to be the topology of uniform convergence on $(B_E)_{c_0}$ the saturated bornology  generated by $B_E$-null sequences of $E=(E'_b)'$ \cite[Prop 10.4.4]{Jarchow}. Since both bornologies are saturated this means \cite[\S 21 .1. (4)]{Kothe} that $E'_b$ is a Schwartz space if and only if $B_E=(B_E)_{c_0}$. 
\end{proof}
We call $\mathscr{S}L_b(E,F)$ the same \lcs as $L_b(E,F)$ but given the bornology $(B_{L_b(E,F)})_{c_0}$ namely the associated Schwartz bornology. Note that $\mathscr{S}L_b(E,\K)=E'_b$  as compactological \lcs for $E\in \CSch$. 
 
\begin{theorem}\label{FirstMALL}
$\textbf{CLCS}$ (resp. $\CSch$) is a complete and cocomplete $*$-autonomous category with dualizing object $\K$ and internal Hom $L_b(E,F)$ (resp. $\mathscr{S}L_b(E,F)$).
\begin{enumerate}\item The functor $(.)'_c:\LCS\to \textbf{CLCS}^{op}$ giving the Arens dual the equicontinuous bornology, is right adjoint to $U((.)'_b)$, with $U$ the underlying lcs and 
$U((.)'_b)\circ (.)'_c=Id_{\LCS}.$ The functor $(.)'_\sigma:\LCS\to \textbf{CLCS}^{op}$ giving the weak dual the equicontinuous bornology, is left adjoint to $U((.)'_b)$ and 
$U((.)'_b)\circ (.)'_\sigma=Id_{\LCS}.$ 
\item
The functor $U:\textbf{CLCS}\to \LCS$  is left adjoint and also left inverse to $(.)_c$, the functor  $E\mapsto E_c$ the space with the same topology and the absolutely convex compact bornology. $U$ is right adjoint to  $(.)_\sigma$, the functor  $E\mapsto E_\sigma$ the space with the same topology and the saturated bornology generated by finite sets. $U,(.)_c,(.)_\sigma$ are faithful.
\item
The functor $U:\textbf{CSch}\to \Sch$  is left adjoint and also left inverse to $(.)_{sc}$, the functor  $E\mapsto E_{sc}$ the space with the same topology and the Schwartz bornology associated to the absolutely convex compact bornology. $U$ is again right adjoint to  $(.)_\sigma$ (restriction of the previous one).  $U,(.)_{sc},(.)_\sigma$ are faithful.
\item The $\varepsilon$-product in $\LCS$ is given by $E\varepsilon F=U(E_c\parr_b F_c)$ with $G\parr_b H=L_b(G'_b,H)$  and of course the Arens dual by $U((E_c)'_b),$ and more generally $L_c(E,F)=U(L_b(E_c,F_c)).$ The inductive tensor product $E\o_iF=U(E_\sigma\o_b F_\sigma)$ with $G\o_b H=(G'_b\parr_b H'_b)'_b$ and of course the weak dual is $U((E_\sigma)'_b).$
 \end{enumerate}
\end{theorem}

\begin{proof}

\begin{step}
Internal Hom functors $L_b,\mathscr{S}L_b$.
\end{step}
We first need to check that the equibounded equicontinuous bornology on $L_b(E,F)$ is made of relatively  compact sets when $E,F\in \textbf{CLCS}$. In the case $F=\K$,  the bornology is the equicontinuous bornology since an equicontinuous set is equibounded for von Neumann bornologies \cite[\S 39.3.(1)]{Kothe2}. 
Our claimed statement is then explained in \cite[note 4 p 16]{Schwartz} since it is proved there that every equicontinuous closed absolutely convex set is compact in $E'_c=(U(E))'_c$ and our assumption that the saturated bornology is made of relatively compact sets implies there is a continuous map $E'_c\to E'_b$. This proves the case $F=\K$.

 Note that by definition, $G=L_b(E,F)$ identifies with the dual $H=(E\o_H F'_b)'_b.$ 
  Indeed, the choice of bornologies implies the topology of $H$ is the topology of uniform convergence on equicontinuous sets of $F'$ and on bounded sets of $E$ which is the topology of $G$. An equicontinuous set in $H$ is known to be an equihypocontinuous set \cite[p 10]{Schwartz2}, i.e. a set taking a bounded set in E and giving an equicontinuous set in $(F'_b)'$, namely a bounded set in $F$, hence the equibounded condition, and taking symmetrically a bounded set in $F'_b$ i.e. an equicontinuous set and sending it to an equicontinuous set in $E'$, hence the equicontinuity condition \cite[\S 39.3.(4)]{Kothe2}. 
 
{ Let $E\widehat{\o}_H F'_b\subset E\widetilde{\o}_H F'_b$ the subset of the completion obtained by taking the union of bipolars of bounded sets.} It is easy to see this is a vector subspace on which we put the induced topology.
 One deduces that $H=(E\widehat{\o}_H F'_b)'_b$ where the $E\widehat{\o} F'_b$ is given the bornology generated by bipolars of bounded sets (which covers it by our choice of subspace). Indeed the completion does not change the dual and the equicontinuous sets herein \cite[\S 21.4.(5)]{Kothe} and the extension to bipolars does not change the topology on the dual either. But in $E\widehat{\o}_H F'_b$, bounded sets for the above bornology are included into  bipolars of tensor product of bounded sets.
Let us recall why tensor products $A\o B$ of such bounded sets are precompact in $E\o_H F'_b$ (hence also in $E\widehat{\o}_H F'_b$ by \cite[\S 15.6.(7)]{Kothe}) if $E,F\in \textbf{CLCS}$. Take $U'$ (resp. $U$) a neighbourhood of $0$ in it (resp. such that $U+U\subset U'$), by definition there is a neighbourhood $V$ (resp. $W$) of $0$ in $E$ (resp. $ F'_b$) such that $V\o B\subset U$ (resp. $A\o W\subset U$). Since $A,B$ are relatively compact hence  precompact, cover $A\subset \cup_i x_i+V$, $x_i\in A$ (resp. $B\subset \cup_j y_j+W$, $y_j\in B$) so that one gets the finite cover giving totally boundedness: 
$$A\o B\subset \cup_i x_i\o B+V\o B\subset \cup_{i,j} x_i\o y_j+x_i\o W +V\o B\subset \cup_{i,j} x_i\o y_j+U+U\subset\cup_{i,j} x_i\o y_j+U'.$$

Note that we used strongly compactness here in order to exploit hypocontinuity, and weak compactness and the definition of Jarchow for compactologies wouldn't work with our argument.

  Thus from hypocontinuity, we deduced the canonical map $E\times F'_b\to E\widehat{\o}_H F'_b$ send $A\times B$ to a precompact (using \cite[\S 5.6.(2)]{Kothe}), hence its bipolar is complete (since we took the bipolar in the completion which is closed there) and precompact \cite[\S 20.6.(2)]{Kothe} hence compact (by definition \cite[\S 5.6]{Kothe}). Thus  $E\widehat{\o}_H F'_b\in\textbf{CLCS},$ if $E,F\in\textbf{CLCS}.$ From the first case for the dual, one deduces $L_b(E,F)\in\textbf{CLCS}$ in this case.
Moreover, once the next step obtained, we will know $E\widehat{\o}_H F'_b\simeq E\o_b F'_b.$

Let us explain why $\CSch$ is stable by the above internal Hom functor. First for $E,F\in\CSch$ we must see that  $L_b(E,F)$ is a Schwartz \lcs. By definition $F,E'_b$ are Schwartz spaces, hence this is \cite[Thm 16.4.1]{Jarchow}. From the choice of bornology, $\mathscr{S}L_b(E,F)\in \CSch$ since by definition $U((\mathscr{S}L_b(E,F))'_b)\simeq \mathscr{S}(U((L_b(E,F))'_b)).$

 \begin{step}
$\CLCS$ and $\CSch$ as Closed categories.
\end{step}

It is well know that $\mathbf{Vect}$ the category of Vector spaces is a symmetric monoidal category and especially a closed category in the sense of \cite{EilenbergKelly}. $\CLCS\subset \mathbf{Vect}$ is a (far from being full) subcategory, but we see that we can induce maps on our smaller internal Hom.
 Indeed, the linear map $L_{FG}^E:L_b(F,G)\to L_b(L_b(E,F),L_b(E,G))$ is well defined since a bounded family in $L_b(F,G)$ is equibounded, hence it sends an equibounded set in $L_b(E,F)$ to an equibounded set in $L_b(E,G)$, and also equicontinuous, hence its transpose sends an equicontinuous set in $(L_b(E,G))'$ (described as bipolars of bounded sets in $E$ tensored with equicontinuous sets in $G'$) to an equicontinuous  set in $(L_b(E,F))'$. This reasoning implies $L_{FG}^E$ is indeed valued in continuous equibounded maps and even bounded with our choice of bornologies.  Moreover we claim $L_{FG}^E$ is continuous. Indeed, an equicontinuous set in 
$(L_b(L_b(E,F),L_b(E,G)))'$ is generated by the bipolar of equicontinuous $C$ set in $G'$, a bounded set $B$ in $E$ and an equibounded set $A$ in $(L_b(E,F))$ and the transpose consider $A(B)\subset F$ and $C$ to generate a bipolar which is indeed equicontinuous in  $(L_b(F,G))'.$ Hence, $L_{FG}^E$ is a map of our category. Similarly, the morphism giving identity maps $j_E:\K\to L_b(E,E)$ is indeed valued in the smaller space and the canonical $i_E:E\to L_b(\K,E)$ indeed sends a bounded set to an equibounded equicontinuous set 
and is tautologically equicontinuous. Now all the relations for a closed category are induced from those in $\mathbf{Vect}$ by restriction. The naturality conditions are easy.

Let us deduce the case of $\CSch$. 
First, let us see that for $E\in\CSch$, \begin{equation}\label{InnerAdjointS}\mathscr{S}L_b(E,\mathscr{S}L_b(F,G))=\mathscr{S}L_b(E,L_b(F,G))\end{equation}

By definition of boundedness, a map $f\in L_b(E,L_b(F,G))$ sends a Mackey-null sequence in $E$ to a Mackey-null sequence in $L_b(F,G)$ hence by continuity the bipolar of such a sequence is sent to a bounded set in $\mathscr{S}L_b(F,G)$, hence from lemma \ref{SchwartzBorno}, so is a bounded set in $E$. We deduce the algebraic equality in \eqref{InnerAdjointS}. The topology of $L_b(E,H)$ only depends on the topology of $H$, hence we have the topological equality since both target spaces have the same topology. It remains to compare the bornologies. But from the equal target topologies, again, the equicontinuity condition is the same on both spaces hence boundedness of the map $L_b(E,\mathscr{S}L_b(F,G))\to L_b(E,L_b(F,G))$ is obvious. Take a sequence $f_n$ of maps Mackey-null in $L_b(E,L_b(F,G))$ hence in the Banach space generated by the Banach disk $D$ of another Mackey-null sequence $(g_n)$.
Let us see that $\{g_n,n\in \N\}^{oo}$ is equibounded in $L_b(E,\mathscr{S}L_b(F,G))$.
 For take $B\subset L_b(E,L_b(F,G))$ the disk for $(g_n)$ with $||g_n||_B\to 0$ and take a typical generating bounded set $A=\{x_n, n\in \N\}^{oo}\subset E$ for $x_n$  $B_E$-Mackey-null. Then $g_n(A)\subset\{g_m(x_n), m,n\in \N\}^{oo}=:C$ and $||g_m(x_n)||_{(B(A))^{oo}}\leq ||g_m||_{B}||x_n||_{A}$ and since $B(A)$ is bounded by equiboundedness of $B$, $(g_m(x_n))$ is Mackey-null, hence 
 $C$ is bounded in $\mathscr{S}L_b(F,G)$ and hence $D=\{g_n, n\in \N\}^{oo}$ is equibounded as stated. But since $D$ is also bounded in 
$L_b(E,L_b(F,G))$ it is also equicontinuous, hence finally, bounded in $L_b(E,\mathscr{S}L_b(F,G))$. This gives that $f_n$ Mackey-null there which concludes to the bornological equality in \eqref{InnerAdjointS}.

As a consequence, for $E,F,G\in \CSch$, the previous map $L_{FG}^E$ induces a map \begin{align*}L_{FG}^E:\mathscr{S}L_b(F,G)\to \mathscr{S}L_b(L_b(E,F),L_b(E,G))\to& \mathscr{S}L_b(\mathscr{S}L_b(E,F),L_b(E,G))\\&=\mathscr{S}L_b(\mathscr{S}L_b(E,F),\mathscr{S}L_b(E,G))\end{align*}
coinciding with the previous one as map. Note that we used the canonical continuous equibounded map $L_b(L_b(E,F),G)\to L_b(\mathscr{S}L_b(E,F),G)$ obviously given by the definition of associated Schwartz bornologies which is a smaller bornology.

 \begin{step}
$*$-autonomous property.
\end{step}
First note that $L_b(E,F)\simeq L_b(F'_b,E'_b)$ by transposition. Indeed, the space of maps and their bornologies are the same since equicontinuity (resp. equiboundedness) $E\to F$ is equivalent to equiboundedness (resp. equicontinuity) of the transpose $F'_b\to E'_b$ for equicontinuous bornologies (resp. for topologies of uniform convergence of corresponding bounded sets). Moreover the topology is the same since it is the topology of uniform convergence on bounded sets of $E$ (identical to equicontinuous sets of $(E'_b)'$) and equicontinuous sets of $F'$ (identical to bounded sets for $F'_b$).  Similarly $\mathscr{S}L_b(E,F)\simeq \mathscr{S}L_b(F'_b,E'_b)$ since on both sides one considers the bornology generated by Mackey-null sequences for the same bornology.

It remains to check  $L_b(E,L_b(F,G))\simeq  L_b(F,L_b(E,G))$. The map is of course the canonical map. Equiboundedness in the first space means sending a bounded set in $E$ and a bounded set in $F$ to a bounded set in $G$ and also a bounded set in $E$ and an equicontinuous set in $G'$ to an equicontinuous set in $F'$. This second condition is exactly equicontinuity $F\to L_b(E,G)$. Finally, analogously, equicontinuity  $E\to L_b(F,G)$ implies it sends a bounded set in $F$ and an equicontinuous set in $G'$ to an equicontinuous set in $E'$ which was the missing part of equiboundedness in  $L_b(F,L_b(E,G))$. The identification of spaces and bornologies follows. Finally, the topology on both spaces is the topology of uniform convergence on products of bounded sets of $E,F$.

Again, the naturality conditions of the above two isomorphisms are easy, and the last one induces from $\mathbf{Vect}$ again the structure of a symmetric closed category, hence lemma \ref{ClosedCatToStarAut} concludes to $\CLCS$ $*$-autonomous.

Let us prove the corresponding statement for $\CSch$.  Note that \eqref{InnerAdjointS} implies the compactological isomorphism $$\mathscr{S}L_b(E,\mathscr{S}L_b(F,G))\simeq \mathscr{S}L_b(E,L_b(F,G))\simeq  \mathscr{S}L_b(F,L_b(E,G)) \simeq \mathscr{S}L_b(F,\mathscr{S}L_b(E,G)).$$

Hence, application of lemma \ref{ClosedCatToStarAut} concludes in the same way.

\begin{step}
Completeness and cocompleteness.
\end{step}
Let us describe first coproducts and cokernels. This is easy in $\CLCS$ it is given by the colimit of separated locally convex spaces, given the corresponding final bornology.
Explicitely, the coproduct is the direct sum of vector spaces with coproduct topology and the bornology is the one generated by finite sum of bounded sets, hence included in finite sums of compact sets which are compact \cite[\S 15.6.(8)]{Kothe}. Hence the direct sum is in $\textbf{CLCS}$ and clearly has the universal property from those of topolocial/bornological direct sums.
For the cokernel of $f,g:E\to F$, we take the coproduct in 
$\LCS$, $Coker(f,g)=F/\overline{(f-g)(E)}$ with the final bornology, i.e. the bornology generated by images of bounded sets. Since the quotient map is continuous between Hausdorff spaces, the image of a compact containing a bounded set is compact, hence $Coker(f,g)\in\CLCS$. Again the universal property comes from the one in locally convex and bornological spaces. Completeness then follows from the $*$-autonomous property since one can see $\lim_i E_i=(\mathrm{colim}_i (E_i)'_b)'_b$ gives a limit.

Similarly in $\CSch$ the colimit of Schwartz bornologies is still Schwartz since the dual is a projective limit of Schwartz spaces hence a Schwartz space (cf lemma \ref{SchwartzFunctor}).  We therefore claim that the colimit is the Schwartz topological space associated to the colimit in $\CLCS$ with same bornology. Indeed this is allowed since there are more compact sets hence the compatibility condition in $\CLCS$ is still satisfied and functoriality of  $\mathscr{S}$ in lemma \ref{SchwartzFunctor} implies the universal property.

\begin{step}
Adjunctions and consequences.
\end{step}
The fact that the stated maps are functors is easy. 
We start by the adjunction for $U$ in (2): $\LCS(U(F),E)=L_b(F,E_c)=\textbf{CLCS}(F,E_c)$ since the extra condition of boundedness beyond continuity is implied by the fact that a bounded set in $F$ is contained in an absolutely convex compact set which is sent to the same kind of set by a continuous linear map.
Similarly, $\LCS(E,U(F))=L_b(E_\sigma,F)=\textbf{CLCS}(E_\sigma,F)$ since the image of a finite set is always in any bornology (which must cover $E$ and is stable by union), hence the equiboundedness is also automatic. 

For $(3),$ since $(E_\sigma)'_b=E'_\sigma$ is always Schwartz, the functor $(.)_\sigma$ restricts to the new context, hence the adjunction. Moreover $U((E_c)'_b)=\mathscr{S}(E'_c)$ by construction. The key identity 
$\Sch(U(F),E)=L_b(F,E_c)=\CSch(F,E_{sc})$ comes from the fact that a Mackey-null sequence in $F$ is send by a continuous function to a Mackey-null sequence for the compact bornology hence to a bounded set in $E_{sc}$.
All naturality conditions are easy.

Moreover, for the adjunction in (1), we have the equality as set (using involutivity and functoriality of $(.)'_b$ and the previous adjunction): 
$$\textbf{CLCS}^{op}(F,E'_c)=L_b((E_c)'_b,F)=L_b(F'_b,E_c)=\LCS(U(F'_b),E),$$
 $$\textbf{CLCS}^{op}(E'_\sigma,F)=L_b(F,(E_\sigma)'_b)=L_b(E_\sigma,F'_b)=\LCS(E,U(F'_b)).$$
The other claimed identities are obvious by definition.
\end{proof}

The second named author explored in \cite{Kerjean} models of linear logic using the positive product $\o_i$ and $(.)'_\sigma$. We will use in this work the negative product $\varepsilon$ and the Arens dual $(.)'_c$ appearing with a dual role in the previous result. Let us summarize the properties obtained in \cite{Schwartz} that are consequences of our categorical framework.

\begin{corollary}\begin{enumerate}\item Let $E_i\in \LCS, i\in I.$
The iterated $\epsilon$-product is $\varepsilon_{i\in I}E_i=U(\parr_{b,i\in I} (E_i)_c)$, it is symmetric in its arguments and commute with limits. 
\item There is a continuous injection $(E_1\varepsilon E_2 \varepsilon E_3)\to E_1\varepsilon (E_2\varepsilon E_3).$ 
\item For any continuous linear map $f:F_1\to E_1$ (resp. continuous injection, closed embedding), so is $f\varepsilon Id:F_1\varepsilon E_2\to E_1\varepsilon E_2.$
\end{enumerate}
\end{corollary}
Note that (3) is also valid for non-closed embeddings and (2) is also an embedding \cite{Schwartz}, but this is not a categorical consequence of our setting. 
\begin{proof}
The equality in (1) is a reformulation of definitions, symmetry is an obvious consequence. Commutation with limits come from the fact that $U, ()_c$ are right adjoints and $\parr_b$ commutes with limits from universal properties. 

Using associativity of $\parr_b$: $E_1\varepsilon (E_2\varepsilon E_3)=U((E_1)_c\parr_b[U((E_2)_c\parr_b(E_3)_c)]_c)$ hence functoriality and the natural transformation coming from adjunction $Id\to (U(\cdot))_c$ concludes to the continuous map in (2). It is moreover a monomorphism since $E\to (U(E))_c$ is one since $U(E)\to U((U(E))_c)$ is identity and $U$ reflects monomorphisms and one can use the argument for (3).

For (3) functorialities give definition of the map, and recall that closed embeddings in $\LCS$ are merely regular monomorphisms, hence a limit, explaining its commutation by (1). If $f$ is a monomomorphism, in categorical sense, so is $U(f)$ using a right inverse for $U$ and so is $(f)_c$ since $U((f)_c)=f$ and $U$ reflects monomorphisms as any faithful functor. Hence it suffices to see $\parr_b$ preserves monomorphisms but $g_1,g_2:X\to E\parr_b F$ correspond by Cartesian closedness to maps $X \o_b F'_b\to E$ that are equal when composed with $f:E\to G$ if $f$ monomorphism, hence so is $f\parr_b id_F.$
\end{proof}

In general, we have just seen that $\varepsilon$ has features for a negative connective as $\parr$, but it lacks associativity. We will have to work to recover a monoidal category, and then models of LL. In that respect, we want to make our fix of associativity compatible with a class of smooth maps, this will be the second leitmotiv. We don't know if there is an extension of the model of MALL given by $\textbf{CLCS}$ into a model of LL using a kind of smooth maps.

\section{Mackey-complete spaces and a first interpretation for $\parr$}\label{sec:zeta}
Towards our goal of obtaining a model of LL with conveniently smooth maps as non-linear morphisms, it is natural to follow \cite{KrieglMichor} and consider Mackey-complete spaces as in \cite{BET12,KerjeanT}. In order to fix associativity of $\varepsilon$ in this context, we will see appear the supplementary Schwartz condition. This is not such surprising as seen the relation with Mackey-completeness appearing for instance in \cite[chap 10]{Jarchow} which treats them simultaneously. This Schwartz space condition will enable to replace Arens duals by Mackey duals (lemma \ref{MackeyArensSchwartz}) and thus simplify lots of arguments in identifying duals as Mackey-completions of inductive tensor products (lemma \ref{DualArensMc}). This will strongly simplify the construction of the strength for our doubly negation monad later in section \ref{sec:rhodual}. Technically, this is possible by various results of \cite{Kothe} which points out a nice alternative tensor product $\eta$ which replaces $\varepsilon$-product exactly in switching Arens with Mackey duals. But we need to combine $\eta$ with $\mathscr{S}$ in a clever way in yet another product $\zeta$ in order to get an associative product. Said in words, this is a product which enables to ensure at least two Schwartz spaces among three in an associativity relation. This technicality is thus a reflection of the fact that for Mackey-complete spaces, one needs to have at least 2 Schwartz spaces among three to get a 3 term associator for an $\varepsilon$-product. In course of getting our associativity, we get the crucial relation $\mathscr{S}(\mathscr{S}(E)\varepsilon F)=\mathscr{S}(E)\varepsilon  \mathscr{S}(F)$ in corollary \ref{SchwartzCondExpEpsilon}.  This is surprising because this seems really specific to Schwartz spaces and we are completely unable to prove an analogue for the associated nuclear topology functor  $\mathscr{N}$, even if we expect it for the less useful associated strongly nuclear topology.
We conclude in Theorem \ref{zetaparr} with our first interpretation of $\parr$ as $\zeta$.

\subsection{A Mackey-Completion with continuous canonical map}

Note that for a $\gamma$-Mackey-Cauchy sequence, topological convergence  is equivalent to Mackey convergence (since the class of bounded sets is generated by bounded closed sets).


\begin{remark}\label{MackeyInclusion}Note also that is $E\subset F$ is a continuous inclusion, then a Mackey-Cauchy/convergent sequence  in $E$ is also Mackey-Cauchy/convergent in $F$ since a linear map is bounded.
\end{remark}

We now recall two alternative constructions of the Mackey-completion, from above by intersection and from below by union. The first construction is already considered in \cite{PerrezCarreras}.

\begin{lemma}\label{gammacompletion}
The intersection $\widehat{E}^M$ of all Mackey-complete spaces containing $E$ and contained in the completion $\tilde{E}$ of $E$, is Mackey-complete and called the Mackey-completion of $E$.

We define $E_{M;0}=E$, and for any ordinal $\lambda$, the subspace $E_{M;\lambda +1}=\cup_{(x_n)_{n\geq 0}\in  M(E_{M;\lambda })}\overline{ \Gamma(\{x_n, n\geq 0\})} \subset \tilde{E}$ where the union runs over all Mackey-Cauchy sequences  $  M(E_{M;\lambda })$  of $E_{M;\lambda }$, and the closure is taken in the completion. We also let for any limit ordinal $E_{M;\lambda}=\cup_{\mu<\lambda}E_{M;\mu}$. Then for any ordinal $\lambda$, $E_{M;\lambda }\subset\widehat{E}^M$ and eventually for {$\lambda\geq  \omega_1$ the first uncountable ordinal, we have equality.}
\end{lemma}
\begin{proof}
The first statement comes from stability of Mackey-completeness  by intersection (using remark \ref{MackeyInclusion}). It is easy to see that $E_{M;\lambda }$ is a subspace. 
{At stage $E_{M;\omega_1 +1}$, by uncountable cofinality of $\omega_1$ any Mackey-Cauchy sequence has to be in $E_{M;\lambda}$ for some $\lambda<\omega_1$ and thus each term of the union is in some $E_{M;\lambda +1}$, therefore $E_{M;\omega_1 +1}=E_{M;\omega_1}.$}

Moreover if at some $\lambda$, $E_{M;\lambda +1}=E_{M;\lambda }$, then by definition, $E_{M,\lambda }$ is Mackey-complete (since we add with every sequence its limit that exists in the completion which is Mackey-complete) and then the ordinal sequence is eventually constant. Then, we have $E_{M,\lambda }\supset \widehat{E}^M$. One shows for any $\lambda$ the converse by transfinite induction. For, let $(x_n)_{n\geq 0}$ is a Mackey-Cauchy sequence in $E_{M;\lambda }\subset F:=\widehat{E}^M$ . Consider $A$ a closed bounded absolutely convex set in $F$ with $x_n\to x$ in $F_A$. Then by \cite[Prop 10.2.1]{Jarchow}, $F_A$ is a Banach space, thus $\overline{\Gamma(\{x_n,n\geq 0\})}$ computed in this space is complete and thus compact (since $\{x\}\cup\{x_n, n\geq 0\}$ is compact in the Banach space), thus its image in $\tilde{E}$ is compact and thus agrees with the closure computed there. Thus every element of $\overline{\Gamma(\{x_n,n\geq 0\})}$ is a limit in $E_A$ of a sequence in $\Gamma(\{x_n,n\geq 0\})\subset E_{M;\lambda }$ thus by Mackey-completeness, $\overline{\Gamma(\{x_n,n\geq 0\})}\subset F.$
We thus conclude to the successor step $E_{M;\lambda +1}\subset \widehat{E}^M$, the limit step is obvious.

\end{proof}

 \subsection{A $\parr$ for Mackey-complete spaces}
 
 We first define a variant of the Schwartz $\epsilon$-product:
 \begin{definition}
For two separated \lcs $E$ and $F$, we define $E\eta F=L(E'_\mu,F)$ the space of continuous linear maps on the Mackey dual with the topology of uniform convergence on equicontinuous sets of $E'$. We write $\eta(E,F)=(E'_\mu\o_{\beta e}F'_\mu)'$ with the topology of uniform convergence on products of equicontinuous sets and $\zeta(E,F)\equiv E\zeta F$ for the same space with the weakest topology making continuous the canonical maps to $\eta(\mathscr{S}(E),F)$ and $\eta(E,\mathscr{S}(F)).$
\end{definition}

This space $E\eta F$ has already been studied in \cite{Kothe2} and we can summarize  its properties similar to the Schwartz $\epsilon$ product in the next proposition, after a couple of lemmas.

We first recall an important property of the associated Schwartz topology from \cite{Junek}. These properties follow from the fact that the ideal of compact operators on Banach spaces in injective, closed and surjective. Especially, from \cite[Corol 6.3.9]{Junek} it is an idempotent ideal.

\begin{lemma}\label{SchwartzFunctor}
The associated Schwartz topology functor $\mathscr{S}$ commutes with arbitrary products, quotients and embeddings (and as a consequence with arbitrary projective kernels or categorical limits).
\end{lemma}
\begin{proof}
For products and (topological) quotients, this is \cite[Prop 7.4.2]{Junek}. For embeddings (that he calls topological injections), this is \cite[Prop 7.4.8]{Junek} based on the previous ex 7.4.7. The consequence comes from the fact that any projective kernel is a subspace of a product, as a categorical limit is a kernel of a map between products.
\end{proof}

We will also often use the following relation with duals 
\begin{lemma}\label{MackeyArensSchwartz}
If $E$ is a Schwartz \lcs, $E'_c\simeq E'_\mu$ so that for any  \lcs $F$, $E\eta F\simeq  E\varepsilon F$ topologically.  Thus for any \lcs $E$, $E'_\mu\simeq (\mathscr{S}(E))'_c$.
\end{lemma}

\begin{proof}
Take $K$ an absolutely convex $\sigma(E',E)$-weakly compact in $E$, it is an absolutely convex closed set in $E$ and precompact as any bounded set in a Schwartz space \cite[3,\S 15 Prop 4]{Horvath}. 
\cite[IV.5 Rmq 2]{Bourbaki} concludes to $K$ complete since $E\to (E'_\mu)'_\sigma$ continuous with same dual and $K$ complete in $(E'_\mu)'_\sigma$, and since $K$ precompact, it is therefore compact in $E$. As a consequence $E'_c$ is the Mackey topology. Hence, $E\eta F=L(E'_\mu,F)=L(E'_c,F)=  E\varepsilon F$ algebraically and the topologies are defined in the same way. The last statement comes from the first and $E'_\mu\simeq (\mathscr{S}(E))'_\mu$.
\end{proof}
 
 \begin{proposition}\label{etageneral}
 Let $E,F,G,H$, be separated \lcs, then :
\begin{enumerate}
\item We have a topological canonical isomorphisms $\zeta(E,F)=\zeta(F,E)$, $$E\eta F=F\eta E\simeq \eta(E,F)$$ and we have a continuous linear map $E\varepsilon F\to E\eta F\to E \zeta F$ which is a topological isomorphism as soon as either $E$ or $F$ is a Schwartz space. In general, $E\varepsilon F$ is a closed subspace of  $E\eta F$.
\item $E\eta F$ is complete if and only if $E$ and $F$ are complete. 
\item If $A:G\to E$, $B:H\to F$ are linear continuous (resp. linear continuous one-to-one, resp. embeddings) so are the tensor product map $(A\eta B),(A\zeta B) $  both defined by $(A\eta B)(f)=B\circ f\circ A^t, f\in L(E'_\mu, F).$
\item If $F=\mathrm{K}_{i\in I}(A_i)^{-1}F_i$ is a projective kernel so are $E\eta F=\mathrm{K}_{i\in I}(1\eta A_i)^{-1}E\eta F_i $ and $E\zeta F=\mathrm{K}_{i\in I}(1\zeta A_i)^{-1}E\zeta F_i.$ Moreover, both $\eta,\zeta,\varepsilon$ commute with categorical limits in $\LCS$.
\item $\eta(E,F)$ is also the set of bilinear forms on $E'_\mu\times F'_\mu$ which are separately continuous. As a consequence, the Mackey topology $((E'_\mu\o_{\beta e}F'_\mu)'_\mu)'_\mu=E'_\mu\o_{i}F'_\mu$ is the inductive tensor product.
\item A set is bounded in $\eta(E,F)$ or $\zeta(E,F)$  if and only if it is $\epsilon$-equihypocontinuous on  $E'_\beta\times F'_\beta$.
\end{enumerate} 
 
 \end{proposition}

 \begin{proof}
It is crucial to note that $\eta(E,F)$, $\eta(\mathscr{S}(E),F)$, $\eta(E,\mathscr{S}(F))$ are the same space algebraically since $(\mathscr{S}(E))'_\mu\simeq E'_\mu$.
 
 (2)  is \cite[\S 40.4.(5)]{Kothe} and (1) is similar to the first statement there. With more detail functoriality of Mackey dual gives a map $L(E'_\mu,F)\to L(F'_\mu,(E'_\mu)'_\mu)$ and since $(E'_\mu)'_\mu\to E$ continuous we have also a map $L(F'_\mu,(E'_\mu)'_\mu)\to L(F'_\mu,E)$. This explains the first map of the first isomorphism (also explained in \cite[Corol 8.6.5]{Jarchow}). The canonical linear map from a bilinear map in $\eta(E,F)$ is clearly in $E\eta F$, conversely, if $A\in E\eta F$, $\langle A(.),.\rangle_{F,F'}$ is right $\epsilon$-hypocontinuous by definition and the other side of the hypocontinuity comes from the $A^t\in F\eta E$. 
 
The closed subspace property is \cite[\S 43.3.(4)]{Kothe2}. 
 
 (3) for $\eta$ is \cite[\S 44.4.(3,5,6)]{Kothe2}.
For $\zeta$ since $A,B$ are continuous after taking the functor $\mathscr{S}$, one deduces $A\eta B$ is continuous (resp one-to-one, resp. an embedding using lemma \ref{SchwartzFunctor}) on $$\eta(\mathscr{S}(G),H)\to\eta(\mathscr{S}(E),F),\ \  \eta(G,\mathscr{S}(H))\to\eta(E,\mathscr{S}(F))$$ and this conclude by universal properties of projective kernels (with two terms) for $\zeta$. Since the spaces are the same algebraically, the fact that the maps are one-to-one also follows.
 
  (4) The $\eta$ case with kernels is a variant of \cite[\S 44.5.(4)]{Kothe2} which is also a direct application of \cite[\S 39.8.(10)]{Kothe2}. As a consequence $E\zeta F$ is a projective kernel of $\mathscr{S}(E)\eta F=\mathrm{K}_{i\in I}(1\eta A_i)^{-1}[\mathscr{S}(E)]\eta F_i $ and, using lemma \ref{SchwartzFunctor} again, of  :$$E\eta \mathscr{S}(F)=E\eta \Big(\mathrm{K}_{i\in I}A_i^{-1}\mathscr{S}(F_i)\Big)=\mathrm{K}_{i\in I}(1\eta A_i)^{-1}\big(E\eta \mathscr{S}(F_i)\big).$$
The transitivity of locally convex kernels (coming from their universal property) concludes. 

For categorical limits, it suffices commutation with products and kernels. In any case the continuous map $I:(\lim E_i)\eta F\to \lim (E_i\eta F)$ comes from universal properties
, it remains to see it is an algebraic isomorphism, since then the topological isomorphism will follow from the kernel case. We build the inverse as follows, for $f\in F'_\mu$, the continuous evaluation map $E_i\eta F=L(F'_\mu,E_i)\to E_i$ induces a  continuous linear map $J_f:\lim (E_i\eta F)\to (\lim E_i)$. It is clearly linear in $f$ and gives a bilinear map $J:\lim (E_i\eta F)\times F'_\mu\to (\lim E_i).$ We have to see it is separately continuous yielding a linear inverse map $I^{-1}$ and then continuity of this map. We divide into the product and kernel case. 

For products one needs for $g\in\prod_{i\in I}(E_i\eta F)$ $J(g,.)^t: (\prod_{i\in I}E_i)'\to (F'_\mu)'$ send equicontinuous sets i.e. a finite sum of equicontinuous set in the sum $\sum_{i\in I}E_i'$ to an equicontinuous set in $(F'_\mu)'$. But absolutely convex weakly compact sets are stable by bipolars of sum, since they are stable by bipolars of finite unions \cite[\S 20.6.(5)]{Kothe} (they don't even need closure to be compact, absolutely convex cover is enough), hence it suffices to see the case of images of equicontinuous sets $E_i'\to F$ but they are equicontinuous by assumption. This gives the separate continuity in this case. Similarly, to see the continuity of $I^{-1}$ in this case means that we take $A\subset F'$ equicontinuous and a sum of equicontinuous sets $B_i$ in $(\prod_{i\in I}E_i)'$ and one notices that $(I^{-1})^t(A\times \sum B_i)\subset\sum (I^{-1})^t(A\times  B_i)$ is a sum of equicontinuous sets in $(\prod_{i\in I}(E_i\eta F))'$ and it is by hypothesis equicontinuous.

For kernels, of $f,g:E\to G$, $I:Ker(f-g)\eta F\to  Ker(f\eta id_F-g\eta id_F)$ is an embedding by (3) since source and target are embeddings in $E\eta F$, the separate continuity is obtained by restriction of the one of $E\eta F\times F'_\mu\to E\supset Ker(f-g)$ and similarly continuity by restriction of $E\eta F\to L(F'_\mu,E)$.

For $\zeta$, this is then a consequence of this and lemma \ref{SchwartzFunctor} again.
 (5) is an easier variant of \cite[Rmq 1 p 25]{Schwartz}. Of course $\eta(E,F)$ is included in the space of separately continuous forms. Conversely, if $f:E'_\mu\times F'_\mu\to \K$ is separately continuous, from \cite[Corol 8.6.5]{Jarchow}, it is also separately continuous on $E'_\sigma\times F'_\sigma$ and the non-trivial implication follows from \cite[\S 40.4.(5)]{Kothe}. For the second part, the fact that both algebraic tensor products have the same dual implies there is, by Arens-Mackey Theorem, a continuous identity map $((E'_\mu\o_{\beta e}F'_\mu)'_\mu)'_\mu\to E'_\mu\o_{i}F'_\mu$. Conversely, one uses the universal property of the inductive tensor product which gives a separately continuous map $E'_\mu\times F'_\mu\to E'_\mu\o_{\beta e}F'_\mu$.
But applying functoriality of Mackey duals on each side gives for each $x\in E'_\mu$ a continuous map 
  $F'_\mu\to((E'_\mu\o_{\beta e}F'_\mu)'_\mu)'_\mu$ and by symmetry, a separately continuous map $E'_\mu\times F'_\mu\to ((E'_\mu\o_{\beta e}F'_\mu)'_\mu)'_\mu.$ The universal property of the inductive tensor product  again concludes.

 (6) can be obtained for $\eta$ with the same reasoning as in  \cite[\S 44.3.(1)]{Kothe2}.  For $\zeta$ the first case gives by definition equivalence with $\epsilon$-equihypocontinuity both on  $(\mathscr{S}(E))'_\beta\times F'_\beta$ and on  $E'_\beta\times (\mathscr{S}(F))'_\beta$. But the second implies that for equicontinuous on E', one gets an equicontinuous family on $(\mathscr{S}(F))'_\beta\simeq F'_\beta$ and the first gives the converse, and the other conditions are weaker, hence the equivalence with the first formulation.
 \end{proof}

 We then deduce a Mackey-completeness result:
 
 \begin{proposition}\label{Meta}
 If $L_1$ and $L_2$ are separated Mackey-complete locally convex spaces, then  so are $L_1\eta L_2$ and $L_1\zeta L_2$.
 \end{proposition}

 \begin{proof} Since both topologies on the same space have the same bounded sets (proposition \ref{etageneral}.(6)), it suffices to consider $L_1\eta L_2$. 
 Consider a Mackey-Cauchy sequence $(x_n)_{n\geq 0}$, thus topologically Cauchy. By completeness of the scalar field, $x_n$ converges pointwise to a multilinear form $x$ on  $\prod_{i=1}^2 (L_i)'_\mu$. Since the topology of the $\eta$-product is the topology of uniform convergence on products of equicontinuous parts (which can be assumed absolutely convex and weakly compact), $x_n\to x$ uniformly on these products (since $(x_n)$ Cauchy in the Banach space of continuous functions on these products).
 From proposition \ref{etageneral}.(5) we only have to check that the limit $x$ is separately continuous. For each $y\in (L_2)'$, and $B$ a bounded set in $L_1\eta L_2=L((L_2)'_\mu,L_1)$, one deduces $B(y)$ is bounded in $((L_1)'_\mu)'=L_1$ with its original topology of convergence on equicontinuous sets of $L'_1$. 
 Therefore, $(x_n(y))$ is Mackey-Cauchy in $L_1$, thus Mackey-converges, necessarily to $x(y)$. Therefore $x(y)$ defines an element of $((L_1)'_\mu)'$. With the similar symmetric argument, $x$ is thus separately continuous, as expected. We have thus obtained the topological convergence of $x_n$ to $x$ in $L_1\eta L_2$. It is easy to see $x_n$ Mackey converges to $x$ in $L_1\eta L_2$ in taking the closure of the bounded set from its property of being Mackey-Cauchy. Indeed, the established topological limit $x_n\to x$ transfers the Mackey-Cauchy property in Mackey convergence as soon as the bounded set used in Mackey convergence is closed.
  \end{proof}

We will need the relation of Mackey duals and Mackey completions:
 
 \begin{lemma}\label{MMc}
For any separated \lcs $F$, we have a topological isomorphism $\widehat{((F'_\mu)'_\mu)}^M\simeq ((\widehat{F}^M)'_\mu)'_\mu.$ 
 \end{lemma}
 \begin{proof}

 Recall also from \cite[\S 21.4.(5)]{Kothe} the completion of the Mackey topology has its Mackey topology $\widetilde{((F'_\mu)'_\mu)}=((\widetilde{F})'_\mu)'_\mu$
therefore an absolutely convex weakly compact set in $F'$  coincide for the weak topologies induced by $F$ and $\widetilde{F}$ and therefore also $\hat{F}^M$, which is in between them. Thus the continuous inclusions $((F'_\mu)'_\mu)\to (\hat{F}^M)'_\mu)'_\mu\to ((\widetilde{F})'_\mu)'_\mu$ have always the induced topology. In the transfinite description of the Mackey completion, the Cauchy sequences and the closures are the same in $((\widetilde{F})'_\mu)'_\mu$ and $\widetilde{F}$ (since they have same dual hence same bounded sets), therefore  one finds the stated topological isomorphism.
 \end{proof}

{ \begin{lemma}\label{DualArensMc}
 If $L,M$ are separated locally convex spaces we have embeddings:
  $$ L'_\mu\o_{\beta e} M'_\mu\to (L\zeta M)'_\epsilon\to (L\eta M)'_\epsilon\to L'_\mu\widetilde{\o}_{\beta e} M'_\mu,$$
 with the middle duals coming with their $\epsilon$-topology as biduals of $L'_\mu\o_{\beta e} M'_\mu$. The same holds for : $$ L'_\mu\o_{i} M'_\mu\to (L\eta M)'_\mu\to L'_\mu\widetilde{\o}_{i} M'_\mu.$$
 Finally, $(L\eta M)'_\epsilon\subset L'_\mu\widehat{\o}_{\beta e}^M M'_\mu$ as soon as either $L$ or $M$ is a Schwartz space, and in any case we have $(L\zeta M)'_\epsilon\subset L'_\mu\widehat{\o}_{\beta e}^M M'_\mu.$
 \end{lemma}
 \begin{proof}
\setcounter{Step}{0}
\begin{step}
First line of embeddings.
\end{step}
From the identity continuous map $L\eta M\to L\zeta M$, there is an injective linear map $(L\zeta M)'\to (L\eta M)'$. Note that, on $L'_\mu\o M'_\mu$, one can consider the strongest topology weaker than $(\mathscr{S}(L))'_\mu\o_{\beta e} M'_\mu$
and $L'_\mu\o_{\beta e} (\mathscr{S}(M))'_\mu$. Let us call it $L'_\mu\o_{\zeta} M'_\mu$ and see it is topologically equal to $L'_\mu\o_{\beta e} M'_\mu$ by checking its universal property. We know by definition the map $L'_\mu\o_{\zeta} M'_\mu\to L'_\mu\o_{\beta e} M'_\mu$. Conversely, there is an $\varepsilon$-equihypocontinuous map $(\mathscr{S}(L))'_\mu\times M'_\mu\to L'_\mu\o_{\zeta} M'_\mu$
so that for every equicontinuous set in $M'$, the corresponding family is equicontinuous $(\mathscr{S}(L))'_\mu=L'_\mu\to L'_\mu\o_{\zeta} M'_\mu$ from the topological equality. Similarly, by symmetry, one gets for every equicontinuous set in $L'_\mu$, an equicontinuous family of maps $M'_\mu\to L'_\mu\o_{\zeta} M'_\mu$.
As a consequence, the universal property gives the expected map $L'_\mu\o_{\beta e} M'_\mu\to L'_\mu\o_{\zeta} M'_\mu$ concluding to equality. As a consequence, since by definition $(L'_\mu\o_{\zeta} M'_\mu)'=L\zeta M$ is the dual kernel for the hull defining the $\o_{\zeta}$ tensor product, one gets that an equicontinuous set in the kernel is exactly an equicontinuous set in $(L'_\mu\o_{\zeta} M'_\mu)'=(L'_\mu\o_{\beta e} M'_\mu)'$ namely an $\epsilon$-equihypocontinuous family. This gives the continuity of our map $(L\zeta M)_\epsilon'\to (L\eta M)_\epsilon'$ and even the embedding property (if we see the first as bidual of $L'_\mu\o_{\zeta} M'_\mu$ but we only stated an obvious embedding in the statement).

We deduce that $L'_\mu\o_{\beta e} M'_\mu\simeq L'_\mu\o_{\zeta} M'_\mu\to (L\zeta M)_\epsilon'$ is an embedding from \cite[\S 21.3.(2)]{Kothe} which proves that the original topology on a space is the topology of uniform convergence on equicontinuous sets.


We then build a continuous linear injection $(L\eta M)'_\epsilon\to L'_\mu\widetilde{\o}_{\beta e} M'_\mu$ to the full completion.
Since both spaces have the same dual, it suffices to show that  the topology on $L\eta M$ is stronger than Grothendieck's topology $\mathfrak{I}^{lf}(L'_\mu\o_{\beta e} M'_\mu)$ following \cite{Kothe} in notation. Indeed, let $C$ in $L\eta M$ equicontinuous. Assume a net  in $C$ converges pointwise $x_n\to x\in C$ in the sense $x_n(a,b)\to x(a,b), a\in  L'_\mu, b\in  M'_\mu$. For equicontinuous sets $A\subset L'_\mu, B\subset  M'_\mu$ which we can assume absolutely convex weakly compact, it is easy to see $C$ is equicontinuous on products $A\times B$. Thus it is an equicontinuous bounded family in $C^0(A\times B)$ thus relatively compact by Arzela-Ascoli Theorem \cite[3 \S 9 p237]{Horvath}. Thus since any uniformly converging subnet converges to $x$, the original net must converge uniformly on $A\times B$ to $x$. As a consequence the weak topology on $C$ coincides with the topology of $L\eta M$, and by definition we have a continuous identity map, $(L\eta M,\mathcal{I}^{lf}(L'_c\o_{\beta e} M'_c))\to L\eta M$. By Grothendieck's construction of the completion, the dual of the first space is the completion and this gives the expected injection between duals. Since a space and its completion induce the same equicontinuous sets, one deduces the continuity and induced topology property with value in the full completion.

\begin{step}
Second  line of embeddings.
\end{step}

It suffices to apply $((.)'_\mu)'_\mu$ to the first line. We identified the first space in proposition \ref{etageneral}.(5) and the last space as the completion of the first (hence of the second and this gives the induced topologies) in the proof of lemma \ref{MMc}.

\begin{step}
Reduction of computation of Mackey completion to the Schwartz case.
\end{step}
It remains to see the $(L\zeta M)_\epsilon'$ is actually valued in the Mackey completion.

Note that as a space, dual of a projective kernel, $(L\zeta M)_\epsilon'$ is the inductive hull of the maps 
$A=(\mathscr{S}(L)\varepsilon M)_\epsilon'\to(L\zeta M)_\epsilon'=C$ and $B=(L\varepsilon \mathscr{S}(M))_\epsilon'\to(L\zeta M)_\epsilon'=C.$
Therefore, it suffices to check that the algebraic tensor product is Mackey-dense in both these spaces $A,B$ that span $C$ since the image of a bounded set in $A,B$ being bounded in $C$, there are less Mackey-converging sequences in $A,B$. This reduces the question to the case $L$ or $M$  a Schwartz space. By symmetry, we can assume $L$ is.

\begin{step}
Description of the dual $(L\zeta M)'=(L\varepsilon M)'$ for $L$ Schwartz and conclusion.
\end{step}
 We take inspiration from the classical description of the dual of the injective tensor product as integral bilinear maps (see \cite[\S 45.4]{Kothe2}). As in \cite[Prop 6]{Schwartz}, we know any equicontinuous set (especially any point) in $(L\varepsilon M)'$ is included in the absolutely convex weakly closed hull $\Gamma$ of $A \o B$ with $A$ equicontinuous in $L'$, $B$ in $M'$. Since the dual of $L'_\mu\widehat{\o}_{\beta e}^M M'_\mu$ is the same, this weakly closed hull can be computed in this space too. Moreover, since $L$ is a Schwartz space, we can and do assume that $A=\{x_n, n\in \N\}$ is a $\epsilon$-Mackey-null sequence in $L'_\mu$, since they generate the equicontinuous bornology as a saturated bornology. We can also assume $A,B$ are weakly compact and $B$ absolutely convex.

Any element $f\in L\varepsilon M$ defines a continuous map on $A\times B$ (see e.g. \cite[Prop 2]{Schwartz} and following remark). We equip $A\times B$ with the above weakly compact topology to see $f|_{A\times B}\in C^0(A\times B)$. For $\mu$ a (complex) measure on $A\times B$ (i.e. $\mu \in (C^0(A\times B))'$, we use measures in the Bourbaki's sense, which define usual Radon measures \cite{SchwartzMeasure}) with norm $||\mu||\leq 1$ so that $\int_{A\times B}f(z)d\mu(z)=\mu(f|_{A\times B})=:w_\mu(f)$ make sense. 

Note that $|w_\mu(f)|\leq ||f||_{C^0(A\times B)}$ which is a seminorm of the $\epsilon$-product, so that $\mu$ defines a continuous linear map $w_\mu\in (L\varepsilon M)'.$ Note also that if $f$ is in the polar of $A\o B$, so that $|w_\mu(f)|\leq 1$ and thus by the bipolar theorem, $w_\mu\in \Gamma$. We want to check the converse that any element of $w\in\Gamma$ comes from such a measure. But if $H$ is the subspace of $C^0(A\times B)$ made of restrictions of functions $f\in L\varepsilon M$, $w$ induces a continuous linear map on $H$ with $|w(f)|\leq ||f||_{C^0(A\times B)}$, Hahn-Banach theorem enables to extend it to a measure $w_\mu$, $||\mu||\leq 1$. This concludes to the converse. 

Define the measure $\mu_n$ by $\int_{A\times B}f(z)d\mu_n(z)=\frac{1}{\mu(1_{\{x_n\}\times B})}\int_{A\times B}f(z)1_{\{x_n\}\times B}(z)d\mu(z)$ using its canonical extension to semicontinuous functions. Note that by Lebesgue theorem (dominated by constants) $$w_\mu(f)=\sum_{n=0}^\infty \mu(1_{\{x_n\}\times B}))w_{\mu_n}(f)$$

As above one sees that the restriction of $w_{\mu_n}$ to $f\in L\varepsilon M$ belongs to the weakly closed absolute convex hull of $\{x_n\}\times B.$ Thus since $B$ absolutely convex closed $w_{\mu_n}(f)=f(x_n\otimes y_n)$  for some $y_n\in B$. We thus deduces that any $w_\mu\in \Gamma$ has the form :$w_\mu(f)=\sum_{n=0}^\infty \mu(1_{\{x_n\}\times B}))f(x_n\otimes y_n).$
Since the above convergence holds for any $f$, this means 
the convergence in the weak topology :\begin{equation}\label{seriesrepresentation}w_\mu= \sum_{n=0}^\infty \mu(1_{\{x_n\}\times B})) x_n\otimes y_n.\end{equation}
Let $D$ the equicontinuous closed disk  such that $x_n$ tends to $0$ in $(L')_D$. Consider the closed absolutely convex cover $\Lambda=\overline{\Gamma(D\o B)}$. 
The closed absolutely convex cover can be computed in $(L\varepsilon M)'_\epsilon$ 
 or $(L\varepsilon M)'_\sigma$, both spaces having same dual \cite[\S 20.7.(6) and 8.(5)]{Kothe}, and $D\o B$ being equicontinuous \cite[Corol 4 p 27, Rmq p 28]{Schwartz}, so is $\Lambda$ \cite[\S 21.3.(2)]{Kothe}  hence it is weakly compact by Mackey Theorem, so complete in $(L\varepsilon M)'_\epsilon$ \cite[IV.5 Rmq 2]{Bourbaki}, so that $\Lambda$ is therefore a Banach disk there. But $||x_n\o y_n||_{(L\varepsilon M)'_\Lambda}\leq 1$
 so that since $\sum_{n=0}^\infty |\mu(1_{\{x_n\}\times B}))|\leq 1$ the above series is summable in $(L\varepsilon M)'_\Lambda$ and thus Mackey converges in $(L\varepsilon M)'_\epsilon$. As a conclusion, $\Gamma\subset  L'_\mu\widehat{\o}_{\beta e}^M M'_\mu$ and this gives the final statement.
 \end{proof}
The above proof has actually the following interesting consequence :


\begin{corollary}\label{SchwartzCondExpEpsilon}
For any $E,F$ separated locally convex spaces, we have the topological isomorphism:
$$ \mathscr{S}\Big([\mathscr{S}(E)]\varepsilon F\Big)=[\mathscr{S}(E)]\varepsilon [\mathscr{S}(F)].$$
\end{corollary}

 \begin{proof}
We have the canonical continuous map $[\mathscr{S}(E)]\varepsilon F\to [\mathscr{S}(E)]\varepsilon [\mathscr{S}(F)],$ hence since the $\varepsilon$-product of Schwartz spaces is Schwartz (see below proposition \ref{Sepsilon}), one gets by functoriality the first continuous linear map:
\begin{equation}\label{SchwartzCondExpEpsilonTrivial}\mathscr{S}\Big([\mathscr{S}(E)]\varepsilon F\Big)\to [\mathscr{S}(E)]\varepsilon [\mathscr{S}(F)].\end{equation}

Note that we have the algebraic equality $[\mathscr{S}(E)]\varepsilon F=L(E'_\mu,F)= L(E'_\mu,\mathscr{S}(F))= [\mathscr{S}(E)]\varepsilon [\mathscr{S}(F)]$
where the crucial middle equality comes from the map (see \cite[Corol 8.6.5]{Jarchow}) $$L(E'_\mu,F)=L(E'_\mu,(F'_\mu)'_\mu)=L(E'_\mu,([\mathscr{S}(F)]'_\mu)'_\mu)= L(E'_\mu,\mathscr{S}(F)).$$ 

To prove the topological equality, we have to check the duals are the same with the same equicontinuous sets. We can apply the proof of the previous lemma (and we reuse the notation there) with $L=\mathscr{S}(E)$, $M=F$ or $M=[\mathscr{S}(F)]$. First the space in which the Mackey duals are included $L'_\mu\widetilde{\o}_{i} M'_\mu$ is the same in both cases, and the duals are described as union of absolutely convex covers, it suffices to see those unions are the same to identify the duals. Of course, the transpose of \eqref{SchwartzCondExpEpsilonTrivial} gives $\Big([\mathscr{S}(E)]\varepsilon [\mathscr{S}(F)]\Big)'\subset \Big([\mathscr{S}(E)]\varepsilon  F\Big)'$ so that we have to show the converse.  From \eqref{seriesrepresentation} and rewriting $x_n\o y_n$ as $\frac{1}{\lambda_n}x_n\o \lambda_n y_n$ with $\lambda_n=\sqrt{||x_n||_{L'_C}}$, one gets that both sequences $x_n'=(\frac{1}{\lambda_n}x_n), y_n'(\mu)=(\lambda_n y_n)$ are null sequences for the equicontinuous bornology of $E',F'$ and therefore included in equicontinuous sets for the duals of  associated Schwartz spaces.  This representation therefore gives the equality of duals. Finally, to identify equicontinuous sets, in the only direction not implied by \eqref{SchwartzCondExpEpsilonTrivial}, we must see that an $\varepsilon$-null sequence $w_{\nu_n}$ of linear forms in the dual is included in the closed absolutely convex cover of  a tensor product of two such sequences in $[\mathscr{S}(E)]',[\mathscr{S}(F)]'$.  From the null convergence, $\nu_n$ can be taken measures on the same $A\times B$, for each $\nu_n$, we have a representation $w_{\nu_n}=\sum z_m(\nu_n) x_m'\o y_m'(\nu_n)$ where $x_m'$ is a fixed sequence and $(y_m'(\nu_n))_m$ are null sequences in the same Banach space $M'_B$. Moreover $\sum |z_m(\nu_n)|\leq ||\nu_n||\to 0$ from the assumption that $\nu_n$ is a null sequence in the Banach space generated by $\Gamma$ (we can assume $||\nu_n||\neq 0$ otherwise $w_{\nu_n}=0$). Therefore, we rewrite, the series as $w_{\nu_n}=\sum \frac{1}{||\nu_n||} z_m(\nu_n) x_m'\o y_m'(\nu_n)||\nu_n||$ and we gather all the sequence $(y_m'(\nu_n)||\nu_n||)_m$ into a huge sequence converging to $0$ in $M'_B$ which generates the equicontinuous set $B'$ of $(\mathscr{S}(M))'$ we wanted. $(x_m')$ generates another such equicontinuous set $A'$. This concludes to $w_{\nu_n}\in \overline{\Gamma(A'\o B')}$ so that the equicontinuous set generated by our sequence $(w_{\nu_n})$ must be in this equicontinuous set for $([\mathscr{S}(E)]\varepsilon [\mathscr{S}(F)])'.$
 \end{proof}
 
We are ready to obtain the associativity of the $\zeta$ tensor product:

 \begin{proposition}\label{AssocMc}
Let $L_1,L_2,L_3$ be \lcs with $L_3$  Mackey-complete, then there is a continuous linear map $$Ass:L_1\zeta (L_2\zeta L_3)\to(L_1\zeta L_2)\zeta L_3.$$
If also $L_1$ is  Mackey-complete, this is a topological isomorphism.
\end{proposition}
\begin{proof}  

First note that we have the inclusion $$L_1\zeta (L_2\zeta L_3)\subset L((L_1)'_\mu,L_\sigma((L_2)'_\mu,L_3))=L((L_1)'_\mu\o_i(L_2)'_\mu,L_3).$$
Since $L_3$ is Mackey-complete, such a map extends uniquely to the Mackey completion $L((L_1)'_\mu\widehat{\o}^M_i(L_2)'_\mu,L_3)$ and since lemma \ref{DualArensMc} gives $(L_1\eta L_2)'_\mu$ as a subspace, we can restrict the unique extension and get our expected linear map:
$$i:L_1\zeta (L_2\zeta L_3)\to L((L_1\eta L_2)'_\mu,L_3)=(L_1\eta L_2)\eta L_3.$$

It remains to check continuity. Since the right hand side is defined as a topological kernel, we must check continuity after applying several maps. Composing with $$J_1:(L_1\eta L_2)\eta L_3\to (L_1\varepsilon \mathscr{S}(L_2))\varepsilon \mathscr{S}(L_3),$$ one gets a map $J_1\circ i$ which is continuous since it coincides with 
the composition of the map obtained from corollary \ref{SchwartzCondExpEpsilon}: $$I_1:L_1\zeta (L_2\zeta L_3)\to L_1\varepsilon \mathscr{S}(\mathscr{S}(L_2)\varepsilon L_3)=L_1\varepsilon (\mathscr{S}(L_2)\varepsilon \mathscr{S}(L_3))$$
with a variant $i':L_1\varepsilon (\mathscr{S}(L_2)\varepsilon \mathscr{S}(L_3))\to (L_1\varepsilon \mathscr{S}(L_2))\varepsilon \mathscr{S}(L_3)$ of $i$ in the Schwartz case, so that $i'\circ I_1=i\circ J_1.$ And $i'$ is continuous since the equicontinuous set in their duals are generated by tensor products of equicontinuous sets for the base spaces (easy part in the corresponding associativity in \cite{Schwartz}). 
The case of composition with $J_2:(L_1\eta L_2)\eta L_3\to (\mathscr{S}(L_1)\varepsilon L_2)\varepsilon \mathscr{S}(L_3)$ is similar and easier.

The last two compositions are gathered in one using corollary \ref{SchwartzCondExpEpsilon} again. We have to compose with the map $$J_3:(L_1\eta L_2)\eta L_3\to \Big(\mathscr{S}\big(\mathscr{S}(L_1)\varepsilon L_2\big)\Big)\varepsilon L_3=\Big((\mathscr{S}(L_1))\varepsilon (\mathscr{S}(L_2))\Big)\varepsilon L_3=\Big(\mathscr{S}\big(L_1\varepsilon (\mathscr{S}(L_2))\big)\Big)\varepsilon L_3.$$

Again we use the canonical continuous factorization via $L_1\zeta (L_2\zeta L_3)\to(\mathscr{S}(L_1))\varepsilon \Big((\mathscr{S}(L_2))\varepsilon L_3\Big)$ and use the same argument as before between $\varepsilon$-products.
\end{proof}
}

We can now summarize the categorical result obtained, which gives a negative connective, hence an interpretation of $\parr$.

\begin{theorem}
\label{zetaparr}
The full subcategory $\Mc\subset \LCS$ of Mackey-complete spaces is a reflective subcategory with reflector (i.e. left adjoint to inclusion) the Mackey completion  $\ \widehat{\cdot}^M$. It is complete and cocomplete and symmetric monoidal with product $\zeta$ which commutes with limits.
\end{theorem}
\begin{proof}
The left adjoint relation $\Mc(\widehat{E}^M,F)=\LCS(E,F)$ is obvious by restriction to $E\subset \widehat{E}^M$ and functoriality of $\ \widehat{\cdot}^M$ \cite[Prop 5.1.25]{PerrezCarreras}. As usual, naturality is easy. 
As a consequence, limits in $\Mc$ are those of $\LCS$ and colimits are the Mackey-completed colimits. The unit for $\zeta$ is of course $\K$. The associator has been built in Proposition \ref{AssocMc}. With $E\zeta F=L(E'_\mu,F)$, we saw  the braiding  is the transpose map, left unit $\lambda_F$ is identity and right unit is identification $\rho_E:(E'_\mu)'_\epsilon\simeq E.$ Taking the Mackey-dual of expected maps in relations (pentagon, triangle and hexagon identities) one gets the transposed relations, which restrict to the known relations for $(\LCS,\o_i)$ as symmetric monoidal category. By Mackey-density obtained in proposition \ref{DualArensMc}, the relations extend to the expected relations for the transpose maps. Hence, transposing again (i.e. applying functor $(\cdot)'_\epsilon$ from dual spaces with linear maps preserving equicontinuous sets to $\LCS$) imply the expected relations. We already saw in lemma  \ref{etageneral} the commutation of limits with $\zeta$.
\end{proof}

\section{Original setting for the Schwartz $\varepsilon$-product and smooth maps.}\label{sec:kref}
In his original paper \cite{Schwartz}, Schwartz used quasi-completeness as his basic assumption to ensure associativity, instead of restricting to Schwartz spaces and assuming only Mackey-completeness as we will do soon inspired by section 3. Actually, what is really needed is that the absolutely convex cover of a compact set is still compact. Indeed, as soon as one takes the image (even of an absolutely convex) compact set by a continuous bilinear map, one gets only what we 
know from continuity, namely compactness and the need to recover absolutely convex sets, for compatibility with the vector space structure, thus makes the above assumption natural. 
Since this notion is related to compactness and continuity, we call it $k$-quasi-completeness.

This small remark reveals this notion is also relevant for differentiability since it is necessarily based on some notion of continuity, at least at some level, even if this is only on $\R^n$ as in convenient smoothness. Avoiding the technical use of Schwartz spaces for now and benefiting from \cite{Schwartz}, we find a $*$-autonomous category and an adapted notion of smooth maps.

We will see this will give us a strong notion of differentiability with Cartesian closedness. We will come back to convenient smoothness in the next sections starting from what we will learn in this basic example with a stronger notion of smoothness.

\subsection{$*$-autonomous category of $k$-reflexive spaces.}

 \begin{definition}
 \label{k-complete}
 A (separated) locally convex space $E$ is said to be $k$-quasi-complete, if for any compact set $K\subset E$, its closed absolutely convex cover $\overline{\Gamma (K)}$ is complete (equivalently compact \cite[\S 20.6.(3)]{Kothe}). We denote by $\Kc$ the category of $k$-quasi-complete spaces and linear continuous maps.
 \end{definition}
 \begin{remark}\label{exValdivia}
There is a $k$-quasi-complete space which is not quasi-complete, hence our new notion of $k$-quasi-completeness does not reduce to the usual notion. Indeed in \cite{Valdivia},  is built a completely regular topological space $W$ such that $C^0(W)$ with compact-open topology is bornological and such that it is an hyperplane in its completion, which is not bornological. 
If $C^0(W)$ were quasi-complete, it would be complete by  \cite[Corol 3.6.5]{Jarchow} and this is not the case. $C^0(W)$ is $k$-quasi-complete since by Ascoli Theorem twice \cite[X.17 Thm 2]{BourbakiTG} a compact set for the compact open topology is pointwise bounded and equicontinuous, hence so is the absolutely closed convex cover of such a set, which is thus compact too. 
\end{remark}
 
 The following result is similar to lemma \ref{gammacompletion} and left to the reader.

\begin{lemma}\label{deltacompletion}
The intersection $\widehat{E}^K$ of all $k$-quasi-complete spaces containing $E$ and contained in the completion $\tilde{E}$ of $E$, is $k$-quasi-complete and called the $k$-quasi-completion of $E$.

We define $E_0=E$, and for any ordinal $\lambda$, the subspace $E_{\lambda +1}=\cup_{K\in C(E_{\lambda })}\overline{\Gamma (K)} \subset \tilde{E}$ where the union runs over all compact subsets $ C(E_{\lambda })$ of $E_{\lambda }$ with the induced topology, and the closure is taken in the completion. We also let for any limit ordinal $E_{\lambda}=\cup_{\mu<\lambda}E_\mu$. Then for any ordinal $\lambda$, $E_{\lambda }\subset\widehat{E}^K$ and eventually for $\lambda$ large enough, we have equality.
\end{lemma}



\begin{definition}
 For a (separated) locally convex space $E$, the topology $k(E',E)$ on $E'$ is the topology of uniform convergence on absolutely convex compact sets of $\widehat{E}^K$. The dual $(E',k(E',E))=(\widehat{E}^K)'_c$ is nothing but the Arens dual of the $k$-quasi-completion and is written $E'_k.$ We let $E^*_k=\widehat{E'_k}^K.$
 A (separated) locally convex space $E$ is said $k$-reflexive if $E$ is $k$-quasi-complete and if $E=(E'_k)'_k$ topologically. Their category is written $\kref$. 
 \end{definition}
 From Mackey theorem, we know that $(E'_k)'=(E^*_k)'=\widehat{E}^k.$

We first want to check that $\kref$ is logically relevant in showing that $(E'_k)'_k$ and $E^*_k$ are always in it. Hence we will get a $k$-reflexivization functor. This is the first extension of the relation $E'_c=((E'_c)'_c)'_c$ that we need. 

We start by proving a general lemma we will reuse several times. Of course to get a $*$-autonomous category, we will need some stability of our notions of completion by dual. The following lemma says that if a completion can be decomposed by an increasing ordinal decomposition as above and that for each step the duality we consider is sufficiently compatible in terms of its equicontinuous sets, then the process of completion in the dual does not alter any kind of completeness in the original space.

\begin{lemma}
\label{ordinalCompletion}
Let $D$ a contravariant duality functor on $\mathbf{LCS}$, meaning that algebraically $D(E)=E'$. We assume it is  compatible with duality ($(D(E))'=E$). Let $E_0\subset E_\lambda\subset \widetilde{E_0}$ an increasing family of subspaces of the completion $\widetilde{E_0}$ indexed by ordinals $\lambda\leq \lambda_0$. We assume that for limit ordinals $E_\lambda=\cup_{\mu<\lambda}E_\mu$ and, at successor ordinals that every point $x\in E_{\lambda+1}$ lies in $\overline{\Gamma(L)},$ for a set $L\subset E_{\lambda}$, equicontinuous in $[D(E_{\lambda_0})]'$. 

Then any complete  set $K$ in $D(E_0)$ is also complete for the stronger topology of $D(E_{\lambda_0})$.
\end{lemma} 
 \begin{proof} Let $E=E_0$.
 Note that since $D(E)=D(\widetilde{E})$ we have $D(E)=D(E_{\lambda})$ algebraically.
 
 Take a net $x_n\in K$ which is a Cauchy net in  $D(E_{\lambda_0})$. Thus $x_n\to x\in K$ in $D(E_0).$ We show by transfinite induction on $\lambda$ that $x_n\to x$ in  $D(E_{\lambda})$. 
 
 First take $\lambda$ limit ordinal.
  The continuous embeddings $E_\mu\to E_\lambda$ gives by functoriality a continuous identity map 
$D(E_{\lambda})\to D(E_{\mu})$ for any $\mu<\lambda$.
Therefore since we know  $x_n\to x$ in any $D(E_{\mu})$ the convergence takes place in the projective limit $D_\lambda=\text{proj} \lim_{\mu<\lambda} D(E_{\mu})$.

But we have a continuous identity map $D(E_{\lambda})\to D_\lambda$ and both spaces have the same dual $E_\lambda=\cup_{\mu<\lambda}E_\mu.$ 
For any equicontinuous set $L$ in $(D(E_{\lambda}))'$ $x_n$ is Cauchy thus converges uniformly in $C^0(L)$ on the Banach space of weakly continuous maps. It moreover converges pointwise to $x$, thus we have uniform convergence to $x$ on any equicontinuous set i.e. $x_n\to x$ in $D(E_{\lambda})$.


 Let us prove convergence in $D(E_{\lambda+1})$ at successor step assuming it in $D(E_{\lambda})$.  
Take an absolutely convex closed equicontinuous set $L$ in $(D(E_{\lambda+1}))'=E_{\lambda+1}$, we have to show uniform convergence on any such equicontinuous set.
 Since $L$ is weakly compact, one can look at the Banach space of weakly continuous functions $C^0(L).$ Let $\iota_L:D(E_{\lambda+1})\to C^0(L).$ $\iota_L(x_n)$ is Cauchy by assumption and therefore converges uniformly to some $y_L$. We want to show $y_L(z)=\iota_L(x)(z)$ for any $z\in L$. Since $z\in  E_{\lambda+1}$ there is by assumption a set $M\subset E_{\lambda}$  equicontinuous in $[D(E_{\lambda_0})]'$ such that $z\in \overline{\Gamma(M)}$ computed in $E_{\lambda+1}.$ Let $N=\overline{\Gamma(M)}$ computed in $E_{\lambda_0}$, so that $z\in N$. Since $M$ is equicontinuous  in $(D(E_{\lambda_0}))'$ we conclude that so is $N$ and it is also weakly compact there. One can apply the previous reasoning to $N$ instead of $L$ (since $x_n$ Cauchy in $D(E_{\lambda_0})$, not only in $D(E_{\lambda +1})$).  $\iota_N(x_n)\to y_N$ and since $z\in L\cap N$ and using pointwise convergence $y_L(z)=y_N(z)$. Note also $\iota_N(x)(z)=\iota_L(x)(z)$. Moreover, for $m\in M\subset E_{\lambda}$, $\iota_N(x_n)(m)\to \iota_N(x)(m)$ since $\{m \}$ is always equicontinuous in $(D(E_{\lambda}))'$ so that $\iota_N(x)(m)=y_N(m)$. Since both sides are affine on the convex $N$ and weakly continuous (for $\iota_N(x)$ since $x\in D(E_{\lambda_0})=E_{\lambda_0}'$), we extend the relation to any $m\in N$ and thus $\iota_N(x)(z)=y_N(z)$. Altogether, this gives the expected $y_L(z)=\iota_L(x)(z)$. Thus $K$ is complete as expected.
 \end{proof}

\begin{lemma}\label{kduality}
For any separated locally convex space, $E^*_k=((E'_k)'_k)'_k$ is $k$-reflexive. A space is $k$-reflexive if and only if $E=(E'_c)'_c$ and both $E$ and $E'_c$ are $k$-quasi-complete. More generally, if $E$ is $k$-quasi-complete, so are $(E'_k)'_c=(E'_c)'_c$ and $(E^*_k)'_c$ and  $\gamma(E)=\gamma((E'_c)'_c)=\gamma((E^*_k)^*_k)$.
\end{lemma} 
\begin{remark}\label{exValdivia2}The example $E=C^0(W)$ in Remark \ref{exValdivia},
which is not quasi-complete, is even $k$-reflexive.
Indeed it remains to see that $E'_c$ is $k$-quasi-complete. But from \cite[Thm 13.6.1]{Jarchow}, it is not only bornological but ultrabornological, hence by \cite[Corol 13.2.6]{Jarchow}, $E'_\mu$ is complete (and so is $F=\mathscr{S}(E'_\mu)$. But for a compact set in $E'_c$, the closed absolutely convex cover is closed in $E'_c$, hence $E'_\mu$, hence complete there. Thus, by Krein's Theorem \cite[\S 24.5.(4)]{Kothe}, it is compact in $E'_c$, making $E'_c$ $k$-quasi-complete. 
\end{remark}
\begin{proof}
One can assume $E$ is $k$-quasi-complete (all functors start by this completion) thus so is $(E'_c)'_c$ by \cite[IV.5 Rmq 2]{Bourbaki} since $(E'_c)'_c\to E$ continuous with same dual (see \cite{Schwartz}). There is a continuous map $(E^*_k)'_c\to (E'_c)'_c$ we apply lemma \ref{ordinalCompletion} to $E_0=E'_c$, $E_\lambda$ the $\lambda$-th step of the completion in lemma \ref{deltacompletion}. Any $\overline{\Gamma(K)}$ in the union defining $E_{\lambda+1}$ is equicontinuous in $((E_{\lambda+1})'_c)'$ so a fortiori in $((E_{\lambda_0})'_c)'$ for $\lambda_0$ large enough. We apply the lemma to another $K$ closed absolutely convex cover of a compact set of $(E^*_k)'_c$  computed in $(E'_c)'_c$ therefore compact there by assumption. The lemma gives $K$ is complete there contains the bipolar of the compact computed in $(E^*_k)'_c$ which must also be compact as a closed subset of a compact. In this case we deduced $(E^*_k)'_c=(E'_k)'_k$ is  $k$-quasi-complete.

Clearly $((E'_k)'_k)'_k=((E^*_k)'_k)'_k\to E^*_k$ continuous. Dualizing the continuous $(E'_k)'_k\to E$ one gets $E'_k\to((E'_k)'_k)'_k=((E^*_k)'_k)'_k\to E^*_k$ and since the space in the middle is already $k$-quasi-complete inside the last which is the $k$-quasi-completion, it must be the last space and thus  $E^*_k$ $k$-reflexive and we have the stated equality.

For the next-to-last statement, sufficiency is clear, the already noted $(E'_k)'_k=(E^*_k)'_c\to (E'_c)'_c\to E$ in the $k$-quasi-complete case which implies $(E'_c)'_c\simeq E$ if $(E'_k)'_k\simeq E$ and $E^*_k=((E'_k)'_k)'_k=((E'_k)'_k)'_c=E'_c$ implies this space is also $k$-quasi-complete. For the comparison of absolutely convex compact sets, note that $(E^*_k)^*_k\to (E'_c)'_c$ ensures one implication and if $K\in \gamma( (E'_c)'_c)$   we know it is equicontinuous in $(E'_c)'$ hence \cite[\S 21.4.(5)]{Kothe} equicontinuous in $(\widehat{E'_c}^K)'$ and as a consequence included in an absolutely convex compact in  $(\widehat{E'_c}^K)'_c=(E^*_k)^*_k,$ i.e. $K\in \gamma((E^*_k)^*_k)$. $\gamma(E)=\gamma((E'_c)'_c)$ is a reformulation of $E'_c\simeq ((E'_c)'_c)'_c.$
\end{proof}

We consider $\gamma-\Kc$ the full subcategory of $\Kc$ with their $\gamma$-topology, and  $\gamma$-\textbf{Kb} the full subcategory of $\LCS$ made of spaces of the form $E'_c$ with $E$ $k$-quasi-complete.

We first summarize the results of \cite{Schwartz}. We call $\gamma-\LCS\subset \LCS$ the full subcategory of spaces having their $\gamma$-topology, namely $E=(E'_c)'_c$.  This is equivalent  to saying that subsets of absolutely convex compact sets in $E'_c$ are (or equivalently are exactly the) equicontinuous subsets in $E'$. With the notation of Theorem \ref{FirstMALL}, this can be reformulated by an intertwining relation in $\textbf{CLCS}$ which explains the usefulness  of these spaces :
\begin{equation}\label{gammaviaCLCS}  E\in \gamma-\LCS \Leftrightarrow (E'_c)_c=(E_c)'_b\Leftrightarrow (E_c)'_b=[U((E_c)'_b)]_c\end{equation}

\begin{proposition}
\label{lem:Kc_epsilon}
$k$-quasi-complete spaces are stable by $\varepsilon$-product, and $(\Kc,\varepsilon,\K)$ form a symmetric monoidal category . Moreover, if $E,F$ are $k$-quasi-complete, a set in $E\varepsilon F$ is relatively compact if and only if it is $\varepsilon$-equihypocontinuous. 
Therefore we have canonical embeddings: $$E'_c\o_{\beta e} F'_c\to (E\varepsilon F)'_c\to E'_c\widehat{\o}^K_{\beta e} F'_c.$$
\end{proposition}

\begin{proof}
The characterization of relatively compact sets is \cite[Prop 2 ]{Schwartz}, where it is noted that the direction proving relative compactness does not use any quasi-completeness. It gives $(E\varepsilon F)'_c=(E\varepsilon F)'_\epsilon$ with the epsilon topology as a bidual of $E'_c\o_{\beta e} F'_c$ and in general anyway a continuous linear map: \begin{equation}\label{Prop2SchwartzNonqComp}(E\varepsilon F)'_c\to(E\varepsilon F)'_\epsilon\end{equation}

 For a compact part in $E\varepsilon F$, hence equicontinuous in $(E'_c\o_{\beta e} F'_c)'$, its bipolar is still $\varepsilon$-equihypocontinuous hence compact by the characterization, as we have just explained. This gives stability of $k$-quasi-completeness.

  Associativity of $\varepsilon$ is Schwartz' Prop 7 but we give a reformulation giving a more detailed proof that $(\Kc,\varepsilon)$ is symmetric monoidal. The restriction to $\Kc$ of the functor $(\cdot)_c$ of Theorem \ref{FirstMALL} gives a functor we still call $(\cdot)_c:\Kc\to \textbf{CLCS}$. It has left adjoint $\ \widehat{\cdot}^K\circ U$. Note that for $E,F\in\Kc$,  $E\varepsilon F=\widehat{\cdot}^K\circ U(E_c\parr_b F_c)$ from our previous stability of $\Kc$. Moreover, note that \begin{equation}\label{ReductionKcCLCS}\forall E,F\in \Kc,\ \ \ \ \ (E\varepsilon F)_c=E_c\parr_b F_c\end{equation} thanks to the characterization of relatively compact sets, since the two spaces were already known to have same topology and the bornology on the right was defined as the equicontinuous bornology of $(E'_c\o_{\beta e} F'_c)'$ and on the left the one generated by absolutely convex compact sets or equivalently the saturated bornology generated by compact sets (using $E\varepsilon F\in \Kc$). Lemma \ref{CategoricStabMonoidal} concludes to $(\Kc,\varepsilon,\K)$ symmetric monoidal. They also make $(\cdot)_c$ a strong monoidal functor. 
  
We could deduce from \cite{Schwartz} the embeddings, but we prefer seeing them as coming from $\textbf{CLCS}$. 

Let us apply the next lemma to the embedding of our statement. Note that by definition $E'_c\o_{\beta e} F'_c=U((E_c)'_b\o_H (F_c)'_b)$, and $(E\varepsilon F)'_\varepsilon= U((E_c\parr_b F_c)'_b)=U((E_c)'_b\o_b (F_c)'_b)$ so that we got the embeddings for $E,F\in \LCS$:
\begin{equation}\label{InclusionSchwartz} E'_c\o_{\beta e} F'_c\to (E\varepsilon F)'_\varepsilon\to E'_c\widehat{\o}^K_{\beta e} F'_c
\end{equation}
which specializes to the statement in the $k$-quasi-complete case by the beginning of the proof to identify the middle terms.
\end{proof}
We have used and are going to reuse several times the following:

\begin{lemma}\label{Prop2SchwartzCLCS}
Let $E,F\in \CLCS$ (resp. with $ E,F'_b$ having moreover Schwartz bornologies)  we have the topological embedding (for $U$ the map giving the underlying \lcs):
\begin{equation}\label{InclusionSchwartzCLCS}\qquad\qquad \qquad\ \  U(E\o_H F'_b)\to U(E\o_b F'_b)\to [\ \widehat{\cdot}^K\circ U](E\o_H F'_b).\end{equation}
\begin{equation}\label{InclusionSchwartzCSch}(resp. \qquad\qquad \qquad U(E\o_H F'_b)\to U(E\o_b F'_b)\to [\ \widehat{\cdot}^M\circ U](E\o_H F'_b).\qquad)\end{equation}
\end{lemma}
\begin{proof}
Recall that for $E,F\in \textbf{CLCS}$, $E\o_H F$ has been defined before the proof of Theorem \ref{FirstMALL} and is the algebraic tensor product. Let us explain that, even before introducing the notion of $k$-quasi-completion, we already checked the result of the statement.
By construction we saw $(E\o_b F'_b)'_b=E'_b\parr_bF=L_b(E,F)=(E\o F'_b)'_b$, hence by $*$-autonomy $E\o_b F'_b=((E\o_H F'_b)'_b)'_b=(L_b(E,F))'_b$
and it has been described as a subspace $E\widehat{\o}_H F'_b$ inside the completion (in step 1 of this proof) with induced topology, obtained as union of bipolars of $A\o B$ or $\overline{A}\o \overline{B}$ (image of the product), for $A$ bounded in $E$, $B$ bounded in $F'_b$. Hence the embeddings follows from  the fact we checked $\overline{A}\o \overline{B}$ is precompact, and of course closed 
in the completion hence compact and the bipolar is one of those appearing in the first step of the inductive description of the $k$-quasicompletion.

For the case $E,F'_b$ having Schwartz bornologies, bounded sets are of the form $A\subset\overline{\Gamma(x_n,n\in \N)}, B\subset\overline{\Gamma(y_m,m\in \N)}$ with $(x_n),(y_m)$ Mackey-null in their respective bornologies. Take $C,D$ absolutely convex precompact sets bounded in the respective bornologies with $||x_n||_C\to 0, ||y_m||_D\to 0$, hence $||x_n\o y_m||_{(C\o D)^{oo}}\leq ||x_n||_C||y_m||_D$ and since we checked in the proof of Theorem \ref{FirstMALL} that  $(C\o D)^{oo}$ is precompact hence bounded, $x_n\o y_m$ can be gathered in a Mackey-null sequence has the one whose bipolar appears in the first term of the Mackey-completion.
\end{proof}
We have also used the elementary categorical lemma:
\begin{lemma}\label{CategoricStabMonoidal}
Let $(\mathcal{C},\o_\mathcal{C},I)$ a symmetric monoidal category and $\mathcal{D}$ a category. Consider a functor  $R:\mathcal{D}\to \mathcal{C}$ with left adjoint $L:\mathcal{C}\to \mathcal{D}$ and define $J=L(I)$, and $E\o_\mathcal{D} F=L(R(E)\o_\mathcal{C}R(F)).$ Assume that   for any $E,F\in \mathcal{D}$, $L(R(E))=E$, $R(J)=I$ and
$$R(E\o_\mathcal{D} F)=R(E)\o_\mathcal{C}R(F).$$
Then, $(\mathcal{D},\o_\mathcal{D},J)$ is a symmetric monoidal category.
\end{lemma}
\begin{proof}
The associator is obtained as $Ass^{\o_\mathcal{D}}_{E,F,G}=L( Ass^{\o_\mathcal{C}}_{R(E),R(F),R(G)})$ and the same intertwining defines the braiding and units and hence transports the relations which concludes. For instance in the pentagon we used the relation $L(Ass^{\o_\mathcal{C}}_{R(E),R(F)\o_\mathcal{C}R(G),R(H)})=Ass^{\o_\mathcal{D}}_{E,F\o_\mathcal{D}G,H}$.
\end{proof}

We deduce a description of internal hom-sets in these categories : we write  $L_{co}(E,F)$, the space of all continuous linear maps from $E$ to $F$ endowed with the topology of uniform convergence on compact subsets of $E$. When $E$ is a $k$-quasi-complete space, note this is the same \lcs as $L_c(E,F)$, endowed with the topology of uniform convergence on absolutely convex compacts of $E$.

\begin{corollary} 
\label{prop:defs_epsilon_coincide}
For $E \in \gammaKc$ and $F \in \Kc$ (resp. $F \in \Mc$), one has $L_c (E'_c, F) \simeq E \varepsilon F$, which is $k$-quasi-complete (resp. Mackey-complete).
\end{corollary}
 
 \begin{proof}
Algebraically
, $E \varepsilon F=L(E'_c,F)
$
 and the first space is endowed with the topology of uniform convergence on equicontinuous sets in $E'_c$ which coincides with subsets of absolutely convex compact sets since $E$ has its $\gamma$-topology. 
 \end{proof}
Using that for $E\in \gamma$-\textbf{Kb}, $E=F'_c$ for $F\in \Kc$, hence $E'_c=(F'_c)'_c\in\Kc$ by lemma \ref{kduality}.
\begin{corollary} 
\label{prop:k-compl-lin}
Consider $E\in\gamma$-\textbf{Kb} , $F \in \Kc$ (resp. $F \in \Mc$) then  $L_c(E,F)$ is $k$-quasi-complete (resp. Mackey-complete). 
\end{corollary}

\begin{proposition}
\label{prop:adj_Kc_gammaKc}
 $\gamma-\Kc\subset \Kc$ is a coreflective subcategory   with coreflector (right adjoint to inclusion)
$((\cdot)'_c)'_c$, 
which commutes with $\ \widehat{\cdot}^K$ on $\gamma$-\textbf{Kb}.    For $F\in \gamma-\Kc$,
$\cdot\ \widehat{\o}^K_\gamma F'_c: \LCS\to \Kc$ (resp. $\Kc\to \Kc,\gamma-\Kc\to \gamma-\Kc$) is left adjoint to $F\varepsilon\ \cdot$ (resp. $F\varepsilon\ \cdot,((F\varepsilon\ \cdot)'_c)'_c$). More generally, for $F\in \Kc$,$\cdot\ \widehat{\o}^K_{\gamma,\varepsilon} F'_c: \LCS\to \Kc$ 
 is left adjoint to $F\varepsilon\ \cdot.$ Finally, $\gamma$-\textbf{Kb} is stable by $\widehat{\o}^K_\gamma.$
\end{proposition}

 \begin{proof}
(1)  We start by proving the properties of the inclusion $\gamma-\Kc\subset \Kc.$
 Let $E\in\mathbf{Kc}$. We know the continuous map  $(E'_c)'_c\to E$ and  both spaces have the same dual, therefore for $K$ compact in $(E'_c)'_c$
 its closed absolutely convex cover is the same computed in both by the bipolar Thm \cite[\S 20.7.(6) and 8.(5)]{Kothe} and it is complete in $E$ by assumption so that by \cite[IV.5 Rmq 2]{Bourbaki} again also in $(E'_c)'_c$ which is thus $k$-quasi-complete too.
Hence, by functoriality of Arens dual, we got a functor: $((\cdot)'_c)'_c:\Kc\to \gamma-\Kc.$  Then we deduce from functoriality the continuous inverse maps $L(F,E)\to L((F'_c)'_c,(E'_c)'_c)=L(F,(E'_c)'_c)\to L(F,E)$ (for $F\in  \gamma-\Kc, E\in \Kc$) which gives the first adjunction. The unit is $\eta=id$ and counit given by the continuous identity maps: $\varepsilon_E:((E)'_c)'_c\to E$.

(2) Let us turn to proving the commutation property with completion. For $H\in \gamma$-\textbf{Kb}, $H=G'_c=((G'_c)'_c)'_c, G\in \Kc$ we thus have to note that the canonical map  $((\widehat{H}^K)'_c)'_c\to \widehat{H}^K$ is inverse of 
 the map obtained from canonical map $H\to \widehat{H}^K$ by applying functoriality: $H\to(\widehat{H}^K)'_c)'_c$ and then k-quasi-completion (since we saw the target is in $\gamma-\Kc$:) $\widehat{H}^K\to ((\widehat{H}^K)'_c)'_c.$

(3) For the adjunctions of tensor products, let us start with a heuristic computation. Fix $F\in \gamma-Kc, E\in \LCS,G\in \Kc.$
 From the discussion before \eqref{gammaviaCLCS}, $L_{\gamma}(F'_c, G)\simeq  F\varepsilon G$ thus, there is a canonical injection $$\Kc(E\widehat{\o}^K_\gamma F'_c,G)=L(E\o_\gamma F'_c,G)\to L(E,L_{\gamma}(F'_c, G))=L(E,F\varepsilon G).$$
But an element in  $L(E,F\varepsilon G)$
 sends a compact set in $E$ to a compact set in  $F\varepsilon G$ therefore an $\epsilon$-equihypocontinuous set by proposition \ref{lem:Kc_epsilon} which  is a fortiori an equicontinuous set in $L(F'_c,G)$. This gives the missing hypocontinuity to check the injection is onto. 
 
 Let us now  give a more abstract alternative proof of the first adjunction. Fix $F\in \gamma-\Kc.$ Let us define $\cdot\ \widehat{\o}^K_\gamma F'_c: \LCS\to \Kc$ as the composition $\ \widehat{\cdot}^M\circ U\circ (\cdot \o_b (F_c)'_b)\circ (.)_c
$ so that we will be able to describe the unique adjunction by composing known adjunctions. (Similarly, for $F\in \Kc$ one can define $\cdot\ \widehat{\o}^K_{\gamma,\varepsilon} F'_c: \LCS\to \Kc$ as the same composition $\ \widehat{\cdot}^M\circ U\circ (\cdot \o_b (F_c)'_b)\circ (.)_c
$).  We have to check this is possible by agreement on objects. This reads for $E\in \LCS$ 
as application of \eqref{InclusionSchwartzCLCS}
, \eqref{gammaviaCLCS} and reformulation of the definition $\cdot\o_\gamma\cdot=U((\cdot)_c\o_H(\cdot)_c):$
$$\widehat{\cdot}^K\circ U( E_c \o_b (F_c)'_b)=\widehat{\cdot}^K\circ U( E_c \o_H (F_c)'_b)=\widehat{\cdot}^K\circ U( E_c \o_H (F'_c)_c)=E\ \widehat{\o}^K_\gamma F'_c.$$
The case $F\in\Kc$ is similar since by definition $\cdot\o_{\gamma,\varepsilon}(\cdot)'_c=U((\cdot)_c\o_H((\cdot)_c)'_b)$.

 Then, to compute the adjunction, one needs to know the adjoints of the composed functors, which are from Theorem \ref{FirstMALL} and the proof of proposition \ref{lem:Kc_epsilon}. This gives as adjoint $U\circ  (\cdot \parr_b F_c)\circ (.)_c=\cdot\varepsilon F.$

(4)  The second adjunction is a consequence and so is the last  if we see $\cdot\ \widehat{\o}^K_\gamma F'_c: \gamma-\Kc\to \gamma-\Kc$ as composition of $i:\gamma-\Kc\to \Kc$, $\cdot\ \widehat{\o}^K_\gamma F'_c: \Kc\to \Kc$ and the right adjoint of $i$ (which we will see is not needed here). Indeed, by proposition \ref{lem:Kc_epsilon}, for $E\in \gamma-\Kc$,  $E'_c\widehat{\o}^K_\gamma F'_c=E'_c\widehat{\o}^K_{\beta e} F'_c$ is the $k$-quasi-completion of $(E\varepsilon F)'_c\in \gamma-\mathbf{Kb}$, and therefore from the commutation of $\gamma$-topology and $k$-quasi-completion in that case, that we have just established in (2), it is also in  $\gamma-\Kc.$ Hence, the adjunction follows by composition of previous adjunctions and we have also just proved that $\gamma$-\textbf{Kb} is stable by $\widehat{\o}^K_\gamma.$
 \end{proof}

We emphasize expected consequences from the $*$-autonomous category we will soon get since we will use them in slightly more general form.

 \begin{corollary}
 \label{coro:otimes_gamma_kcompl}
 For any $Y\in \Kc, X,Z_1,...,Z_m, Y_1,...,Y_n\in \gamma-\Kc, T\in k-\mathbf{Ref}$
 the following canonical linear maps are continuous $$ev_{X'_c}:(Y\varepsilon X)\widehat{\o}^K_{\gamma}X'_c\to Y,\ \ \ \ 
 comp_{T'_c}^*: (Y\varepsilon T)\widehat{\o}^K_{\gamma}((T'_c\varepsilon Z_1\cdots\varepsilon Z_m)^*_k)^*_k\to (Y\varepsilon Z_1\cdots\varepsilon Z_m),$$
$$
 comp_{T'_c}: ( Y_1\varepsilon\cdots\varepsilon Y_n\varepsilon T)\o_{\gamma}(T'_c\varepsilon Z_1\cdots\varepsilon Z_m)\to (Y\varepsilon Y_1\cdots\varepsilon Y_n\varepsilon Z_1\cdots\varepsilon Z_m),$$
 $$
 comp_{T'_c}^\sigma: (Y\varepsilon Y_1\cdots\varepsilon Y_n\varepsilon T)\o_{\sigma,\gamma}(T'_c\varepsilon Z_1\cdots\varepsilon Z_m)\to (Y\varepsilon Y_1\cdots\varepsilon Y_n\varepsilon Z_1\cdots\varepsilon Z_m),$$
Moreover for any $F,G\in \Kc$, $V,W\in \gamma-\mathbf{Kb}$ and $U,E$ any separated lcs
, there are continuous associativity maps $$Ass_\varepsilon:E\varepsilon(F\varepsilon G)\to (E\varepsilon F)\varepsilon G, \ \ \ Ass_\gamma:(U\widehat{\o}^K_\gamma V)\widehat{\o}^K_\gamma W\to U\widehat{\o}^K_\gamma (V\widehat{\o}^K_\gamma W),$$ $$ Ass_{\gamma,\varepsilon}:V\widehat{\o}^K_\gamma(T\varepsilon X)\to (V\widehat{\o}^K_\gamma T)\varepsilon X.$$
 \end{corollary}

  \begin{proof}
(1)  From the adjunction, the symmetry map in $L((Y\varepsilon X), (X\varepsilon Y))=L((Y\varepsilon X)\widehat{\o}^K_{\gamma}X'_c,Y)$ gives the first evaluation map.

(2) For the associativity $Ass_\varepsilon$, recall that using definitions and \eqref{ReductionKcCLCS} (using $F,G\in \Kc$):
\begin{align*}E\varepsilon(F\varepsilon G)=U(E_c\parr_b[F\varepsilon G]_c)=U(E_c\parr_b[F_c\parr_b G_c])&\to  U([E_c\parr_bF_c]\parr_b G_c)\\&\to U([U(E_c\parr_bF_c)]_c\parr_b G_c)=(E\varepsilon F)\varepsilon G,\end{align*}
where the first map is $U(Ass^{\parr_b}_{E_c,F_c,G_c})$ and the second obtained by functoriality from the unit $\eta_{E_c\parr_bF_c}:E_c\parr_bF_c\to [U(E_c\parr_bF_c)]_c.$
 
(3)  For the associativity $Ass_{\gamma}$, we know from the adjunction again, since $V'_c,W'_c\in \gamma-\Kc, V=(V'_c)'_c, W=(W'_c)'_c$:
$$L((U\widehat{\o}^K_\gamma V)\widehat{\o}^K_\gamma W, U\widehat{\o}^K_\gamma (V\widehat{\o}^K_\gamma W))=L((U\widehat{\o}^K_\gamma V),W'_c\varepsilon \Big(U\widehat{\o}^K_\gamma (V\widehat{\o}^K_\gamma W)\Big))=L(U,V'_c\varepsilon \Big(W'_c\varepsilon \Big(U\widehat{\o}^K_\gamma (V\widehat{\o}^K_\gamma W)\Big)\Big)).$$
Then composing with $Ass_\varepsilon$ (note the $\gamma$ tensor product term is the term requiring nothing but $k$-quasi-completeness for the adjunction to apply) gives a map:
$$L(U\widehat{\o}^K_{\gamma,\varepsilon}\Big(V'_c\varepsilon W'_c\Big)'_c, \Big(U\widehat{\o}^K_\gamma (V\widehat{\o}^K_\gamma W)\Big))\simeq L(U,\Big(V'_c\varepsilon W'_c\Big)\varepsilon \Big(U\widehat{\o}^K_\gamma (V\widehat{\o}^K_\gamma W)\Big))\to L(U,V'_c\varepsilon \Big(W'_c\varepsilon \Big(U\widehat{\o}^K_\gamma (V\widehat{\o}^K_\gamma W)\Big)\Big))$$

Since an equicontinuous set in $\Big(V'_c\varepsilon W'_c\Big)'_c$ is contained in an absolutely convex compact set, one gets by universal properties a continuous linear map :
$U\widehat{\o}^K_{\gamma,\varepsilon}\Big(V'_c\varepsilon W'_c\Big)'_c\to U\widehat{\o}^K_{\gamma}\Big(V'_c\varepsilon W'_c\Big)'_c.$

Finally by functoriality and the embedding of proposition \ref{lem:Kc_epsilon} there is a canonical continuous linear map: $U\widehat{\o}^K_{\gamma}\Big(V'_c\varepsilon W'_c\Big)'_c\to U\widehat{\o}^K_\gamma(V\widehat{\o}^K_\gamma W)$. Dualizing, we also have a map which we can evaluate at the identity map composed with all our previous maps to get $Ass_\gamma$: $$L(U\widehat{\o}^K_\gamma(V\widehat{\o}^K_\gamma W), \Big(U\widehat{\o}^K_\gamma (V\widehat{\o}^K_\gamma W)\Big))\to L(U\widehat{\o}^K_{\gamma,\varepsilon}\Big(V'_c\varepsilon W'_c\Big)'_c, \Big(U\widehat{\o}^K_\gamma (V\widehat{\o}^K_\gamma W)\Big))$$
 
(4)  We treat similarly the map $comp_{T'_c}^*$ in the case $m=2$, for notational convenience.  It is associated to $ev_{T'_c}\circ (id\o ev_{(Z_1)'_c})\circ (id\o ev_{(Z_2)'_c}\o id)$ via the following identifications. One obtains first a map between Hom-sets  using the previous  adjunction :
$$L\Big(\Big[\big((Y\varepsilon T)\widehat{\o}^K_{\gamma} (((T'_c\varepsilon Z_1)\varepsilon Z_2)^*_k)^*_k\big)\widehat{\o}^K_{\gamma} (Z_2)'_c\Big]\widehat{\o}^K_{\gamma} (Z_1)'_c, Y\Big)=L((Y\varepsilon T)\widehat{\o}^K_{\gamma}(((T'_c\varepsilon Z_1)\varepsilon Z_2)^*_k)^*_k, (Y\varepsilon Z_1)\varepsilon Z_2).$$

We compose this twice with $Ass_\gamma$ and the canonical map $(E^*_k)^*_k\to E$ for $E$ $k$-quasi-complete: 
\begin{align*} &L\Big((Y\varepsilon T)\widehat{\o}^K_{\gamma}\Big[\Big( ((T'_c\varepsilon Z_1)\varepsilon Z_2)\widehat{\o}^K_{\gamma} (Z_2)'_c\Big)\widehat{\o}^K_{\gamma} (Z_1)'_c\Big], Y\Big)\to L\Big((Y\varepsilon T)\widehat{\o}^K_{\gamma}\Big[\Big( (((T'_c\varepsilon Z_1)\varepsilon Z_2)^*_k)^*_k\widehat{\o}^K_{\gamma} (Z_2)'_c\Big)\widehat{\o}^K_{\gamma} (Z_1)'_c\Big], Y\Big)\to \\ &L\Big(\Big[(Y\varepsilon T)\widehat{\o}^K_{\gamma} \Big((((T'_c\varepsilon Z_1)\varepsilon Z_2)^*_k)^*_k\widehat{\o}^K_{\gamma} (Z_2)'_c\Big)\Big]\widehat{\o}^K_{\gamma} (Z_1)'_c, Y\Big)\to L\Big(\Big[\Big((Y\varepsilon T)\widehat{\o}^K_{\gamma} (((T'_c\varepsilon Z_1)\varepsilon Z_2)^*_k)^*_k\Big)\widehat{\o}^K_{\gamma} (Z_2)'_c\Big]\widehat{\o}^K_{\gamma} (Z_1)'_c, Y\Big).\end{align*}
Note that the first associativity uses the added $((\cdot)^*_k)^*_k$ making the Arens dual of the space $k$-quasi-complete as it should to use $Ass_\gamma$
 and the second since $(((T'_c\varepsilon Z_1)\varepsilon Z_2)^*_k)^*_k\widehat{\o}^K_{\gamma} (Z_2)'_c\in \gamma-\mathbf{Kb}$ from Proposition \ref{prop:adj_Kc_gammaKc}.
 
Note that $T'_c\in \Kc$ is required for definition of $ev_{(Z_i)'_c}$ hence the supplementary assumption $T\in\kref$ and not only $T\in\gamma-\Kc$.
  
 (5)By the last statement in lemma \ref{kduality}, we already know that $((T'_c\varepsilon Z_1\cdots\varepsilon Z_m)^*_k)^*_k$ and $T'_c\varepsilon Z_1\cdots\varepsilon Z_m$ have the same absolutely convex compact sets.
Hence for any absolutely compact set in this set $comp_{T'_c}^*$ induces an equicontinuous family in $L(Y_1\cdots\varepsilon Y_n\varepsilon T, Y\varepsilon Y_1\cdots\varepsilon Y_n\varepsilon Z_1\cdots\varepsilon Z_m)$. But now by symmetry on $\varepsilon$ product and of the assumption on $Y_i,Z_j$ one gets the second hypocontinuity to define $comp_{T'_c}$ by a symmetric argument.

 (6)One uses $comp_{T'_c}$  on $((Y'_c)'_c\varepsilon Y_1\cdots\varepsilon Y_n\varepsilon T)=(Y\varepsilon Y_1\cdots\varepsilon Y_n\varepsilon T)$ algebraically, since $(Y'_c)'_c\in \gamma-\Kc.$ This gives the separate continuity needed to define $comp_{T'_c}^\sigma$, the one sided $\gamma$-hypocontinuity follows from $comp_{T'_c}^*$ as in (5).
 
(7)  We finish by $Ass_{\gamma,\varepsilon}$. We know from the adjunction again composed with $Ass_\varepsilon$ and symmetry of $\varepsilon$  that we have a  map:
 $$L(T\varepsilon X, (V\widehat{\o}^K_\gamma T)\varepsilon\Big(V'_c\varepsilon X\Big))\to L(T\varepsilon X, V'_c\varepsilon\Big((V\widehat{\o}^K_\gamma T)\varepsilon X\Big))=L(V\widehat{\o}^K_\gamma(T\varepsilon X), (V\widehat{\o}^K_\gamma T)\varepsilon X)$$
 Similarly, we have canonical maps:
  $$ L(\big(T\varepsilon X\big)\widehat{\o}^K_\gamma(V\widehat{\o}^K_\gamma T)'_c, \Big(V'_c\varepsilon X\Big))\simeq L(T\varepsilon X, ((V\widehat{\o}^K_\gamma T)'_c)'_c\varepsilon\Big(V'_c\varepsilon X\Big))\to L(T\varepsilon X, (V\widehat{\o}^K_\gamma T)\varepsilon\Big(V'_c\varepsilon X\Big))$$
   $$L(\big(X\varepsilon T\big)\widehat{\o}^K_\gamma\big(T'_c\varepsilon V'_c\big), \Big(X\varepsilon V'_c\Big))\to L(\big(T\varepsilon X\big)\widehat{\o}^K_\gamma\big(\big(V'_c\varepsilon T'_c\big)'_c\big)'_c, \Big(V'_c\varepsilon X\Big))\to L(\big(T\varepsilon X\big)\widehat{\o}^K_\gamma(V\widehat{\o}^K_\gamma T)'_c, \Big(V'_c\varepsilon X\Big)).$$
      The image of $comp_{T'_c}\in L(\big(X\varepsilon T\big)\widehat{\o}^K_\gamma\big(T'_c\varepsilon V'_c\big), \Big(X\varepsilon V'_c\Big))$ gives $Ass_{\gamma,\varepsilon}$ since $X,V'_c\in\gamma-\Kc$.
 \end{proof}
  
 We refer to \cite{MelliesTabareau,TabareauPhD} for the study of dialogue categories from their definition, already recalled in subsection 2.3. Note that $*$-autonomous categories are a special case. 
 
We state first a transport lemma for dialogue categories along monoidal functors, which we will use several times.

\begin{lemma}
\label{lemma:transport_dialoque_categories}
Consider $(\mathcal{C}, \otimes_\mathcal{C}, 1_\mathcal{C})$, and  $(\mathcal{D}, \otimes_\mathcal{D}, 1_\mathcal{D})$ two symmetric monoidal categories, $R : \mathcal{C} \to \mathcal{D}$ a  functor, and $ L : \mathcal{D} \to \mathcal{C}$ the  left adjoint to $R$ which is assumed strictly monoidal. If $\neg$ is a tensorial negation on $\mathcal{C}$, then $ E \mapsto R(\neg L(E))$ is a tensorial negation on $\mathcal{D}.$
\end{lemma} 

\begin{proof}Let $\varphi^{\mathcal{C}}$ the natural isomorphism making $\neg$ a tensorial negation.
Let us call the natural bijections given by the adjunction $$\psi_{A,B}:\mathcal{D}(A,R(B))\simeq\mathcal{C}(L(A),B).$$

Define \begin{align*}\varphi^{\mathcal{D}}_{A,B,C}&=\psi_{A, \neg(L (B \otimes_\mathcal{D}  C ))}^{-1}\circ\varphi^{\mathcal{C}}_{L(A),L(B),L(C)}\circ \psi_{A \otimes_\mathcal{D} B,\neg(L (C))}:\mathcal{D}(A \otimes_\mathcal{D} B, F( \neg(L (C))))\to \mathcal{D}(A, F( \neg(L (B \otimes_\mathcal{D}  C ))))\end{align*}

It gives the expected natural bijection:
\begin{align*} \mathcal{D}(A \otimes_\mathcal{D} B, F( \neg(L (C))))&\simeq \mathcal{C}(L(A \otimes_\mathcal{D} B),  \neg(L (C)))=\mathcal{C}(L(A )\otimes_\mathcal{C} L(B)),  \neg(L (C)))\\&\simeq  \mathcal{C}(L(A),  \neg(L (B \otimes_\mathcal{D}  C ))))\simeq\mathcal{D}(A, F( \neg(L (B \otimes_\mathcal{D}  C )))) \end{align*}
where we have used strict monoidality of $L$: $L (B \otimes_\mathcal{D}  C ) = L(B) \otimes_\mathcal{C} L(C)$, and the structure of dialogue category on $C$. 

It remains to check the compatibility relation \eqref{DialogueCompatibility}.
For it suffices to note that by naturality of the adjunction, one has for instance:
$$\mathcal{D}(Ass^{\o_\mathcal{D}}_{A,B,C}, F( \neg(L (D))))
=\psi_{A \otimes_\mathcal{D}( B\otimes_\mathcal{D} C),\neg(L (D))}^{-1}\circ \mathcal{C}(L(Ass^{\o_\mathcal{D}}_{A,B,C}),  \neg(L (D)))\circ \psi_{(A \otimes_\mathcal{D} B)\otimes_\mathcal{D} C,\neg(L (D))}.$$
and since $L(Ass^{\o_\mathcal{D}}_{A,B,C})=
Ass^{\o_\mathcal{C}}_{L(A),L(B),L(C)}$ from compatibility of a strong monoidal functor, the new commutative diagram in $\mathcal{D}$ reduces to the one in $\mathcal{C}$ by intertwining.
\end{proof}

\begin{remark}
Note that we have seen or will see several examples of such monoidal adjunctions :
\begin{itemize}
 \item  between $(\Kc^{op},\varepsilon)$ and $(\CLCS^{op}, \parr_b)$ through the functors $L=(\cdot)_c$ and $R=(\ \hat{\cdot}^K\circ U)$ (proof of proposition \ref{lem:Kc_epsilon} and \eqref{ReductionKcCLCS}),
  \item between $(\CSch^{op}, \parr_b)$ and $(\McS^{op}, \varepsilon)$ through the functors $L=(\cdot)_{sc}$ and $R=(\ \hat{\cdot}^M\circ U)$(proposition \ref{McSSymmMonoidal2}).
\end{itemize} 
 \end{remark}

 \begin{theorem}\label{kRef} 
 $\Kc^{op}$ is a dialogue category with tensor product $\varepsilon$  and  tensorial negation $(\cdot)^*_k$ which has a commutative {and idempotent} continuation monad $((\cdot)^*_k)^*_k$. 
 
 Its continuation category is equivalent to the $*$-autonomous category $\kref$ with tensor product $E\o_{k} F=(E^*_{k}\varepsilon F^*_{k})^*_{k}$, dual $(.)^*_{k}$ and dualizing object $\K$.  It is stable by arbitrary products and direct sums. 
\end{theorem}

\begin{proof}
The structure of a dialogue category follows from the first case of the previous remark since $(\CLCS^{op}, \parr_b,(\cdot)'_b)$ is a $*$-autonomous category, hence a Dialogue category by Theorem \ref{FirstMALL} and then the new tensorial negation is  $R(\neg L(\cdot))=\hat{\cdot}^K\circ(\cdot)'_c$ which coincides with $(\cdot)^*_k$ on $\Kc$. The idempotency of the continuation monad comes from lemma \ref{kduality}.

 In order to check that the monad is commutative, one uses that from \cite[Prop 2.4]{TabareauPhD}, the dialogue category already implies existence of right and left tensor strengths say $t_{X,Y},\tau_{X,Y}$. Note that in order to see they commute, it suffices to see the corresponding result after applying $(\cdot)^*_k$. Then from proposition \ref{lem:Kc_epsilon}, the two maps obtained on $\widehat{X'_c}^K\widehat{\otimes}^K_{\beta e}\widehat{Y'_c}^K$
 must be extensions by continuity of an $\epsilon$-hypocontinuous multilinear map on $X'_c\otimes_{\beta e} Y'_c$, which is unique by \cite[\S 40.3.(1)]{Kothe2} which even works in the separately continuous case but strongly requires known the separate continuity of the extension. Hence we have the stated commutativity. 
 
The $*$-autonomous property follows from the following general lemma.
\end{proof} 

\begin{lemma}\label{DialogueToRef}
Let $(\mathcal{C}^{op},\parr_\mathcal{C},I,\neg)$  a dialogue category with a commutative and idempotent continuation monad and $\mathcal{D}\subset \mathcal{C}$ the full subcategory of objects of the form $\neg C, C\in \mathcal{C}$. Then $\mathcal{D}$ is equivalent to the  Kleisli category of the comonad $T=\neg\neg$ in $\mathcal{C}.$ 
If we define $\cdot \o_\mathcal{D} \cdot=\neg(\neg(\cdot) \parr_\mathcal{C} \neg(\cdot))$, then $(\mathcal{D},\o_\mathcal{D},I,\neg)$ is a $*$-autonomous category and $\neg: \mathcal{C}^{op}\to \mathcal{D}$ is strongly monoidal.
\end{lemma}

\begin{proof}
From the already quoted \cite[Prop 2.9]{TabareauPhD} of Hagasawa, the cited Kleisli category (or Continuation category $\mathcal{C}^{\neg}$) is a $*$-autonomous category since we start from a   Dialogue category with commutative and idempotent continuation monad. 
Consider $\neg: \mathcal{C}^{\neg}\to \mathcal{D}$. $\mathcal{D}(\neg A,\neg B)=\mathcal{C}^{op}(\neg B,\neg A)=\mathcal{C}^{\neg}(A,B)$ which gives that $\neg$ is fully faithful on the continuation category. The map $\neg:\mathcal{D}\to \mathcal{C}^{\neg}$ is the strong inverse of the equivalence since $\neg\circ \neg\simeq Id_{\mathcal{D}}$ by choice of $\mathcal{D}$, and idempotency of the continuation and the canonical map $J_{\neg A}\in \mathcal{C}^{\neg}(\neg\neg(A),id_{ \mathcal{C}^{\neg}}(A))=\mathcal{C}^{op}(\neg A,\neg(T(a))$ is indeed natural in $A$ and it is an isomorphism in $\mathcal{C}^{\neg}.$ Therefore we have a strong equivalence. Recall that the commutative strength $t_{A,B}:A\parr_{\mathcal{C}} T(B)\to T(A\parr_{\mathcal{C}} B),t'_{A,B}:T(A)\parr_{\mathcal{C}} B\to T(A\o_{\mathcal{C}} B)$ in $\mathcal{C}^{op}$,  implies that we have isomorphisms $$I_{A,B}=\neg\Big(T(t'_{A,B})\circ t_{T(A),B}\Big)\circ J^{2op}_{\neg (A\parr_{\mathcal{C}} B)}:\neg (A\parr_{\mathcal{C}} B)\simeq \neg (\neg\neg A\parr_{\mathcal{C}} B)\simeq \neg (\neg\neg A\parr_{\mathcal{C}} \neg\neg   B)$$ with commutation relations $I_{A,B}=\neg\Big(T(t_{A,B})\circ t'_{A,T(B)}\Big)\circ J^{2op}_{\neg (A\parr_{\mathcal{C}} B)}$. 
This gives in $\mathcal{D}$ the compatibility map for the strong monad: $\mu_{A,B}=I_{A,B}^{op-1}:\neg (A\parr_{\mathcal{C}} B)\simeq\neg A\o_{\mathcal{D}} \neg B$. Checking the associativity and unitarity for this map is a tedious computation left to the reader using axioms of strengths, commutativity, functoriality. 
This concludes.
\end{proof}

\subsection{A strong notion of smooth maps}\label{sec:ksmooth}
 During this subsection, $\K=\R$ so that we deal with smooth maps and not holomorphic ones while we explore the consequence of our $*$-autonomy results for the definition of a nice notion of smoothness.
 
We recall the definition of (conveniently) smooth maps as used by Frolicher, Kriegl and Michor : a map $f : E \to F$ is smooth if and only if for every smooth curve $ c : \mathbb{R} \to E$, $f \circ c$ is a smooth curve. See \cite{KrieglMichor}. They define on a space of smooth curves the usual topology of uniform convergence on compact subsets of each derivative. Then they define  on the space of smooth functions between Mackey-complete spaces $E$ and $F$ the projective topology with respect to all smooth curves in $E$ (see also section 6.1 below).

As this definition fits well in the setting of bounded linear maps and bounded duals, but not in our setting using continuous linear maps, we make use of a slightly different approach by Meise \cite{Meise}.  Meise works with k-spaces, that is spaces $E$ in which continuity on $E$ is equivalent to continuity on compacts subsets of $E$. We change his definition and rather use a continuity condition on compact sets in the definition of smooth functions.

\begin{definition}
For $X,F$ separated lcs we call $C^\infty_{co}(X,F)$  the space of infinitely many times G\^ateaux-differentiable  functions with derivatives continuous on compacts with value in the space $L_{co}^{n+1}(E,F)=L_{co}(E,L_{co}^{n}(E,F))$ with at each stage the topology of uniform convergence on compact sets. We put on it the topology of uniform convergence on compact sets of all derivatives in the space $L_{co}^{n}(E,F).$

We denote by $C^\infty_{co}(X)$ the space $\Cinco(X, \mathbb{K})$.
\end{definition}

One could treat similarly the case of an open set $\Omega\subset X.$ We always assume $X$ $k$-quasi-complete.

Our definition is almost the same as in \cite{Meise}, except for the continuity condition restricted to compact sets. Meise works with $k$-spaces, that is spaces $E$ in which continuity on $E$ is equivalent to continuity on compacts subsets of $E$. Thus for $X$ a $k$-space, one recovers exactly Meise's definition. Since a (DFM) space $X$ is a $k$-space (\cite[Th 4.11 p 39]{KrieglMichor}) his corollary 7 gives us that for such an $X$, $C^\infty_{co}(X,F)$ is a Fr\'echet space as soon as $F$ is.
 Similarly for any (F)-space or any (DFS)-space $X$ then his corollary 13 gives $C^\infty_{co}(X,\R)$ is a Schwartz space.

As in his lemma 3 p\ 271,\ 
if $X$ $k$-quasi-complete, the G\^ateaux differentiability condition is automatically uniform on compact sets (continuity on absolutely convex closure of compacts of the derivative is enough for that), and as a consequence, this smoothness implies convenient smoothness. We will therefore use the differential calculus from \cite{KrieglMichor}.

 One are now ready to obtain a category.

\begin{proposition}\label{Meise1}
 $\kref$ is a category with $C^\infty_{co}(X,F)$ as spaces of maps, that we denote $\kref_\infty$. Moreover, for any $g\in C^\infty_{co}(X,Y), Y,X\in k-\mathbf{Ref}$, any $F$ Mackey-complete, $\cdot\circ g:C^\infty_{co}(Y,F)\to C^\infty_{co}(X,F)$ is linear continuous.
\end{proposition}

\begin{proof}
For stability by composition, we show more, consider $g\in C^\infty_{co}(X,Y), f\in C^\infty_{co}(Y,F)$ with $X,Y\in k-\mathbf{Ref}$ and $F\in \Mc$ we aim at showing  $f\circ g\in C^\infty_{co}(X,F)$. We use stability of composition of conveniently smooth maps, we can use the chain rule \cite[Thm 3.18]{KrieglMichor}. This enables to make the derivatives valued in $F$ if $F$ is Mackey-complete so that, up to going to the completion, we can assume $F\in \Kc$ since the continuity conditions are the same when the topology of the target is induced.  This means that we must show continuity on compact sets of expressions of the form $$(x,h)\mapsto df^l(g(x))(d^{k_1}g(x),...,d^{k_l}g(x))(h_1,...,h_m),m=\sum_{i=1}^l k_i, h\in Q^m.$$
First note that $L_{co}(X,F)\simeq X'_c\varepsilon F$, $L_{co}^n(X,F)\simeq (X'_c)^{\varepsilon n}\varepsilon F$ fully associative for the spaces above. 

 Of course for $K$ compact in $X$, $g(K)\subset Y$ is compact, so  $df^l\circ g$ is continuous on compacts with value in $(Y'_c)^{\varepsilon l}\varepsilon F$ so that continuity comes from continuity of the map obtained by composing various $Comp_{Y}^*,Ass_\varepsilon$ from Corollary \ref{coro:otimes_gamma_kcompl} (note $Ass_\gamma$ is not needed with chosen parentheses): $$\Big(\Big(...\Big(\big((Y'_c)^{\varepsilon l}\varepsilon F\big)\otimes_\gamma \big(\big((X'_c)^{\varepsilon k_1}\varepsilon Y\big)^*_k\big)^*_k\Big)\otimes_\gamma ...\Big)\otimes_\gamma \big(\big((X'_c)^{\varepsilon k_l}\varepsilon Y\big)^*_k\big)^*_k\Big)\to \big((X'_c)^{\varepsilon m}\varepsilon F\big)$$
 and this implies continuity on products of absolutely convex compact sets of the corresponding multilinear map even without $((\cdot)^*_k)^*_k$ since from lemma \ref{kduality} absolutely convex compact sets are the same in both spaces (of course with same induced compact topology).  We can compose it with the continuous function on compacts with value in a compact set (on compacts in $x$) :$x\mapsto (df^l(g(x)),d^{k_1}g(x),...,d^{k_l}g(x)).$ The continuity in $f$ is similar and uses hypocontinuity of the above composition (and not only its continuity on products of compacts).
\end{proof}

We now prove the \emph{Cartesian closedeness} of the category $\kref$, the proofs being slight adaptation of the work by Meise \cite{Meise}

\begin{proposition}\label{Meise2}\label{prop:kcomplCco}
For any $X\in k-\mathbf{Ref}$, $C^\infty_{co}(X,F)$ is $k$-quasi-complete (resp. Mackey-complete) as soon as $F$ is.
\end{proposition}

\begin{proof}
This follows from the projective kernel topology on $\Cinco(X,F)$, Corollary \ref{prop:k-compl-lin} and the corresponding statement for $C^0(K,F)$ for $K$ compact. In the Mackey-complete case we use the remark at the beginning of step 2 of the proof of Theorem \ref{rhoref} that a space is Mackey-complete if and only if the bipolar of any Mackey-Cauchy sequence is complete.  We treat the two cases in parallel, if $F$ is $k$-quasi-complete (resp Mackey-complete), take $L$ a compact set (resp. a Mackey-Cauchy sequence) in $C^0(K,F)$, $M$ its bipolar, its image by evaluations $L_x$ are compact (resp. a Mackey-null sequence) in $F$ and the image of $M$ is in the bipolar  of $L_x$ which is complete in $F$ hence a Cauchy net in $M$ converges pointwise in $F$. But the Cauchy property of the net then implies as usual uniform convergence of the net to the pointwise limit. This limit is therefore continuous, hence the result.
\end{proof}

The following two propositions are an adaptation of the result by Meise \cite[Thm 1 p 280]{Meise}. Remember though that his $\varepsilon$ product and $E'_c$ are different from ours, they correspond to replacing absolutely convex compact sets by precompact sets. This different setting forces him to assume quasi-completeness to obtain a symmetric  $\varepsilon$-product in his sense.

\begin{proposition}\label{Meise3}
For any $k$-reflexive space $X$, any compact $K$, and any separated $k$-quasi-complete space $F$ one has $\Cinco(X,F) \simeq \Cinco(X)\varepsilon F, C^0(K,F) \simeq C^0(K)\varepsilon F$. Moreover, if $F$ is any \lcs, we still have a canonical embedding $J_X:\Cinco(X)\varepsilon F\to \Cinco(X,F)$.
\end{proposition}
\begin{proof}
 We build a map $ev_X\in C^\infty_{co}(X, (C^\infty_{co}(X))'_c)$ defined by $ev_{X}(x)(f)=f(x)$ and show that $\cdot \circ ev_X: C^\infty_{co}(X)\varepsilon F=L_\epsilon((C^\infty_{co}(X))'_c,F)\to C^\infty_{co}(X, F)$ is a topological isomorphism and an embedding if $F$ only Mackey-complete. The case with a compact $K$ is embedded in our proof and left to the reader.

(a) We first show that the expected  $j$-th differential $ev_{X}^j(x)(h)(f)=d^jf(x).h$ indeed gives a map: $$ev_X^j\in C^0_{co}(X, L^j_{co}(X,(C^\infty_{co}(X))'_c)).$$
First note that for each $x\in X$, $ev_{X}^j(x)$ is in the expected space. Indeed, by definition of the topology $f\mapsto d^jf(x)$ is  linear continuous in $L(C^\infty_{co}(X), L^j_{co}(X,\R))\subset L(((C^\infty_{co}(X))'_c)'_c, L^j_{co}(X,\R))= (C^\infty_{co}(X))'_c\varepsilon (X'_c)^{\varepsilon j}$. Using successively $Ass^\epsilon$ from Corollary \ref{coro:otimes_gamma_kcompl} (note no completeness assumption on $(C^\infty_{co}(X))'_c$ is needed for that) hence $ev_{X}^j(x)\in (\cdots ((C^\infty_{co}(X))'_c\varepsilon X'_c) \cdots \varepsilon X'_c)=L^j_{co}(X,(C^\infty_{co}(X))'_c)$.
 Then, once the map well-defined, we must check its continuity on compacts sets in variable $x\in K\subset X$, uniformly on compacts sets for $h\in Q$, one must check convergence in  $(C^\infty_{co}(X))'_c.$ But everything takes place in a product of compact sets and from the definition of the topology on $C^\infty_{co}(X)$, $ev_X^j(K)(Q)$ is equicontinuous in $(C^\infty_{co}(X))'$. But from \cite[\S 21.6.(2)]{Kothe} the topology $(C^\infty_{co}(X))'_c$ coincides with $(C^\infty_{co}(X))'_\sigma$ on these sets. Hence we only need to prove for any $f$ continuity of $d^jf$ and this follows by assumption on $f$. This concludes the proof of (a).

(b) Let us note that for $f\in L_\varepsilon((C^\infty_{co}(X))'_c,F)$, $f \circ ev_X\in C^\infty_{co}(X, F)$. We first note that  $f \circ ev_X^j(x)=d^j(f\circ ev_X)(x)$ as in step  (c) in the proof of \cite[Thm 1]{Meise}. This shows for $F=(C^\infty_{co}(X))'_c$ that the G\^ateaux derivative  is $d^jev_X=ev_X^j$ and therefore the claimed 
$ev_X\in C^\infty_{co}(X, (C^\infty_{co}(X))'_c)$.

(c)$f\mapsto f \circ ev_X$ is the stated isomorphism. The monomorphism property is the same as (d) in Meise's proof. Finding a right inverse $j$ proving surjectivity is the same as his (e) . Let us detail this since we only assume $k$-quasi-completeness on $F$. We want $j:C^\infty_{co}(X, F)\to C^\infty_{co}(X)\varepsilon F=L(F'_c,C^\infty_{co}(X))$
for $y'\in F',f\in C^\infty_{co}(X, F)$ we define $j(f)(y')=y'\circ f$. Note that from convenient smoothness we know that the derivatives are $y'\circ d^jf(x)$ and $d^jf(x)\in (X'_c)^{\varepsilon n}\varepsilon F=(X'_c)^{\varepsilon n}\varepsilon (F'_c)'_c$ algebraically so that, since $(F'_c)'_c$ $k$-quasi-complete, one can use $ev_{F'_c}$ from Corollary \ref{coro:otimes_gamma_kcompl} to see $y'\circ d^jf(x)\in (X'_c)^{\varepsilon n}$ and one even deduces (using only separate continuity of $ev_{F'_c}$) its continuity in $y'$. Hence $j(f)(y')$ is valued in $C^\infty_{co}(X)$ and from the projective kernel topology, $j(f)$ is indeed continuous. The simple identity showing that $j$ is indeed the expected right inverse proving surjectivity is the same as in Meise's proof.
\end{proof}

\begin{proposition}\label{Meise4}
For any space $X_1,X_2\in k-\mathbf{Ref}$ and any 
Mackey-complete \lcs $F$ we have :
 $$C^\infty_{co}(X_1\times X_2,F)\simeq C^\infty_{co}(X_1,C^\infty_{co}(X_2,F)).$$
\end{proposition}

 \begin{proof}
Construction of the curry map $\Lambda$ is analogous to \cite[Prop 3 p 296]{Meise}. Since all spaces are Mackey-complete, we already know from \cite[Th 3.12]{KrieglMichor} that there is a Curry map valued in $C^\infty(X_1,C^\infty(X_2,F))$, it suffices to see that the derivatives $d^{j}\Lambda(f)(x_1)$ are continuous on compacts with value $C^\infty_{co}(X_2,F).$ But this derivative coincides with a partial derivative of $f$, hence it is valued pointwise  in  $C^\infty_{co}(X_2,F)\subset C^\infty_{co}(X_2,\widehat{F}^K).$ Since we already know all the derivatives are pointwise valued in $F$, we can assume $F$ $k$-quasi-complete.  But the topology for which we must prove continuity is a projective kernel, hence we only need to see that $d^k(d^{j}\Lambda(f)(x_1))(x_2)$ continuous on compacts in $x_1$ with value in 
$L_{c}^j(X_1,C^0(K_2,L_{c}^j(X_2,F)))$. But we are in the case where the $\varepsilon$ product is associative, hence the above space is merely $C^0(K_2)\varepsilon L_{c}^j(X_1,L_{c}^j(X_2,F))=C^0(K_2, L_{c}^j(X_1,L_{c}^j(X_2,F))).$ We already know the stated continuity in this space from the choice of $f$. The reasoning for the inverse map is similar using again the convenient smoothness setting (and Cartesian closedness $C^0(K_1,C^0(K_2))=C^0(K_1\times K_2)$).
 \end{proof}

\section{Schwartz locally convex spaces, Mackey-completeness and $\rho$-dual.}
\label{sec:rhodual}
\label{sec:rhorefl}
 In order to obtain a $*$-autonomous category adapted to convenient smoothness, we want to replace $k$-quasi-completeness by the weaker Mackey-completeness and adapt our previous section.

In order to ensure associativity of the dual of the $\varepsilon$-product, Mackey-completeness is not enough as we saw in section \ref{sec:zeta}.  We have to restrict simultaneously to Schwartz topologies. After some preliminaries in subsection 5.1, we thus define our appropriate weakened reflexivity (${\rho}$-reflexivity) in subsection 5.2, and investigate categorical completeness in 5.3.

 We want to put a $*$-autonomous category structure on the category 
$\rRef$ of ${\rho}$-reflexive 
(which implies Schwartz Mackey-complete) locally convex 
spaces with continuous linear maps as morphisms.
 
 It turns out that one can carry on as in section \ref{sec:kref} and put a Dialogue category structure on Mackey-complete Schwartz spaces. Technically, the structure is derived via an intertwining from the one in $\CSch$ in Theorem \ref{FirstMALL}. This category can even be seen as chosen in order to fit our current Mackey-complete Schwartz setting. We actually proved all the results first without using it and made appear the underlying categorical structure afterwards.
 
Then the continuation monad will give the $*$-autonomous category structure on $\rRef$ where the internal hom is described as  $$E\multimap_\rho F=(([E^*_{\rho}]\varepsilon F)^*_{\rho})^*_{\rho}$$ and based on a twisted Schwartz $\varepsilon$-product. The space is of course the same (as forced by the maps of the category) but the topology is strengthened. 
Our preliminary work on double dualization in section 6.2 make this construction natural to recover an element of $\rRef$ anyway. 

 \subsection{Preliminaries in the Schwartz Mackey-complete setting}
 
 We define $\Mc$ (resp. $\Sch$) the category of Mackey-complete spaces (resp Schwartz space) and linear continuous maps. The category $\McS$ is the category of Mackey-complete Schwartz spaces.

 We first recall \cite[Corol 16.4.3]{Jarchow}. Of course it is proven their for the completed variant, but by functoriality, the original definition of this product is a subspace and thus again a Schwartz space.
 
 \begin{proposition}\label{Sepsilon}
 If $E$ and $F$ are separated Schwartz locally convex spaces, then  so is $E\varepsilon F$.
 \end{proposition}

 We can benefit from our section 3 to obtain associativity:
 
 \begin{proposition}\label{SMepsilon}\label{AssocScwhartz}
 $(\McS,\varepsilon)$ is a symmetric monoidal complete and cocomplete category. $\McS\subset\Mc$ is a reflective subcategory with reflector (left adjoint to inclusion) $\mathscr{S}$ and the inclusion is strongly monoidal.
 Moreover on $\LCS$, $\mathscr{S}$ and $\ \widehat{\cdot}^M$ commute and their composition is the reflector of $\McS\subset\LCS$.
 \end{proposition}
 
\begin{proof}
From Theorem \ref{zetaparr}, we know $(\Mc,\zeta)$
is of the same type and for $E,F\in \McS$, lemma \ref{etageneral} with the previous lemma gives $E\varepsilon F=E\zeta F\in \McS.$ Hence, we deduce $\McS$ is also symmetric monoidal and the inclusion strongly monoidal. The unit of the adjunction is the canonical identity map $\eta_E:E\to\mathscr{S}(E)$ and counit is identity satisfying the right relations hence the adjunction. 
From the adjunction the limits in $\McS$ are the limits in $\Mc$ and colimits are obtained by applying $\mathscr{S}$ to colimits of $\Mc$.
It is easy to see that $\McS\subset\Sch\subset\mathbf{LCS}$ are also reflective, hence the two ways of writing the global composition gives the commutation of the left adjoints.
\end{proof}

\subsection{$\rho$-reflexive spaces and their Arens-Mackey duals}

We define a new notion for the dual of $E$, which consists of taking the Arens-dual of the Mackey-completion of the Schwartz space $\mathscr{S}( E)$, which is once again transformed into a Mackey-complete Schwartz space $E^*_\rho$.

{  \begin{definition}
 For a \lcs $E$, the topology ${\mathscr{S}\rho}(E',E) $ on $E'$ is the topology of uniform convergence on absolutely convex compact sets of $\widehat{\mathscr{S}(E)}^M$. We write $E'_{\mathscr{S}\rho}=(E',{\mathscr{S}\rho}(E',E))=(\widehat{\mathscr{S}(E)}^M)'_c.$ 
We write $E^*_\rho =  \widehat{\mathscr{S} (E'_{\mathscr{S}\rho})}^M$ and $E'_{\RR}=\mathscr{S}(E'_{\mathscr{S}\rho})$.
 \end{definition} }


\begin{remark}
Note that $E'_{\mathscr{S}\rho}$ is in general not Mackey-complete : there is an Arens dual of a Mackey-complete space (even of a nuclear complete space with its Mackey topology) which is not Mackey-complete using \cite[thm 34, step 6]{BrouderDabrowski}. 
Indeed take $\Gamma$ a closed cone in the cotangent bundle (with $0$ section removed) $\dot{T}^*R^n.$ Consider H\"ormander's space $E=\mathcal{D}_\Gamma'(\R^n)$ of distributions with wave front set included in $\Gamma$ with its normal topology in the terminology of \cite[Prop 12,29]{BrouderDabrowski}. It is shown there that $E$ is nuclear complete. Therefore the strong dual is $E'_\beta=E'_c$. Moreover, \cite[Lemma 10]{BrouderDabrowski} shows that this strong dual if $\mathcal{E}'_\Lambda$, the space of compactly supported distributions with a wave front set in the open cone $\Lambda=-\Gamma^c$ with a standard inductive limit topology. This dual is shown to be nuclear in \cite[Prop 28]{BrouderDabrowski}. Therefore we have $E'_c=E'_{\mathscr{S}\rho}$. Finally, as explained in the step 6 of the proof of \cite[Thm 34]{BrouderDabrowski} where it is stated it is not complete, as soon as $\Lambda$ is not closed (namely by connectedness when $\Gamma\not\in\{\emptyset,\dot{T}^*R^n\}$), then $E'_c$ is not even Mackey-complete. This gives our claimed counter-example. The fact that $E$ above has its Mackey topology is explained in \cite{Dab14a}.
\end{remark}

First note the functoriality lemma  :
 
 \begin{lemma}\label{rhoDualFunctor}
 $(\cdot)^*_{\rho}$ and  $(\cdot)'_{\RR}$ are contravariant endofunctors on $\LCS.$
 \end{lemma}
 \begin{proof}They are obtained by composing $\mathscr{S}$, $(\cdot)'_c$ and $\ \widehat{\cdot}^M$ (recalled in Theorem \ref{zetaparr}).
 \end{proof}

From Mackey theorem and the fact that completion does not change the dual, we can deduce immediately that we have the following algebraic identities  $\mathscr{S(}(E'_{\mathscr{S}\rho})'_{\mathscr{S}\rho})=(E'_{\mathscr{S}\rho})'_{\mathscr{S}\rho}=\widehat{\mathscr{S}(E)}^M$. 

From these we deduce the fundamental algebraic equality:
\begin{equation}
\label{alg_eq_refl}
 (E^*_\rho)^*_{\rho} = \widehat{\mathscr{S}(E)}^M 
\end{equation}

\begin{definition}
A \lcs $E$ is said \emph{$\rho$-reflexive} if the canonical map $E\to\widehat{\mathscr{S}(E)}^M = (E^*_\rho)^*_{\rho}$ gives a topological isomorphism $E \simeq (E^*_{\rho})^*_{\rho}$.
\end{definition}

 We are looking for a condition necessary to make the above equality a topologically one. The following theorem demonstrates an analogous to $E'_c=((E'_c)'_c)'_c$ for our new dual. For, we now make use of lemma \ref{ordinalCompletion}. 
 
\begin{theorem}\label{rhoref}
Let $E$ be a separated  locally convex space, then  $E^*_{\rho}$ is  $\rho$-reflexive. As a consequence, if $E$ is $\rho$-reflexive, so is $E^*_{\rho}$ and $\mathscr{S}((E'_c)'_c)\simeq E\simeq \mathscr{S}((E'_\mu)'_\mu)$ topologically. Moreover, when $E$ is Mackey-complete $(E^*_{\rho})^*_{\rho}=(E^*_{\rho})'_{\RR} $ and $E$ have the same bounded sets.
 \end{theorem}

 \begin{proof}
 \setcounter{Step}{0}
 
Note that the next-to-last statement is obvious since if $E$  $\rho$-reflexive, we have $(E^*_{\rho})=((E^*_{\rho})^*_{\rho})^*_{\rho}$ and the last space is always $\rho$-reflexive. Moreover, from the two first operations applied in the duality, one can and do assume $\mathscr{S}(E)$ is Mackey-complete.
 
  Let us write also $\mathscr{C}_M(.)=\widehat{.}^M$ for the Mackey-completion functor and for an ordinal $\lambda$,$\mathscr{C}_M^\lambda(E)=E_{M,\lambda}$ from lemma \ref{gammacompletion}.

  Note also that since the bounded sets in $E$ and $\mathscr{S}(E)$ coincide by Mackey Theorem \cite[Th 3 p 209]{Horvath}, one is Mackey-complete if and only if the other is.
  
\begin{step}
 $\mathscr{S}((E'_\mu)'_\mu)$ is Mackey-complete if $\mathscr{S}(E)$ is Mackey-complete.
 \end{step}   
This follows from the continuity  $\mathscr{S}((E'_\mu)'_\mu)\to \mathscr{S}(E)$ and the common dual, they have same bounded sets, hence same Mackey-Cauchy/converging sequences.

\begin{step}
 $\mathscr{S}((E'_{\mathscr{S}\rho})'_{\mathscr{S}\rho})$ is Mackey-complete if $\mathscr{S}(E)$ is Mackey-complete.
 \end{step} 
First note that a space is Mackey-complete if and only if any $K$, closed absolutely convex cover of a Mackey-Cauchy sequence, is complete. Indeed, if this is the case, since a Mackey-Cauchy sequence is Mackey-Cauchy for the saturated bornology generated by Mackey-null sequences \cite[Thm 10.1.2]{Jarchow}, it is Mackey in a normed space having a ball (the bipolar of the null sequence) complete in the lcs, hence a Banach space in which the Cauchy sequence must converge. Conversely, if a space is Mackey-complete, the sequence converges in some Banach space, hence its bipolar in this space is compact, and thus also in the \lcs and must coincide with the bipolar computed there which is therefore compact hence complete.
 
We thus apply lemma \ref{ordinalCompletion} to $K$ the closed absolutely convex cover of a Mackey-Cauchy sequence in $\mathscr{S}((E'_{\mathscr{S}\rho})'_{\mathscr{S}\rho}),$  $E_0=\mathscr{S}(E'_{\mathscr{S}\rho})$, $D=\mathscr{S}((\cdot)'_c)$, $E_\lambda=\mathscr{C}_M^{\lambda}(E_0)$ eventually yielding to the Mackey completion so that $D(E_{\lambda_0})=\mathscr{S}((E'_{\mathscr{S}\rho})'_{\mathscr{S}\rho})$ for $\lambda_0$ large enough and with $D(E_0)=\mathscr{S}((E'_\mu)'_\mu)$ using lemma \ref{MackeyArensSchwartz}. The result will conclude since the above bipolar $K$ computed in $D(E_{\lambda_0})$ must be complete by Mackey-completeness of this space hence complete in $\mathscr{S}((E'_{\mathscr{S}\rho})'_{\mathscr{S}\rho})$ by the conclusion of  lemma \ref{ordinalCompletion} and hence the  bipolar computed in there which is a closed subset will be complete too. We thus need to check the assumptions of  lemma \ref{ordinalCompletion}. The assumption at successor ordinal comes from the definition of $\mathscr{C}_M^{\lambda+1}$ since any point there $z$ satisfy $z\in N=\overline{\Gamma(L)}$ with $L=\{t_n,n\in \N\}$ a Mackey-Cauchy sequence in $E_\lambda$. Thus there is an absolutely convex bounded $B\subset E_\lambda$ with $(t_n)$ Cauchy in the normed space $(E_\lambda)_B\subset(E_{\lambda_0})_{\overline{B}} $. 
We know $t_n \to t$ in the completion so $t\in N$.

But since $E_{\lambda_0}$ is Mackey-complete, this last space is a Banach space, $t_n\to t$ and it is  contained in $C_1=\{s_0=2t, s_n=2(t_n-t), n\in \N\}^{oo}$.
$||s_n||_{\overline{B}}\to 0$ we can define $r_n=s_n/\sqrt{||s_n||_{\overline{B}}}$ which converges to $0$ in $(E_{\lambda_0})_{\overline{B}} $. {Hence $\{r_n, n\in \N\}$ is precompact as any converging sequence and so is its bipolar say $C$ computed in the Banach space $(E_{\lambda_0})_{\overline{B}} $, which is also complete thus compact.  $C$ is thus compact  in $E_{\lambda_0}$ too.} Since $||s_n||_C\leq \sqrt{||s_n||_{\overline{B}}}\to 0$ it is Mackey-null for the bornology of absolutely convex compact sets of $E_{\lambda_0}$. Thus $C_1$ is equicontinuous in $(D(E_{\lambda_0}))'$ and so is $t_n$ as expected.


 \begin{step}
 Conclusion.
 \end{step}
 
 Note we will use freely lemma \ref{MackeyArensSchwartz}.
If $\mathscr{S}(E)$ is Mackey-complete, and $Z=(E'_{\mathscr{S}\rho})'_{\mathscr{S}\rho}$ then from step 2, $\mathscr{S}(Z), \mathscr{S}(Z'_{\mathscr{S}\rho})$ are Mackey-complete and as a consequence $Z'_{\mathscr{S}\rho}=Z'_\mu$ and then $(Z'_{\mathscr{S}\rho})'_{\mathscr{S}\rho}=(Z'_{\mathscr{S}\rho})'_\mu
=(Z'_\mu)'_\mu$ topologically. In particular we confirm our claimed topological identity: $$(E^*_{\rho})^*_{\rho}\equiv\mathscr{C}_M(\mathscr{S}(Z))=\mathscr{S}(Z)\equiv\SS((E^*_{\rho})'_{\SS\rho}) .$$
 
 From the continuous linear identity map: $(Z'_\mu)'_\mu\to Z$ one gets a similar map $\mathscr{S}((Z'_{\mathscr{S}\rho})'_{\mathscr{S}\rho})\to \mathscr{S}(Z).$

Similarly, since there is a continuous identity map $Z\to \mathscr{S}(Z)$, one gets a continuous linear map $Z'_\mu\to Z'_c=(H'_c)'_c\to \widehat{\mathscr{S}(E'_{\mathscr{S}\rho})}^M\equiv H$. Since the last space is a Schwartz topology on the same space, one deduces a continuous map $\mathscr{S}(Z'_\mu)\to \widehat{\mathscr{S}(E'_{\mathscr{S}\rho})}^M$. Finally, an application of Arens duality again  leads to a continuous identity map: $Z\to (Z'_\mu)'_\mu=((Z)'_{\mathscr{S}\rho})'_{\mathscr{S}\rho}$. This concludes to the equality $\mathscr{S}((Z'_{\mathscr{S}\rho})'_{\mathscr{S}\rho})=\mathscr{S}((Z'_{\mu})'_{\mu})= \mathscr{S}((Z'_{\mu})'_{c})=\mathscr{S}(Z).$ As a consequence  if $E$ is $\rho$-reflexive, it is of the form $E=\mathscr{S}(Z)$ and one deduces $\mathscr{S}((E'_{c})'_{c})=E=\mathscr{S}((E'_{\mu})'_{\mu})$.

Consider $E$ a Mackey-complete and Schwartz space. Then $(E^*_{\rho})^*_{\rho}=\mathscr{S}\left[\left(\widehat{\mathscr{S}(E'_c)}^M\right)'_c\right]$ and we have continuous linear maps $E'_c\to \mathscr{S}(E'_c)\to \widehat{\mathscr{S}(E'_c)}^M \to \widetilde{\mathscr{S}(E'_c)}$ which by duality and functoriality give continuous linear maps:
\begin{equation} \label{idmap}
 \left(\widetilde{\mathscr{S}(E'_c)}\right)'_c\to (E^*_{\rho})^*_{\rho}\to \mathscr{S}((E'_c)'_c)\to E.
\end{equation}

Let us show that a $\rho$-dual $ Y = E^*_\rho$ is always $\rho$-reflexive (for which we can and do assume  $E$ is Mackey-complete and Schwartz). According to equation \eqref{alg_eq_refl}, as $Y$ is Mackey-complete and Schwartz we already have the algebraic equality $(Y^*_\rho)^*_\rho = Y$.  The above equation gives  a continuous identity map  $(Y^*_\rho)^*_\rho \rightarrow Y$.
Now according to step 2 of this proof $Y'_c\equiv(E'_{\mathscr{S}\rho})'_{\mathscr{S}\rho}=Z$ and $(Y'_\mu)'_\mu$ are Mackey-complete. Thus $(Y^*_\rho)^*_\rho = \mathscr{S} ( [\mathscr{S}(Y'_c)]'_c)=\mathscr{S} ( [Y'_\mu]'_\mu)$. However the equation \eqref{idmap} gives a continuous identity map $(E^*_\rho)^*_\rho \rightarrow \widehat{\mathscr{S}(E)}^M$, which by duality and functoriality of $\mathscr{S}$ leads to a continuous identity map $ Y \rightarrow  (Y^*_\rho)^*_\rho$. Every $\rho$-dual is thus $\rho$-reflexive. 

 Let us show the last statement, since a space and its associated Schwartz space have the same bounded sets, we can assume $E$ Mackey-complete and Schwartz. As a consequence of the equation \eqref{idmap} and of the next lemma,
 the bounded sets in the middle term $(E^*_{\rho})^*_{\rho}$ have to coincide too, and the last statement of the proposition is shown.
  \end{proof}

\begin{lemma}\label{boundedCompletedArens}
If $E$ is Mackey-complete lcs, then $[\widetilde{\mathscr{S}(E'_c)}]'_c$ has the same bounded sets as $E.$
\end{lemma}
\begin{proof}
Since $E$ is Mackey-complete bounded sets are included in absolutely convex closed bounded sets which are Banach disks. On $E'$ the topology $\mathcal{T}_{\mathcal{B}_{b}}$ of uniform convergence on Banach disk (bornology $\mathcal{B}_{b}$) coincides with the topology of the strong dual $E'_\beta$. 
 
 Moreover, by \cite[Th 10.1.2]{Jarchow} $\mathcal{B}_{b}$-Mackey convergent sequences coincide with $(\mathcal{B}_{b})_0$-Mackey convergent sequences  but the closed absolutely convex cover of a null sequence of a Banach space is compact in this Banach space, therefore compact in $E$, thus they coincide with $\varepsilon((E'_c)')$-null sequences (i.e. null sequences for Mackey convergence for the bornology of absolutely convex compact sets). Therefore $\mathscr{S}(E'_c)=\mathscr{S}(E',\mathcal{T}_{\mathcal{B}_{b}})=\mathscr{S}(E'_\beta).$ As a consequence, combining this with \cite[Th 13.3.2]{Jarchow}, the completion of $\mathscr{S}(E'_c)$ is linearly isomorphic to the dual of both the bornologification and the ultrabornologification of $E$. Therefore, the bounded sets in $(\widetilde{\mathscr{S}(E'_c)})'_c$ are by Mackey theorem the bounded sets for $\sigma(E,\widetilde{\mathscr{S}(E'_c)})=\sigma(E,(E_{bor})')$ namely the bounded  of $E_{bor}$ or $E$. 
 \end{proof}

The $\rho$-dual can be understood in a finer way. Indeed, the  Mackey-completion on $E'_\RR =\mathscr{S}(E'_{\mathscr{S} \rho})$ is unnecessary, as we would get a Mackey-complete space back after three dualization. 

\begin{proposition}
For any \lcs $E$,$$((E'_\RR)'_\RR)'_\RR \simeq \widehat{E'_\RR}^M \equiv E^*_\rho$$
and if $E$ Mackey-complete, $(E'_\RR)'_\RR=(E^*_\rho)^*_\rho.$
\end{proposition}

\begin{proof}
 We saw in step 3 of our theorem \ref{rhoref} that, for any Mackey-complete Schwartz space $E$, first $(E'_\RR)'_\RR$ is Mackey-complete hence $(E'_\RR)'_\RR=(E^*_\rho)^*_\rho$ and then \eqref{idmap} gives a continuous identity map $(( E'_\RR)'_\RR) \to E$. By functoriality one gets a continuous linear map:$E'_\RR \to (( E'_\RR)'_\RR)'_\RR.$ Moreover $(( E'_\RR)'_\RR)'_\RR= (( \widehat{E'_\RR}^M)'_\RR)'_\RR$ is Mackey-complete by step 2 of our previous theorem, thus the above map extends to $ \widehat{E'_\RR}^M\to (( E'_\RR)'_\RR)'_\RR.$ This is of course the inverse of the similar continuous (identity) map given by \eqref{idmap}: $(( \widehat{E'_\RR}^M)'_\RR)'_\RR \to \widehat{E'_\RR}^M$ which gives the topological identity. 
\end{proof}

We finally relate our definition with other previously known notions:

\begin{theorem}\label{MackeyCaractrhoRef}
A \lcs is $\rho$-reflexive, if and only if it is Mackey-complete, has its Schwartz topology associated to the Mackey topology of its dual $\mu_{(s)}(E,E')$ and its 
dual is also Mackey-complete with its Mackey topology. As a consequence, Arens=Mackey duals of $\rho$-reflexive spaces are exactly Mackey-complete locally convex spaces with their Mackey topology such that their Mackey dual is Mackey-complete.
\end{theorem}
\begin{remark}
A $k$-quasi-complete space is Mackey-complete hence for a  $k$-reflexive space $E$, $\mathscr{S}((E'_\mu)'_\mu)$ is $\rho$-reflexive (since $E'_c$ $k$-quasi-complete implies that so is $E'_\mu$ which is a stronger topology). Our new setting is a priori more general than the one of section 4. We will pay the price of a weaker notion of smooth maps. Note that a  Mackey-complete space need not be $k$-quasi-complete (see lemma \ref{McMSNotKc} below). 
\end{remark}

\begin{proof}
If $E$ is $\rho$-reflexive we saw in Theorem \ref{rhoref} that $E \simeq \mathscr{S}((E'_\mu)'_\mu)$ and both $E,E'_\RR=\mathscr{S}(E'_\mu)$ (or $E'_\mu$) are Mackey-complete with their Mackey topology.

Conversely, if $E$ with $\mu_{(s)}(E,E')$ is Mackey-complete as well as its dual, $E'_{\mathscr{S}\rho}=E'_c$ and thus $E'_\RR=\mathscr{S}(E'_c)$ which has the same bornology as the Mackey topology and is therefore Mackey-complete too, hence $E'_\RR=E^*_\rho$. Therefore we have a map $(E',\mu_{(s)}(E',E))\to E'_\RR=\mathscr{S}(E'_c).$ Conversely, note that  $E'_c=(E',\mu(E',E))$ from lemma \ref{MackeyArensSchwartz} so that one gets a continuous isomorphism.

 From the completeness and Schwartz property and dualisation, and then lemma \ref{MackeyArensSchwartz}  again, there is a continuous identity map $(E^*_\rho)'_\mathscr{R}=\mathscr{S}([\mathscr{S}(E'_c)]'_c)= \mathscr{S}((E'_\mu)'_\mu)=E$, which is Mackey-complete.  Therefore $(E^*_\rho)^*_\rho=(E'_\RR)'_\mathscr{R}=E$, i.e. $E$ is $\rho$-reflexive.
 
 For the last statement, we already saw the condition is necessary, it is sufficient since for $F$ Mackey-complete with its Mackey topology with Mackey-complete Mackey-dual, $\mathscr{S}(F)$ is $\rho$-reflexive by what we just saw and so that $(\mathscr{S}(F))'_c$ is the Mackey topology on $F'$, by symmetry  $\left[\mathscr{S}\left([\mathscr{S}(F)]'_c\right)\right]'_c=F$
 and therefore $F$ is both Mackey and Arens dual of the $\rho$-reflexive space $\mathscr{S}\left([\mathscr{S}(F)]'_c\right)$.\end{proof}
 
 Several relevant categories appeared.  $\M\subset\LCS$ the full subcategory of spaces having their Mackey topology. $\MS\subset\LCS$ the full subcategory of spaces having the Schwartz topology associated to its Mackey topology. $\Mb\subset\LCS$ the full subcategory of spaces with a Mackey-complete Mackey dual. And then by intersection always considered as full subcategories, one obtains: $$\McMS=\Mc\cap \MS,\ \ \MbMS=\Mb\cap \MS,,\ \ \McMb=\Mb\cap \Mc,$$ $$\MbMcM=\McMb\cap \M,\ \ \rRef=\McMb\cap \MS.$$
 We can summarize the situation as follows : There are two functors $(.)'_c$ and $\mu$  the associated Mackey topology (contravariant and covariant respectively) from the category $\rRef$ to $\MbMcM$ the category of Mackey duals of $\rho$-Reflexive spaces (according to the previous proposition). 
There are two other functors $(.)^*_\rho, \mathscr{S}$ and they are the (weak) inverses of the two previous ones.  

Finally, the following lemma explains that our new setting is more general than the $k$-quasi-complete setting of section 4:

\begin{lemma}\label{McMSNotKc}
There is a space $E\in \McMS$ which is not $k$-quasi-complete.
\end{lemma}

\begin{proof}We take $\K=\R$ (the complex case is similar).
Let $F=C^0([0,1])$ the Banach space with the topology of uniform convergence. We take  $G=\mathscr{S}(F'_\mu)=F'_c$ which is complete since $F$ ultrabornological \cite[Corol 13.2.6]{Jarchow}. Consider $H=Span\{\delta_x, x\in[0,1]\}$ the vector space generated by Dirac measures and $E= \widehat{H}^M$ the Mackey completion with induced topology (since we will see $E$ identifies as a subspace of $G$). Let $K$ be the unit ball of $F'$, the space of measures on $[0,1]$. It is  absolutely convex, closed for any topology compatible with duality, for instance in $G$, and since $G$ is a Schwartz space, it is precompact, and complete by completeness of $G$, hence compact. By Krein-Millman's theorem \cite[\S 25.1.4]{Kothe} it is the closed convex cover of its extreme points. Those are known to be $\delta_x, -\delta_x, x\in [0,1]$ \cite[\S 25.2.(2)]{Kothe}. Especially, $E$ is dense in $G$, which is therefore its completion. By the proof of lemma \ref{MMc}, the Mackey-topology of $E$ is induced by $G$ and thus by lemma \ref{SchwartzFunctor}, $\mathscr{S}(E_\mu)$ is also the induced topology from $G$. Hence $E\in  \McMS$. But by Maharam decomposition of measures, it is known that $F'$ has the following decomposition (see e.g. \cite[p 22]{Haydon}) as an $\ell^1$-direct sum:
$$F'=L^1(\{0,1\}^\omega)^{\oplus_1 2^\omega}\oplus_1 \ell^1([0,1])$$
and the Dirac masses generate part of the second component, so that $H\subset \ell^1([0,1])$ in the previous decomposition. But the bounded sets in $G$ are the same as in $F'_\beta$ (by principle of uniform boundedness), hence Mackey-convergence in $G$ implies norm convergence in $F'_\beta$, so that by completeness of $\ell^1([0,1])$, $E\subset  \ell^1([0,1])$. Hence Lebesgue measure (which gives one of the summands $L^1(\{0,1\}^\omega)$) gives $\lambda\not\in E$. Finally, consider $\delta:[0,1]\to K\subset G$ the dirac mass map. It is continuous since a compact set in $F$ is equicontinuous by Ascoli theorem, which gives exactly uniform continuity of $\delta$  on compact sets in $F$. Hence $\delta([0,1])$ is compact in $E$ while its absolutely convex cover in $G$ contains $\lambda$ so that the intersection with $E$ cannot be complete, hence $E$ is not $k$-quasi-complete.
\end{proof}

\subsection{Relation to projective limits and direct sums}
 \label{subsec:projlim}
 We now deduce the following stability properties from Theorem \ref{MackeyCaractrhoRef}.
 
\begin{corollary}\label{ProjrhoRef}
The class of $\rho$-reflexive spaces is stable by countable locally convex direct sums  and arbitrary products.
 \end{corollary}
 \begin{proof}
 
Let $(E_i)_{i\in I}$ a countable family of $\rho$-reflexive spaces,  and $E=\oplus_{i\in I}E_i$. 
{ Using Theorem \ref{MackeyCaractrhoRef}, we aim at proving that $E$ is Mackey-complete, has its Schwartz topology associated to the Mackey topology of its dual $\mu_{(s)}(E,E')$ and its 
dual is also Mackey-complete with its Mackey topology.}

From the Theorem \ref{MackeyCaractrhoRef}, $E_i$ itself has the Schwartz topology associated to its Mackey topology. From {\cite[\S 22. 5.(4)]{Kothe}}, the Mackey topology on $E$ is the direct sum of Mackey topologies. Moreover the maps $E_i\to \mathscr{S}(E_i)$ give a direct sum map $E\to \oplus_{i\in I}\mathscr{S}(E_i)$ and thus a continuous map $\mathscr{S}(E)\to \oplus_{i\in I}\mathscr{S}(E_i)$ since a countable direct sum of Schwartz spaces is a Schwartz space. Conversely the maps $E_i\to E$ give maps $\mathscr{S}(E_i)\to \mathscr{S}(E)$ and by the universal property this gives $\mathscr{S}(E)\simeq \oplus_{i\in I}\mathscr{S}(E_i)$. Therefore, if all spaces $E_i$ are $\rho$-reflexive, $E$ carries the Schwartz topology associated to its Mackey topology. From \cite[Th 2.14, 2.15]{KrieglMichor}, Mackey-complete spaces are stable by arbitrary projective limits and direct sums, thus the Mackey-completeness condition on the space and its dual {(using the computation of dual Mackey topology from \cite[\S 22. 5.(3)]{Kothe})} are also satisfied.

For an arbitrary product, {\cite[\S 22. 5.(3)]{Kothe}} again gives the Mackey topology, universal properties and stability of Schwartz spaces by arbitrary products give the commutation of $\mathscr{S}$ with arbitrary products and the stability of Mackey-completeness can be safely used (even for the dual, uncountable direct sum).

 \end{proof}

\begin{lemma}\label{projbidual}
For $(E_i,i\in I)$ a (projective) directed system of Mackey-complete Schwartz locally convex space if $E=\mathrm{proj}\ \lim_{i\in I} E_i$ , then : $$((E)^*_\rho)^*_\rho\simeq \Big[\Big[\mathrm{proj}\  \lim_{i\in I} ((E_i)^*_\rho)^*_\rho)\Big]^*_\rho\Big]^*_\rho.$$
 The same holds for general locally convex kernels and categorical limits.
\end{lemma}

 \begin{proof}
 The bidualization functor and universal property of projective limits give  maps $(E)^*_\rho)^*_\rho\to ((E_i)^*_\rho)^*_\rho$ and then $((E^*_\rho)^*_\rho\to \mathrm{proj}\  \lim_{i\in I} ((E_i)^*_\rho)^*_\rho,$ (see \cite[\S 19.6.(6)]{Kothe} for l.c. kernels) and bidualization and $\rho$-reflexivity concludes to the first map. Conversely, the canonical continuous linear map in the Mackey-complete Schwartz case  $((E_i)^*_\rho)^*_\rho)=((E_i)'_c)^*_\rho)\to E_i$ gives the reverse map after passing to the projective limit and double $\rho$-dual. The localy convex kernel case  and the categorical limit case are identical.
 \end{proof}

 \begin{proposition}\label{rhoRefComplete}
 The category $\rRef$ is complete and cocomplete, with products and countable direct sums agreeing  with those in $\LCS$ and limits given in lemma \ref{projbidual}
 \end{proposition}
 \begin{proof}
 Bidualazing after application of $\LCS$-(co)limits clearly gives (co)limits. Corollary \ref{ProjrhoRef} gives the product and sum case.
 \end{proof}

 \subsection{The Dialogue category $\McS$.}

We first deduce from Theorem \ref{FirstMALL} and a variant of \cite[Prop 2]{Schwartz} a useful:

\begin{lemma}
\label{McSSymmMonoidal2}
Let $\parr_{sb}$ be the $\parr$ of the complete $*$-autonomous category $\CSch$ given by $A\parr_{sb} B= \mathscr{S}(L_b((A)'_b,B)$. Then we have the equality in $\CSch$:
\begin{equation}\label{ReductionMcSCSch}\forall E,F\in \McS,\ \ \ \  E_{sc}\parr_{sb} F_{sc}=(E\varepsilon F)_{sc}.\end{equation}
As a consequence, $(\McS,\varepsilon,\K)$ 
 is a symmetric monoidal category.
\end{lemma}

\begin{proof}
We already know that $(\McS,\varepsilon,\K)$ is symmetric monoidal but we give an alternative proof using lemma \ref{CategoricStabMonoidal}.

All spaces $E,F$ are now in $\McS$.
Note that 
$(E_{sc})'_b$ is $\mathscr{S}(E'_c)$ with equicontinuous bornology, which is a Schwartz bornology, hence a continuous linear map from it to any $F$ sends a bounded set into a bipolar of a Mackey-null sequence for the absolutely convex compact bornology. Hence
$$ U(E_{sc} \parr_b F_{sc})=U(L_b((E_{sc})'_b,F_{sc}))=L_\epsilon(\mathscr{S}(E'_c),F)=L_\epsilon(E'_c,F)=E\varepsilon F$$
since the boundedness condition is satisfied hence equality as spaces, and the topology is the topology of convergence on equicontinuous sets, and the next-to-last equality since $F$ Schwartz. Moreover an equicontinuous set in $L_\epsilon(\mathscr{S}(E'_c),F)$ coincide with those in $L_\epsilon(E'_c,F)$ and an equibounded set in $L_b((E_{sc})'_b,F_{sc})$ only depends on the topology on $E$, hence in $\CLCS$: $$L_b((E_{sc})'_b,F_{sc})=L_b((E_{c})'_b,F_{sc})$$

Now in $\CSch$ we have $E_{sc}\parr_{sb}F_{sc}= \mathscr{S}L_b((E_{sc})'_b,F_{sc})= \mathscr{S}L_b((E_{c})'_b,F_{sc})$
and $F_{sc}=\mathscr{S}L_b((F_c)'_b,\K)$ hence \eqref{InnerAdjointS}
gives:$$E_{sc}\parr_{sb} F_{sc}= \mathscr{S}L_b((E_{c})'_b,\mathscr{S}L_b((F_c)'_b,\K))= \mathscr{S}L_b((E_{c})'_b,L_b((F_c)'_b,\K))=\mathscr{S}L_b((E_{c})'_b,F_{c})$$
 so that the bornology is the Schwartz bornology associated to the $\epsilon$-equicontinuous bornology of $E\varepsilon F$ (the one of $E_c\parr_b F_c$). It remains to identify this bornology with the one of $[E\varepsilon F]_{sc}$. Of course from this description the identity map 
$E_{sc}\parr_{sb} F_{sc}\to [E\varepsilon F]_{sc}$ is bounded, one must check the converse.

This is a variant of \cite[Prop 2]{Schwartz}.  Thus take an absolutely convex compact $K\subset E\varepsilon F=L(E'_c,F)$  and a sequence $\{x_n,n\in \N\}\subset (E\varepsilon F)_K$, with $||x_n||_K\to 0$. 
We must check it is Mackey-null in $E_c\parr_b F_c$.
For take as usual $\{y_n,n\in \N\}$ another sequence with  $||y_n||_K\to 0$ and $C=\{y_n,n\in \N\}^{oo}$ such that $||x_n||_C\to 0$. It suffices to check that $C$ is $\varepsilon$-equicontinous in $E\varepsilon F$, the bornology of $E_c\parr_b F_c$.

 For instance, one must show that for $A$ equicontinuous  in $E'$, $D=(C(A))^{oo}$ is absolutely convex compact in $F$ (and the similar symmetric condition).
 But since $E$ is Schwartz, it suffices to take $ A=\{z_n,n\in \N\}^{oo}$ with $z_n$ $\epsilon$-null in $E'$ and especially, Mackey-null. But $D\subset \{y_n(z_m), n,m\in \N\}^{oo}$ so that it suffices to see that $(y_n(z_m))_{n,m\in \N^2}$ is Mackey-null (since $F$ is Mackey-complete, this  will imply Mackey-null for the bornology of Banach disk, hence with compact bipolar). But from \cite[Prob 2bis p 28]{Schwartz} since $C$ is bounded in $E\varepsilon F$ it is $\varepsilon$-equihypocontinuous on $E'_\beta\times F'_\beta$ and hence it sends an equicontinuous set as $A$ to a bounded set in $F$, so that $D$ is bounded in $F$. Finally, $||(y_n(x_m))||_D\leq ||x_n||_A||y_m||_C$ hence the claimed Mackey-null property.

Let us prove again that $(\McS,\varepsilon,\K)$ is symmetric monoidal using lemma \ref{CategoricStabMonoidal} starting from $(\CSch,\parr_{sb},\K)$.
We apply it to the adjunction $(\cdot)_{sc}:\McS\to \CSch$ with left adjoint $\widehat{\cdot}^M\circ U$ using $U$ from Theorem \ref{FirstMALL}.(3).
The lemma concludes since the assumptions are easily satisfied, especially $E\varepsilon F=U([E\varepsilon F]_{sc})=\widehat{\cdot}^M\circ U([E\varepsilon F]_{sc})$ from stability of Mackey-completeness and using the key \eqref{ReductionMcSCSch}
\end{proof}

 We will now use lemma \ref{lemma:transport_dialoque_categories} to obtain a Dialogue category.
 \begin{proposition} \label{rhoDialogue}
 The negation $(.)^*_\rho$ gives $\McS^{op}$ the structure of a Dialogue category with tensor product $\varepsilon$.
\end{proposition} 
\begin{proof}
Proposition \ref{SMepsilon} or lemma \ref{McSSymmMonoidal2} gives $\McS^{op}$ the structure of a symmetric monoidal category. 
 We have to check that $(.)^*_\rho:\McS^{op}\to \McS$ is a tensorial negation on $\McS^{op}$.

For, we write it as a composition of functors involving $\CSch.$ Note that $(\cdot)_{sc}:  \McS\to \CSch$ the composition of inclusion and the functor of the same name in Theorem \ref{FirstMALL} is right adjoint to $L:=\ \widehat{\cdot}^M\circ U$ in combining this result with the proof of Proposition \ref{SMepsilon} giving the left adjoint to  $\McS\subset \Sch$. Then on $\McS$, $$(.)^*_\rho= \widehat{\cdot}^M\circ\mathscr{S}\circ  (\cdot)'_c= L\circ\mathscr{S}\circ (\cdot)'_b\circ (\cdot)_{c}=L\circ (\cdot)'_b\circ (\cdot)_{sc}.$$

Lemma \ref{lemma:transport_dialoque_categories} and the following remark concludes.
 \end{proof}

 \subsection{Commutation of the double negation monad on $\McS$} 
 
Tabareau shows in his theses \cite[Prop 2.9]{TabareauPhD} 
that if the continuation monad $\neg \neg $ of a Dialogue category is commutative and idempotent  then, the continuation category is $*$-autonomous. Actually, according to a result attributed to Hasegawa \cite{MelliesTabareau}, for which we don't have a published reference, it seems that idempotency and commutativity are equivalent in the above situation. This would simplify our developments since we chose our duality functor to ensure idempotency, but we don't use this second result in the sequel.

Thus we check $((\cdot)^*_{\rho})^*_{\rho}$ is a commutative monad. We deduce that from the study of a dual tensor product. Let us motivate its definition first.
 
As recalled in the preliminary section the $\varepsilon$-product is defined as $E\varepsilon F=(E'_c\otimes_{\beta e} F'_c)'$. Moreover, we saw in Theorem \ref{rhoref} that when $E$ is $\rho$-reflexive (or $E\in\MSch$) then $E=\mathscr{S}((E'_c)'_c)$. Recall also from \cite[10.4]{Jarchow} that a Schwartz space is endowed with the topology of uniform convergence on the $\epsilon$-null sequences of $E'$. 

Thus when $E\in\MSch$, its equicontinuous subsets $\varepsilon(E')$ are exactly the collection $\RR(E'_c)$  of all sets included in the closed absolutely convex cover of a $\varepsilon(((E'_c)'_c)')$-Mackey-null sequence.\footnote{Remember that a Mackey-null sequence is a sequence which Mackey-converges to $0$. By \cite[Prop 10.4.4]{Jarchow} any such Mackey-null sequence is an equicontinuous set : indeed the associated Schwartz topology is the topology of uniform convergence on those sequences and conversely using also the standard \cite[\S 21.3.(2)]{Kothe}.}

Remember also from section \ref{subsec:remindertvs} that every Arens-dual $E'_c$ is endowed with its $\gamma$-topology of uniform convergence on absolutely convex compact subsets of $E$. Thus if $\gamma(E'_c)$ is the bornology generated by absolutely convex compact sets,  the equicontinuous sets of $(((E'_c)'_c)')$ equals $\gamma(E'_c)$, as $E$ is always endowed with the topology of uniform convergence on equicontinuous subsets of $E'$. Thus $\RR(E'_c)= (\gamma(E'_c))_0$ is the bornology generated by bipolars of null sequences of  the bornology $\gamma(E'_c)$ (with the notation of \cite[subsection 10.1]{Jarchow}). We write in general $\RR(E)=(\gamma(E))_{0}$.

We call $\RR(E)$ the saturated bornology generated by $\gamma$-null sequences. Note that they are the same as null sequences for the bornology of Banach disks hence \cite[Th 8.4.4 b]{Jarchow} also for the bornology of absolutely convex weakly compact sets.

\begin{definition}
The $\RR$-tensor product $E\otimes _\RR F$ is the algebraic tensor product endowed with the finest locally convex topology making $E \times F \rightarrow E \otimes F$ a $(\RR(E)-\RR(F))$-hypocontinuous bilinear map.  We define $L_{\RR}(E,F)$ the space of continuous linear maps with the topology of convergence on $\mathscr{R}(E)$.
\end{definition}

Note that with the notation of Theorem \ref{FirstMALL}, for any $E,F\in\LCS$, this means $$E\otimes _\RR F=U(E_{sc}\o_H F_{sc}),\ \ \  L_{\mathscr{R}}(E,F)=U(L_b(E_{sc},F_{sc})).$$

 Pay attention $E'_\RR=L_{\mathscr{R}}(\mathscr{S}(\widehat{E}^M),\K)\neq L_{\mathscr{R}}(E,\K)$ in general,  
  which may not be the most obvious convention when $E\not\in\McS$.

For the reader's convenience, we spell out an adjunction motivating those definitions  even if we won't really use it.

\begin{lemma}
\label{univproperty} Let $E,F,G$ separated lcs.
If $F$ is a Schwartz space, so is $L_{\RR}(E,F).$ Moreover,  if we also assume $E\in \McMS$, then: $$L_{\RR}(E,F)\simeq E'_{\RR}\varepsilon F.$$ Finally if $E,F,G$ are Schwartz spaces and $F\in\MbMS$, then we have an algebraic isomorphism:
$$L(E\otimes_{\RR} F,G)\to L(E,L_{\RR}(F,G)).$$
\end{lemma}

\begin{proof}
For the Schwartz property, one uses \cite[Th 16.4.1]{Jarchow}, it suffices to note that $L_{\RR}(E,\K)$ is a Schwartz space and this comes from \cite[Prop 13.2.5]{Jarchow}. If $E$ is Mackey-complete Schwartz space with $E=\mathscr{S}((E'_\mu)'_\mu)$ then $(E'_{\RR})'_c= (E'_\mu)'_\mu$ and therefore $E'_{\RR}\varepsilon F=L((E'_{\RR})'_c,F)= L(E,F)$ and the topology is the one of convergence on equicontinuous sets, namely on $\RR(E)= \RR((E'_{\RR})'_c)$ since Mackey-null sequences coincides with $\gamma(E)$-null ones since $E$ Mackey-complete and thus does not depend on the topology with same dual.

Obviously, there is an injective linear map $$L(E\otimes_{\RR} F,G)\to L(E,L_{\mathscr{R}}(F,G))$$
Let us see it is surjective. For $f\in L(E,L_{\RR}(F,G))$ defines a separately continuous bilinear map and if $K\in \RR(F)$ the image $f(.)(K)$ is equicontinuous by definition. What is less obvious is the other equicontinuity. For $(x_n)_{n\geq 0}$ a $\gamma(E)$-null sequence, i.e null in $E_K$ for $K$ absolutely convex compact set, we want to show $\{f(x_n), n\geq 0\}$ equicontinuous, thus take $U^{\circ}$ in $G'$ an equicontinuous set, since $G$ is a Schwartz space, it is contained in the closed absolutely convex cover of a $\varepsilon(G')$-null sequence,   say $\{y_n,n\geq 0\}$ with $||y_n||_{V^{\circ}}\to 0$. $f(K)$ is compact thus bounded, thus $f(K)^t(V^{\circ})$ is bounded in  $L_{\mathscr{R}}(F,\K)=\mathscr{S}(F'_c)$ or in $F'_c$.  Thus $(f(x_n)^t(y_m))_{n,m}$ is Mackey-null in $F'_c$. Since $F=\mathscr{S}((F'_c)'_c)$, $F'_c= F'_\mu$ and as recalled earlier $\RR(F'_c)=\varepsilon(F')$. If moreover, $F'_c$ is Mackey-complete, 
$(f(x_n)^t(y_m))_{n,m}$ is Mackey for the bornology of Banach disks hence in $\RR(F'_c)$, thus it is equicontinuous in $F'$.
\end{proof}

We continue with two general lemmas deduced from lemma \ref{Prop2SchwartzCLCS}.
\begin{lemma}\label{TensorTech}
Let $X,Y\in \Sch$ and define $G=(X\varepsilon Y)'_\varepsilon$ the dual with the topology of convergence on equicontinuous sets from the duality with $H=X'_c\otimes_{\beta e} Y'_c$. Then we have embeddings $$H\subset G\subset \widehat{H}^M,\ ((H)'_\mu)'_\mu\subset ((G)'_\mu)'_\mu\subset ((\widehat{H}^M)'_\mu)'_\mu,\ \mathscr{S}(((H)'_\mu)'_\mu)\subset \mathscr{S}(((G)'_\mu)'_\mu)\subset\mathscr{S}(((\widehat{H}^M)'_\mu)'_\mu).$$ 
\end{lemma} 
\begin{proof}
We apply lemma \ref{Prop2SchwartzCLCS} to $(X_c)'_b,(Y_c)'_b$ which have a Schwartz bornology since  $X,Y\in \Sch$. Note that $H=U((X_c)'_b\o_H (Y_c)'_b)$ and that $$U((X_c)'_b\o_b (Y_c)'_b)=U\Big(\Big[((X_c)'_b)'_b\parr_b ((X_c)'_b)'_b\Big]'_b\Big)=U\Big(\Big[X_c\parr_b Y_c\Big]'_b\Big)=G.$$
 Lemma \ref{Prop2SchwartzCLCS} concludes exactly to the first embedding. The second follows using lemma \ref{MMc} and the third from lemma \ref{SchwartzFunctor}.
\end{proof}

\begin{lemma}\label{Tensor}
Let $X,Y\in \MSch$ and define $G=(X\varepsilon Y)'_\varepsilon$ the dual with the topology of convergence on equicontinuous sets from the duality with $H=X'_\mu\otimes_{\RR} Y'_\mu$. Then we have embeddings $$H\subset G\subset \widehat{H}^M,\ ((H)'_\mu)'_\mu)\subset ((G)'_\mu)'_\mu\subset ((\widehat{H}^M)'_\mu)'_\mu,\ \mathscr{S}(((H)'_\mu)'_\mu)\subset \mathscr{S}(((G)'_\mu)'_\mu)\subset\mathscr{S}(((\widehat{H}^M)'_\mu)'_\mu).$$ 
As a consequence for $X,Y\in \McMS$, we have topological identities $(X \varepsilon Y)^*_\rho \simeq \widehat{\mathscr{S}(H)}^M$ 
and $$((X\varepsilon Y)^*_\rho)^*_\rho\simeq (X'_c\otimes_\RR Y'_c)^*_\rho.$$
\end{lemma} 

\begin{proof} This is a special case of the previous result. Indeed since $X\in \MSch$ so that $X'_c=X'_\mu, X=\mathscr{S}((X'_c)'_\mu)=\mathscr{S}((X'_c)'_c)$, equicontinuous sets in its dual are those in $\RR(X'_c)=\RR(X'_\mu)$, hence:
$$(X\varepsilon Y)= (X'_c\o_{\beta e}Y'_c)'=(X'_c\o_{\RR}Y'_c)',\ \ \ \ \  X'_c\o_{\beta e}Y'_c\simeq X'_c\o_{\RR}Y'_c.$$

The Mackey-complete case is a reformulation using only the definition of $(\cdot)^*_\rho$ (and the commutation in Proposition \ref{SMepsilon}).
\end{proof}Let us state a consequence on
$X\otimes_{\kappa} Y:=((X\otimes_{\RR} Y)'_\mu)'_\mu\in \MLCS.$ We benefit from the work in lemma \ref{DualArensMc} that made appear the inductive tensor product.

\begin{proposition}\label{StarAutonomousParrTech}
 For any $X\in\Mb\cap \MLCS,Y\in \MLCS$,  then the canonical map is a topological isomorphism:
 \begin{equation}\label{tensTech}I:X\widehat{\otimes}_{\kappa}^MY\simeq X\widehat{\otimes}_{\kappa}^M (\widehat{Y}^M).
\end{equation}
 \end{proposition}
  \begin{proof}Let us write for short $\mathcal{F}=\mathscr{S}((\cdot)'_\mu), \mathcal{G}=\ \widehat{\cdot}^M\circ (\cdot)'_\mu$, $(\cdot)_\mu=((\cdot)'_\mu)'_\mu$. 
 Note that from the canonical continuous linear map $\mathcal{F}(X)\to \mathcal{G}(X)$ one deduces a continuous identity map $\mathcal{F}(\mathcal{G}(X))\to X=\mathcal{F}(\mathcal{F}(X)).$

Similarly, using lemma \ref{Tensor} for the equality, one gets by functoriality the continuous linear map:
 $$I:X\widehat{\otimes}_{\kappa} Y= \mathcal{G}\Big(\mathcal{F}(X)\varepsilon \mathcal{F}(Y)\Big)\to X\widehat{\otimes}_{\kappa} (\widehat{Y}^M)
.$$

For the converse, we apply lemma \ref{DualArensMc} to $L=\mathcal{F}(X)$ and 
$M=\mathcal{F}(Y)$, we know that  $[\mathcal{F}(X)\varepsilon \mathcal{F}(Y)]'_\epsilon=[\mathcal{F}(X)\eta \mathcal{F}(Y)]'_\epsilon$ induces on $L'_\mu\o M'_\mu$ the $\epsilon$-hypocontinuous tensor product. Using the reasoning of the previous lemma to identify the tensor product, this gives a continuous map 
 $$L'_\mu\o_{\beta e} M'_\mu=X_\mu\o_{\RR} Y_\mu\to [\mathcal{F}(X)\varepsilon \mathcal{F}(Y)]'_\epsilon\to X_\mu\widehat{\o}^M_{\RR} Y_\mu
 .$$

This gives by definition of hypocontinuity a  continuous linear map in $L(Y_\mu, L_{\RR}(X_\mu,X_\mu\widehat{\o}^M_{\RR} Y_\mu)
.$
Note that from the computation of equicontinuous sets and lemma \ref{MackeyArensSchwartz}, we have the topological identity:
$$L_{\RR}(X_\mu,X_\mu\widehat{\o}^M_{\RR} Y_\mu))=L_{\epsilon}((\mathscr{S}(X'_\mu))'_c,X_\mu\widehat{\o}^M_{\RR} Y_\mu))\simeq \mathscr{S}(X'_\mu)\varepsilon(X_\mu\widehat{\o}^M_{\RR} Y_\mu).$$

From this identity, one gets $L_{\RR}(X_\mu,X\widehat{\otimes}_{\kappa} Y))=\mathscr{S}(X'_\mu)\varepsilon (X_\mu\widehat{\o}^M_{\RR} Y_\mu)$  is Mackey-complete since $\mathscr{S}(X'_\mu)=\mathcal{F}(X)$ is supposed so and $X_\mu\widehat{\o}^M_{\RR} Y_\mu$ is  too by construction.

 As a consequence by functoriality of Mackey-completion, the map we started from has an extension  to $L(\widehat{Y_\mu}^M, \mathscr{S}(X'_\mu)\varepsilon (X_\mu\widehat{\o}^M_{\RR} Y_\mu))=L(\widehat{Y_\mu}^M, L_{\RR}(X_\mu,(X_\mu\widehat{\o}^M_{\RR} Y_\mu))).$
 A fortiori, this gives a separately continuous bilinear map and thus a continuous linear map extending the map we started from:
 $$J:X_\mu\o_i\widehat{Y_\mu}^M\to (X_\mu\widehat{\o}^M_{\RR} Y_\mu)$$

  We apply lemma \ref{DualArensMc} again  to $L=\mathcal{F}(X)$ and 
$M=\mathcal{F}(\widehat{Y}^M)$, we know that  $[\mathcal{F}(X)\varepsilon \mathcal{F}(\widehat{Y}^M)]'_\mu=[\mathcal{F}(X)\eta \mathcal{F}(\widehat{Y}^M))]'_\mu$ induces on $L'_\mu\o M'_\mu$ the inductive tensor product. Therefore, using also \cite[Corol 8.6.5]{Jarchow}, one gets a continuous linear map $$J:X_\mu\otimes_{i} [\widehat{Y_\mu}^M]_\mu
 \to \Big[X_\mu\widehat{\o}^M_{\RR} Y_\mu\Big]_\mu=X\widehat{\otimes}_{\kappa} Y
 .$$
 
In turn this maps extends to the Mackey completion $X_\mu\widehat{\otimes}_{i}^M \widehat{Y_\mu}^M=[X\widehat{\otimes}^M_{\kappa} (\widehat{Y}^M)]_\mu$ and our map $J:X_\mu\widehat{\otimes}_{\kappa} (\widehat{Y}^M)\to X\widehat{\otimes}_{\kappa}^M Y$ which is the expected inverse of $I$.
 
 \end{proof} 

\begin{corollary}\label{rhoCommut}
$T=((\cdot)^*_{\rho})^*_{\rho}$ is             a commutative monad on $(\McS^{op},\varepsilon,\K)$.
\end{corollary}
\begin{proof}
Fix $X,Y\in \McS$. Hence there is a continuous identity map $J_Y:(Y^*_\rho)^*_\rho\to Y.$ In order to build the strength, we use lemma \ref{DualArensMc} and $(Y^*_\rho)^*_\rho=(Y^*_\rho)'_\RR$ to get the
the identity $$(X\varepsilon ((Y^*_\rho)^*_\rho)^*_\rho=\mathcal{C}_M\Big(
\mathscr{S}\Big(\big[X\varepsilon ((Y^*_\rho)'_\RR)\big]'_\mu\Big)\Big)=
\mathscr{S}\Big(\big[X'_\mu\widehat{\o}^M_i [(Y^*_\rho)'_c]'_\mu\Big)=
\mathscr{S}\Big(X'_\mu\widehat{\o}^M_i (\widehat{Y'_\mu}^M)\Big)=\mathscr{S}\Big(X'_\mu\widehat{\otimes}_{\kappa}^M (\widehat{Y'_\mu}^M)\Big)$$
and similarly $(X\varepsilon Y)^*_\rho=\mathscr{S}\Big(X'_\mu\widehat{\otimes}_{\kappa}^M Y'_\mu\Big)$.

Hence applying proposition   \ref{StarAutonomousParrTech}
to $X'_\mu,Y'_\mu$ one gets that the canonical map is an isomorphism:

\noindent $(X\varepsilon Y)^*_\rho\to (X\varepsilon ((Y^*_\rho)^*_\rho)^*_\rho$
hence by duality the topological isomorphism:

\noindent $I_{X,Y}:((X\varepsilon ((Y^*_\rho)^*_\rho)^*_\rho)^*_\rho\simeq((X\varepsilon Y)^*_\rho)^*_\rho$
and we claim the expected strength is $$t_{X,Y}= J_{X\varepsilon T(Y)}\circ I_{X,Y}^{-1} \in \McS^{op}(X\varepsilon T(Y),T(X\varepsilon Y)).$$
Instead of checking the axioms directly, one uses that from \cite[Prop 2.4]{TabareauPhD}, the dialogue category already implies existence of a strength say $\tau_{X,Y}$ so that it suffices to see $\tau_{X,Y}=t_{X,Y}$ to get the relations for a strength for $t$. Of course we keep working in the opposite category.
From the axioms of a strength, see e.g.\cite[Def 1.19,(1.10),(1.12)]{TabareauPhD}, and of a monad, we know that $\tau_{X,Y}=J_{X\epsilon T(Y)}\circ T(\tau_{X,Y})\circ J_{T(X\epsilon Y)}^{-1}.$ Hence it suffices to see  $I_{X,Y}^{-1}= T(\tau_{X,Y})\circ J_{T(X\epsilon Y)}^{-1}$ or equivalenty $((I_{X,Y})^{*}_\rho)^{-1}=(J_{T(X\epsilon Y)}^{-1})^*_{\rho}\circ (\tau_{X,Y})^*_\rho.$ But recall that the left hand side is defined uniquely by continuous extension, hence it suffices to see the restriction agrees on $X'_\mu\o Y'_\mu$ and the common value is determined for both sides by axiom \cite[(1.12)]{TabareauPhD}.

Finally with our definition, the relation for a commutative monad ends with the map   $J_{T(X)\varepsilon T(Y)}$ and the map obtained after removing this map and taking dual of both sides is determined as a unique extension of the same map, hence the commutativity must be satisfied.
\end{proof}

\subsection{The $*$-autonomous category $\rRef$.}
\label{subsec:starautonomouscat}

\begin{definition}
We thus consider $\rRef$, the category of $\rho$-reflexive spaces, with tensor product $E \otimes_\rho F = ((E_{\mu} \otimes_\RR F_{\mu})^*_\rho)^*_\rho $ and internal hom $E \multimap_\rho F = (((E^*_\rho)\varepsilon F)^*_\rho)^*_\rho$ .
\end{definition}

{Recall $E_\mu=(E'_\mu)'_\mu$. For $E\in\rRef$ we deduce from lemma \ref{univproperty} that $E \multimap_\rho F \simeq ((L_{\RR}(E , F))^*_\rho)^*_\rho$.}

 The tensor product $E \otimes_\rho F$ is indeed a $\rho$-reflexive space by Theorem \ref{rhoref}.

We are ready to get that $\rRef$ is $*$-autonomous.

\begin{theorem}\label{StarAutonomous}
The category  $\rRef$ endowed with the tensor product $\otimes_\rho$, and internal Hom  $\multimap_\rho$ is a complete and cocomplete $*$-autonomous category with dualizing object $\K$. It is equivalent to the  Kleisli category of the comonad $T=((\cdot)^*_\rho)^*_\rho$ in $\McS.$
\end{theorem}

\begin{proof}Corollary \ref{rhoRefComplete} has already dealt with categorical (co)completeness.
$(\McS,\epsilon,K)$ is a dialogue category by proposition \ref{rhoDialogue} with a commutative and idempotent continuation monad by Corollary \ref{rhoCommut} and Theorem \ref{rhoref}.

The lemma \ref{DialogueToRef} gives $*$-autonomy.
As a consequence, the induced $\parr_\rho$ is $E\parr_\rho F= ((E\varepsilon F)^*_\rho)^*_\rho$ and the dual is still $(\cdot)^*_\rho$. The identification of $ \multimap_\rho$ is obvious while $\o_\rho$ comes from lemma \ref{Tensor}. 
\end{proof}

\part{Models of LL and DiLL}
\textbf{From now on, to really deal with smooth maps, we assume $\K=\R.$}
\section{Smooth maps and induced topologies. New models of LL}

Any denotational model of linear logic has a morphism interpreting dereliction on any space $E$ :  $d_E : E \to ?E$. In our context of smooth functions and reflexive spaces, it means that the topology on $E$ must be finer than the one induced by $\Cin (E^*, \K)$. From the model of $k$-reflexive spaces, we introduce a variety of new classes of smooth functions, each one inducing a different topology and a new smooth model of classical Linear Logic. We show in particular that each time the $\parr$ is interpreted as the $\varepsilon$-product.
 
We want to start from the famous Cartesian closedness \cite[Th 3.12]{KrieglMichor} and its corollary, but we want an exponential law in the topological setting, and not in the bornological setting. We thus change slightly the topology on (conveniently)-smooth maps $C^\infty(E,F)$ between two locally convex spaces. We follow the simple idea to consider spaces of smooth curves on a family of base spaces stable by product, thus at least on any $\R^n$. Since we choose at this stage a topology, it seems reasonable to look at the induced topology on linear maps, and singling out smooth varieties indexed by $\R^n$ does not seem to fit well with our Schwartz space setting for $\rho$-reflexive spaces, but rather with a stronger nuclear setting. This suggests that the topology on smooth maps could be a guide to the choice of a topology even on the dual space. In our previous developments, the key property for us was stability by $\varepsilon$ product of the topology we chose, namely the Schwartz topology. This property is shared by nuclearity but there are not many functorial and commonly studied topologies having this property. We think the Seely isomorphism is crucial to select such a topology in transforming stability by tensor product into stability by product.

\subsection{$\mathscr{C}$-Smooth maps and $\mathscr{C}$-completeness}

We first fix a small Cartesian category $\mathscr{C}$ that will replace the category of finite dimensional spaces $\R^n$ as parameter space of curves. 

 We will soon restrict to the full category $\FDFS\subset\LCS$ consisting of (finite) products of Fr\'echet spaces and strong duals of Fr\'echet-Schwartz spaces,  but we first explain the most general context in which we know our formalism works. We assume $\mathscr{C}$ is a full Cartesian small subcategory of $\kref$ containing $\R$,
  { with smooth maps as morphisms}. 

Proposition \ref{Meise4} and the convenient smoothness case suggests the following space and topology. For any $X\in \mathscr{C}$,
for any $c\in C^\infty_{co}(X,E)$ a ($\kref$ space parametrized) curve we define $C^\infty_{\mathscr{C}}(E,F)$ as the set of maps $f$ such that $f\circ c\in C^\infty_{co}(X,F)$ for any such curve $c$. We call them \textit{$\mathscr{C}$-smooth maps}.
{ Note that $.\circ c$ is in general not surjective, but valued in the closed subspace:$$[C^\infty_{co}(X,F)]_c=\{g\in C^\infty_{co}(X,F):\forall x\neq y: c(x)=c(y)\Rightarrow g(x)=g(y)\}.$$}
 One gets a linear map $.\circ c:C^\infty_{\mathscr{C}}(E,F)\to C^\infty_{co}(X,F)$. We equip the target space of the topology of uniform convergence of all differentials on compact subsets as before.  We equip $C^\infty_{\mathscr{C}}(E,F)$ with the projective {kernel} topology of those maps for all $X\in \mathscr{C}$ and $c$ smooth maps as above, with connecting maps all smooth maps  $C^\infty_{co}(X,Y)$ inducing reparametrizations. {Note that this projective kernel can be identified with a projective limit (indexed by a directed set). Indeed, we put an order on the set of curves $C^\infty_{co}(\mathscr{C},E):=\sqcup_{X\in\mathscr{C}} C^\infty_{co}(X,E)/\sim$ (where two curves are identified with the equivalence relation making the preorder we define into an order). This is an ordered set with  $c_1\leq c_2$ if $c_1\in C^\infty_{co}(X,E), c_2\in C^\infty_{co}(Y,E)$ and there is $f\in C^\infty_{co}(X,Y)$ such that $c_2\circ f=c_1$. This is moreover a directed set. Indeed given   $c_i\in C^\infty_{co}(X_i,E), $ one considers $c_i'\in C^\infty_{co}(X_i\times \R,E)$, $c_i'(x,t)=tc_i(x)$ so that $c_i'\circ (.,1)=c_i$ giving $c_i\leq c_i'$. Then one can define $c\in C^\infty_{co}(X_1\times \R\times X_2\times \R,E)$ given by $c(x,y)=c_1'(x)+c_2'(y)$. This satisfies $c\circ(.,0)=c_1',c\circ(0,.)=c_2'$, hence $c_i\leq c_i'\leq c$. We claim that $C^\infty_{\mathscr{C}}(E,F)$ identifies with { the projective limit along this directed set (we fix one $c$ in each equivalence class) of $[C^\infty_{co}(X,F)]_c$ on the curves $c\in C^\infty_{co}(X,E)$ with connecting maps for $c_1\leq c_2$, $.\circ f$ for one fixed $f$ such that $c_2\circ f=c_1$. This is well-defined since if $g$ is another curve with $c_2\circ g=c_1$, then for $u\in[C^\infty_{co}(X_2,F)]_{c_2}$ for any $x\in X_1$, $u\circ g(x)=u\circ f(x)$ since $c_2(g(x))=c_1(x)=c_2(f(x))$ hence $\cdot\circ g=\cdot\circ f:[C^\infty_{co}(X_2,F)]_{c_2}\to [C^\infty_{co}(X_1,F)]_{c_1}$ does not depend on the choice of $f$.}
  
 For a compatible sequence of such maps in $[C^\infty_{co}(X,F)]_c$, one associates the map $u:E\to F$ such that $u(x)$ is the value at the constant curve $c_x$ equal to $x$ in $C^\infty_{co}(\{0\},E)=E$. For, the curve $c\in C^\infty_{co}(X,E)$ satisfies for $x\in X$, $c\circ c_x=c_{c(x)}$, hence $u\circ c$ is the element of the sequence associated to $c$, hence $u\circ c\in [C^\infty_{co}(X,F)]_c.$ Since this is for any curve $c$, this implies $u\in C^\infty_{\mathscr{C}}(E,F)$ and the canonical map from this space to the projective limit is therefore surjective. The topological identity is easy.}

{We summarize this with the formula:
\begin{equation}\label{CinCasProjLimit}C^\infty_{\mathscr{C}}(E,F)=\mathrm{proj}\lim_{c\in C^\infty_{co}(X,E)}[C^\infty_{co}(X,F)]_c\end{equation}}

For $\mathscr{C}=Fin$ the category of finite dimensional spaces, $C^\infty_{Fin}(E,F)=C^\infty(E,F)$ is the space of conveniently smooth maps considered by Kriegl and Michor. We call them merely smooth maps. Note that our topology on this space is slightly stronger than theirs (before they bornologify) and that any $\mathscr{C}$-smooth map is smooth, since all our $\mathscr{C}\supset Fin$. Another important case for us is $\mathscr{C}=Ban$ the category of Banach spaces (say, to make it into a small category, of density character smaller than some fixed inaccessible cardinal, most of our considerations would be barely affected by taking the category of separable Banach spaces instead).

\begin{lemma}\label{thm:CartesianClosedBasic}
We fix $\mathscr{C}$ any  Cartesian small and full subcategory of $\mathbf{k-Ref}$ containing $\R$  and the above projective limit topology on $C^\infty_{\mathscr{C}}$.
For any $E,F,G$ lcs, with  $G$ $k$-quasi-complete, there is a topological isomorphism:
$$C^\infty_{\mathscr{C}}(E,C^\infty_{\mathscr{C}}(F,G))\simeq C^\infty_{\mathscr{C}}(E\times F,G)\simeq  C^\infty_{\mathscr{C}}(E\times F)\varepsilon G.$$

Moreover, the first isomorphism also holds for $G$ Mackey-complete, and $C^\infty_{\mathscr{C}}(F,G)$ is Mackey-complete (resp. $k$-quasi-complete) as soon as $G$ is. If $X\in \mathscr{C}$ then $C^\infty_{\mathscr{C}}(X,G)\simeq C^\infty_{co}(X,G)$ and if only $X\in \kref$ there is a continuous inclusion: $C^\infty_{co}(X,G)\to C^\infty_{\mathscr{C}}(X,G).$
\end{lemma}

\begin{proof}
The first algebraic isomorphism comes from \cite[Th 3.12]{KrieglMichor} in the case $\mathscr{C}=Fin$ (since maps smooth on smooth curves are automatically smooth when composed by ``smooth varieties" by their Corollary 3.13).
More generally, for any $\mathscr{C}$, the algebraic isomorphism works with the same proof in using Proposition \ref{Meise4} instead of their Proposition 3.10. We also use their notation $f^\vee, f^\wedge$ for the maps given by the algebraic Cartesian closedness isomorphism.

Concerning the topological identification we take the viewpoint of projective kernels, for any curve $c=(c_1,c_2):X\to E\times F$, one can associate a curve $(c_1 \times c_2):(X\times X)\to E\times F$, $(c_1 \times c_2)(x,y)=(c_1(x),c_2(y))$ and for $f\in C^\infty(E,C^\infty(F,G))$, one gets $(\cdot\circ c_2) (f\circ c_1)=f^\wedge\circ (c_1 \times c_2)$ composed with the diagonal embedding gives $f^\wedge\circ (c_1 , c_2)$ and thus uniform convergence of the latter is controlled by uniform convergence of the former. This gives by taking projective kernels, continuity of the direct map.

Conversely, for $f\in C^\infty(E\times F,G)$, $(\cdot\circ c_2)(f^\vee \circ c_1)=(f\circ (c_1\times c_2))^\vee $ with $c_1$ on $X_1, c_2$ on $X_2$ is controlled by a map $f\circ (c_1\times c_2)$ with $(c_1\times c_2):X_1\times X_2\to E\times F$ and this gives the converse continuous linear map (using proposition \ref{Meise4}).

The topological isomorphism with the $\varepsilon$ product comes from its commutation with projective limits as soon as we note that $[C^\infty_{co}(X,G)]_c=[C^\infty_{co}(X,\R)]_c \varepsilon G$ but these are also projective limits as intersections and kernels of evaluation maps. Therefore this comes from lemma \ref{etageneral} and from proposition \ref{Meise3}. 

Finally, $C^\infty_{\mathscr{C}}(F,G)$ is a closed subspace of a product of $C^\infty_{co}(X,G)$ which are  Mackey-complete or $k$-quasi-complete if  so is $G$  by proposition \ref{Meise2}
. 

For the last statement, since $id:X\to X$ is smooth, we have a continuous map $I:C^\infty_{\mathscr{C}}(X,G)\to C^\infty_{co}(X,G)$ in case $X\in \mathscr{C}$. Conversely, it suffices to note that for any $Y\in \mathscr{C}$, $c\in C^\infty_{co}(Y,X)$, $f\in C^\infty_{co}(X,G)$, then $f\circ c\in C^\infty_{co}(Y,G)$ by the chain rule from proposition \ref{Meise1} and that this map is continuous linear in $f$ for $c$ fixed. This shows $I$ is the identity map and gives continuity of its inverse by the universal property of the projective limit.
\end{proof}

We now want to extend this result beyond the case $G$ $k$-quasi-complete in finding the appropriate notion of completeness depending on $\mathscr{C}$.

\begin{lemma}\label{Ccomplete}
Consider the statements:
\begin{enumerate}
\item $F$ is Mackey-complete.
\item For any $X\in \mathscr{C}$, $J_X:C^\infty_{co}(X)\varepsilon F\to C^\infty_{co}(X,F)$ is a topological isomorphism
\item For any lcs $E$, $J_E^{\mathscr{C}}:C^\infty_{\mathscr{C}}(E)\varepsilon F\to C^\infty_{\mathscr{C}}(E,F)$ is a topological isomorphism.
\item For any $X\in \mathscr{C},f\in (C^\infty_{co}(X))'_c$, any $c\in C^\infty_{co}(X,F)\subset C^\infty_{co}(X,\tilde{F})=C^\infty_{co}(X)\varepsilon\tilde{F}$, we have $(f\varepsilon Id)(c)\in F$ instead of its completion (equivalently with its $k$-quasi-completion).
\end{enumerate}
We have equivalence of (2),(3) and (4) for any $\mathscr{C}$ Cartesian small and full subcategory of $\kref$ containing $\R$. They always imply (1) and when  $ \mathscr{C}\subset \FDFS,$ (1) is also equivalent to them.
\end{lemma}
This suggests the following condition weaker than $k$-quasi-completeness:
\begin{definition}
A locally convex space $E$ is said $\mathscr{C}$-complete (for a $\mathscr{C}$  as above
) if one of the equivalent conditions (2),(3),(4) are satisfied.
\end{definition}

This can be the basis to define a $\mathscr{C}$-completion similar to Mackey completion with a projective definition (as intersection in the completion) based on (2) and an inductive construction (as union of a chain in the completion) based on (4).

\begin{proof}
(2) implies (3) by the commutation of $\varepsilon$ product with projective limits as in lemma \ref{thm:CartesianClosedBasic} and (3) implies (2) using $C^\infty_{\mathscr{C}}(X,F)=C^\infty_{co}(X,F)$, for $X\in \mathscr{C}$.
(2) implies (4) is obvious since the map $(f\varepsilon Id)$ gives the same value when applied in  $C^\infty_{co}(X)\varepsilon F$. Conversely, looking at $u\in C^\infty_{co}(X,F)\subset C^\infty_{co}(X,\tilde{F})=L\big((C^\infty_{co}(X))'_c,\tilde{F}\big)$, (4) says that the image of the linear map $u$ is valued in $F$ instead of $\tilde{F}$, so that since continuity is induced, one gets $u\in L\big((C^\infty_{co}(X))'_c,F\big)$ which gives the missing surjectivity hence (2) (using some compatibility of $J_X$ for a space and its completion).

Let us assume (4) and prove (1). We use  a characterization of Mackey-completeness in \cite[Thm 2.14 (2)]{KrieglMichor}, we check that any smooth curve has an anti-derivative. As in their proof of (1) implies (2) we only need to check any smooth curve has a weak integral in $E$ (instead of the completion, in which it always exists uniquely by their lemma 2.5). But take $Leb_{[0,x]}\in (C^\infty_{co}(\R))'$, for a curve $c\in  C^\infty(\R,\tilde{F})$ it is easy to see that $(Leb_{[0,x]}\epsilon Id)(c)=\int_0^xc(s)ds$ is this integral (by commutation of both operations with application of elements of $F'$). Hence (4) gives exactly that this integral is in $F$ instead of its completion, as we wanted.

Let us show that (1) implies (2) first in the case $\mathscr{C}=Fin$  and take $X=\R^n$. One uses \cite[Thm 5.1.7]{FrolicherKriegel} which shows that $S=\textrm{Span}(ev_{\R^n}(\R^n))$ is Mackey-dense in $C^\infty(\R^n)'_c.$ But for any map $c\in C^\infty(\R^n,F)$, there is a unique possible value of $f\in L(C^\infty(\R^n)'_c,F)$ such that $J_X(f)=c$ once restricted to $\textrm{Span}(ev_{\R^n}(\R^n))$. Moreover $f\in L(C^\infty(\R^n)'_c,\tilde{F})$ exists and Mackey-continuity implies that the value on the Mackey-closure of $S$ lies in the Mackey closure of $F$ in the completion, which is $F$. This gives surjectivity of $J_X$.

In the case $X\in \mathscr{C}\subset \FDFS$, it suffices to show that $S=\textrm{Span}(\cup_{k\in\N}ev_{X}^{(k)}(X^{k+1}))$ is Mackey dense in $C^\infty(X)'_c.$ Indeed, one can then reason similarly since for $c\in C^\infty_{co}(X,F)$ and $f\in L(C^\infty_{co}(X)'_c,\tilde{F})$ with  $J_X(f)=c$ satisfies $f\circ ev_{X}^{(k)}=c^{(k)}$ which takes value in $F$ by convenient smoothness and Mackey-completeness, hence also Mackey limits so that $f$ will be valued in $F$. Let us prove the claimed density. First recall that $C^\infty_{co}(X)$ is a projective kernel of spaces $C^0(K,(X'_c)^{\epsilon k})$ via maps induced by differentials and this space is itself a projective kernel of $C^0(K\times L^k)$ for absolutely convex compact sets $K,L\subset X$. Hence by \cite[\S 22.6.(3)]{Kothe}, $(C^\infty_{co}(X))'$ is a locally convex hull (at least a quotient of a sum) of the space of signed measures $(C^0(K\times L^k))'.$ As recalled in the proof of \cite[Corol 13 p 279]{Meise}, every compact set $K$ in  $X\in \FDFS$ is a compact subset of a Banach space, hence metrizable. Hence the space of measure signed measures $(C^0(K\times L^k))'$ is metrizable too for the  weak-* topology (see e.g. \cite{DellacherieMeyer}), and by Krein-Millman's Theorem    \cite[\S 25.1.(3)]{Kothe} every point in the (compact) unit ball is  a weak-* limit of an absolutely convex combination of extreme points, namely Dirac masses \cite[\S 25.2.(2)]{Kothe}, and by metrizability one can take a sequence of such combinations, which is bounded in  $(C^0(K\times L^k))'$. Hence its image in $E=(C^\infty_{co}(X))'$ is bounded in some Banach subspace, with equicontinuous ball $B$ (by image of an equicontinuous sets, a ball in a Banach space by the transpose of a continuous map)  and converges weakly. But from \cite[Prop 11 p 276]{Meise}, $C^\infty_{co}(X)$ is a Schwartz space, hence there is an other equicontinuous set $C$ such that $B$ is compact in $E_C $ hence the weakly convergent sequence admitting only at most one limit point must converge normwise in $E_C$. Finally, we have obtained Mackey convergence of this sequence in $E=(C^\infty_{co}(X))'$ and looking at its form, this gives exactly Mackey-density of $S$.
\end{proof}

\subsection{Induced topologies on linear maps}

In the setting of the previous subsection, $E'\subset C^\infty_{\mathscr{C}}(E,\R)$. From Mackey-completeness, this extends to an inclusion of the Mackey completion, on which one obtains an induced topology which coincides with the topology of uniform convergence on images by smooth curves with source $X\in\mathscr{C}$ of compacts in this space. Indeed, the differentials of the smooth curve is also smooth on a product and the condition on derivatives therefore reduces to this one. This can be described functorially in the spirit of $\mathscr{S}$.

We first consider $\mathscr{C}\subset\kref$ a full Cartesian subcategory.

Let $\mathscr{C}^\infty$ be the smallest class of locally convex spaces  containing  $C^\infty_{co}(X,\K)$ for $X\in \mathscr{C}$ ($X=\{0\}$ included) and stable by products and subspaces. Let $\mathscr{S}_{\mathscr{C}}$ the functor on $\LCS$ of associated topology in this class described by \cite[2.6.4]{Junek}. This functor commutes with products. 

{\begin{ex}
If $\mathscr{C}=\{0 \}$ then $\mathscr{C}^\infty=\mathbf{Weak}$ the category of spaces with their weak topology, since  $\K$ is a universal generator for  spaces with their weak topology.
Thus  the weak topology functor is $\mathscr{S}_{\{0\}}(E).$
\end{ex}}

\begin{ex}
If $\mathscr{C}^\infty\subset \mathscr{D}^\infty$ (e.g. if $\mathscr{C}\subset \mathscr{D}$) then, from the very definition, there is a natural transformation $id\to \mathscr{S}_{\mathscr{D}}\to \mathscr{S}_{\mathscr{C}}$ with each map $E\to \mathscr{S}_{\mathscr{D}}(E)\to \mathscr{S}_{\mathscr{C}}(E)$ is a continuous identity map.
\end{ex}
\begin{lemma}\label{SCdual}
For any \lcs E, $(\mathscr{S}_\mathscr{C}(E))' = E'$ algebraically.
\end{lemma}

\begin{proof}
Since $\{0\}\subset \mathscr{C}$, there is a continuous identity map $E\to \mathscr{S}_\mathscr{C}(E)\to \mathscr{S}_{\{0\}}(E)=(E'_\sigma)'_\sigma.$ The Mackey-Arens theorem concludes.  
\end{proof}

As a consequence, $E$ and  $\mathscr{S}_\mathscr{C}(E)$ have the same bounded sets and therefore are simultaneously Mackey-complete.
Hence $\mathscr{S}_\mathscr{C}$ commutes with Mackey-completion. Moreover, the class $\mathscr{C}^\infty$ is also stable by $\varepsilon$-product, since this product commutes with projective kernels and $C^\infty_{co}(X,\K)\varepsilon C^\infty_{co}(Y,\K)=C^\infty_{co}(X\times Y,\K)$ and we assumed $X\times Y\in \mathscr{C}.$

We now consider the setting of the previous subsection, namely we also assume $\R\in \mathscr{C}$, $\mathscr{C}$ small and identify the induced topology  $E'_{\mathscr{C}}\subset C^\infty_{\mathscr{C}}(E,\R)$.

\begin{lemma}\label{InducedTopoGene}
For any \lcs E, there is a continuous identity map: $E'_{\mathscr{C}}\to \mathscr{S}_\mathscr{C}(E'_c)$.

If moreover $E$ is $\mathscr{C}$-complete, this is a topological isomorphism.
\end{lemma}
\begin{proof}
For the direct map we use the universal property of projective kernels. Consider a continuous linear map $f\in L(E'_c, C^\infty_{co}(X,\K))=C^\infty_{co}(X,\K)\varepsilon E$ and the corresponding $J_X(f)\in C^\infty_{co}(X,E)$, then by definition of the topology $ \cdot \circ J_X(f):C^\infty_{\mathscr{C}}(X,E)\to C^\infty_{co}(X,\R)$ is continuous and by definition, its restriction to $E'$ agrees with $f$, hence $f:E'_{\mathscr{C}}\to C^\infty_{co}(X,\K)$ is also continuous. Taking a projective kernel over all those maps
gives the expected continuity.

Conversely, if  $E$ is $\mathscr{C}$-complete, note that $E'_c\to E'_{\mathscr{C}}$ is continuous using again the universal property of a kernel, it suffices to see that for any $X\in\mathscr{C}, c\in C^\infty_{\mathscr{C}}(X,E)$ then $\cdot\circ c:E'_c\to C^\infty_{\mathscr{C}}(X,K)$ is continuous, and this is the content of the surjectivity of $J_X$ in lemma \ref{Ccomplete} (2) since $\cdot\circ c=J_X^{-1}(c)$. Hence since $E'_{\mathscr{C}}\in\mathscr{C}^\infty$ by definition as projective limit, one gets by functoriality the continuity of $\mathscr{S}_\mathscr{C}(E'_c)\to E'_{\mathscr{C}}.$
\end{proof}

We are going to give more examples in a more restricted context.
{We now fix $Fin\subset\mathscr{C}\subset \FDFS.$ But the reader may assume $\mathscr{C}\subset Ban$ if he or she wants, our case is not such more general. Note that then $C^\infty_{\mathscr{C}}(E,F)=C^\infty(E,F)$ algebraically. For it suffices to see $C^\infty_{co}(X,F)=C^\infty(X,F)$ for any $X\in\FDFS$ (since then the extra smoothness condition will be implied by convenient smoothness). Note that any such $X$ is ultrabornological (using \cite[Corol 13.2.4]{Jarchow},  \cite[Corol 13.4.4,5]{Jarchow} since a DFS space is reflexive hence its strong dual is barrelled \cite[Prop 11.4.1]{Jarchow} and for a dual of a Fr\'echet space, the quoted result implies it is also ultrabornological, for products this is \cite[Thm 13.5.3]{Jarchow}). By Cartesian closedness of both sides this reduces to two cases. For any Fr\'echet space $X$, Fr\'echet smooth maps  are included in $C^\infty_{co}(X,F)$ which is included in $C^\infty(X,F)$ which coincides with the first space of Fr\'echet smooth maps  by \cite[Th 4.11.(1)]{KrieglMichor} (which ensures the continuity of Gateaux derivatives with value in bounded linear maps with strong topology for derivatives, those maps being the same as continuous linear maps as seen the bornological property). The  case of strong duals of Fr\'echet-Schwartz spaces is similar using \cite[Th 4.11.(2)]{KrieglMichor}.
 The index $\mathscr{C}$ in $C^\infty_{\mathscr{C}}(E,F)$ remains to point out the different topologies.

\begin{ex}
If $\mathscr{C}=\FDFS$ (say with objects of density character smaller than some inaccessible cardinal)
 then $\mathscr{C}^\infty\subset\mathbf{Sch}$,  from \cite[Corol 13 p 279]{Meise}. Let us see equality. Indeed, $(\ell^1(\N))'_c\subset C^\infty_{co}(\ell^1(\N),\K)$ and 
 $(\ell^1(\N))'_c=(\ell^1(\N))'_\mu$ (since on $\ell^1(\N)$ compact and weakly compact sets coincide \cite[p 37]{HogbeNlendMoscatelli}), and $(\ell^1(\N))'_\mu$ is a universal generator of Schwartz spaces \cite[Corol p 36]{HogbeNlendMoscatelli}, therefore $C^\infty_{co}(\ell^1(\N),\K)$ is also such a universal generator. Hence we even have $\mathscr{C}^\infty=Ban^\infty=\Sch$. Let us deduce even more of such type of equalities.
 
Note also that $Sym(E'_c\varepsilon  E'_c)\subset C^\infty_{co}(E,\K)
$ is a complemented subspace given by quadratic forms. In case $E=H$ is an infinite dimensional Hilbert space, by Buchwalter's theorem $H'_c\varepsilon  H'_c=(H\widehat {\o}_\pi H)'_c$ and it is well-known that $\ell^1(\N)\simeq D$ is a complemented subspace (therefore a quotient) of $H\widehat {\o}_\pi H$
as diagonal copy (see e.g. \cite[ex 2.10]{Ryan}) with the projection a symmetric map. Thus $D'_c\subset H'_c\varepsilon  H'_c$ and it is easy to see it is included in the symmetric part $Sym(E'_c\varepsilon  E'_c)$. As a consequence, $C^\infty_{co}(H,\K)$ is also such a universal generator of Schwartz spaces. 

Finally, consider $E=\ell^{m}(\N,\C)$ $m\in \N, m\geq 1$. 
The canonical multiplication map from Holder $\ell^{m}(\N,\C)^{\o_\pi m}\to \ell^{1}(\N,\C)$ is a metric surjection realizing the target as a quotient of the symmetric subspace generated by tensor powers (indeed $\sum a_k e_k$ is the image of $(\sum a_k^{1/m} e_k)^{\o m}$ so that 
 $(\ell^{1}(\N,\C))'_c\subset Sym([(\ell^{m}(\N,\C))'_c]^{\varepsilon m})$. Thus $C^\infty_{co}(\ell^{m}(\N,\C),\K)$ is also such a universal generator of Schwartz spaces.

 We actually checked that for any $\mathscr{C}\subset \FDFS$ with $\ell^1(\N)\in \mathscr{C}$ or $\ell^2(\N)\in \mathscr{C}$ or $\ell^{m}(\N,\C)\in \mathscr{C}$ then $\mathscr{C}^\infty=\Sch$   so that $$\mathscr{S}=\mathscr{S}_{Ban}=\mathscr{S}_{Hilb}=\mathscr{S}_{\mathscr{C}}=\mathscr{S}_{\FDFS}.$$
\end{ex}
As a consequence, we can improve slightly our previous results in this context:
\begin{lemma}\label{InducedTop} Let $\mathscr{C}\subset \FDFS$ as above.
For any \lcs E, there is a continuous identity map: $E'_{\mathscr{C}}\to \mathscr{S}_\mathscr{C}(E'_\mu)\to \mathscr{S}_\mathscr{C}(E'_c)$.
If moreover $E$ is Mackey-complete, this is a topological isomorphism. 
\end{lemma}
\begin{proof}
Indeed by definition $\mathscr{S}_{\mathscr{C}}(E'_\mu)$ is described by a projective limit over maps $L(E'_\mu, C^\infty_{co}(X,\K))=E\eta C^\infty_{co}(X,\K)=C^\infty_{co}(X,\K)\varepsilon E\subset C^\infty_{co}(X,E)$ by the Schwartz property. 
As in lemma \ref{InducedTopoGene}, the identity map $E'_{\mathscr{C}}\to \mathscr{S}_{\mathscr{C}}(E'_\mu)$.
But by functoriality one has also a continuous identity map $\mathscr{S}_{\mathscr{C}}(E'_\mu)\to \mathscr{S}_{\mathscr{C}}(E'_c)$ and in the Mackey-complete case $\mathscr{S}_{\mathscr{C}}(E'_c)\to E'_{\mathscr{C}}$ by lemma \ref{InducedTopoGene}. (This uses that Mackey-complete implies $\mathscr{C}$-complete in our case by the last statement in lemma \ref{Ccomplete}).
\end{proof}

\begin{ex}Note also that if $D$ is a quotient with quotient topology of a Fr\'echet space $C$ with respect to a closed subspace, then  $C^\infty_{co}(D,\K)$ is a subspace of $C^\infty_{co}(C,\K)$ with induced topology. Indeed, the injection is obvious and derivatives agree, and since from \cite[\S 22.3.(7)]{Kothe}, compacts are quotients of compacts, the topology is indeed induced.
Therefore if   $\mathscr{D}$ is obtained from $\mathscr{C}\subset Fre$, the category of Fr\'echet spaces, by taking all quotients by closed subspaces, then $\mathscr{C}^\infty=\mathscr{D}^\infty.$
\end{ex}

\begin{ex}
If $\mathscr{C}=Fin$ then $Fin^\infty=\mathbf{Nuc}$, since  $C^\infty_{co}(\R^n,\K)\simeq \mathfrak{s}^{\N}$ \cite[(7) p 383]{Valdivia}, a countable direct product of classical sequence space $\mathfrak{s}$, which is a universal generator for  nuclear spaces.
Thus,  the associated nuclear topology functor is $\mathscr{N}(E)=\mathscr{S}_{Fin}(E).$
\end{ex}

We now provide several more advanced examples which will enable us to prove that we obtain different comonads in several of our models of $LL$. They are all based on the important approximation property of Grothendieck.

\begin{ex}\label{Ex:notAP}
If $E$ a Fr\'echet space without the approximation property (in short AP, for instance $E=B(H)$ the space of bounded operators on a Hilbert space), then from \cite[Thm 7 p 293]{Meise}, $C^\infty_{co}(E)$ does not have the approximation property. Actually, $E'_c\subset C^\infty_{co}(E)$ is a continuously complemented subspace so that so is $((E'_c)^*_\rho)^*_\rho\subset ((C^\infty_{co}(E))^*_\rho)^*_\rho$. But for any Banach space $E'_c=\mathscr{S}(E'_\mu)$ is Mackey-complete so that $(E'_c)^*_\rho=\mathscr{S}(E)$, $((E'_c)^*_\rho)^*_\rho=E^*_\rho=E'_c=\mathscr{S}(E'_\mu).$ Thus since for a Banach space $E$ has the approximation property if and only if $\mathscr{S}(E'_\mu)$ has it \cite[Thm 18.3.1]{Jarchow}, one deduces that $((C^\infty_{co}(E))^*_\rho)^*_\rho$ does not have the approximation property \cite[Prop 18.2.3]{Jarchow}.
\end{ex}

\begin{remark}We will see in appendix in lemma \ref{WhyNotAP} that for any \lcs $E$, $((C^\infty_{Fin}(E))^*_\rho)^*_\rho$ is  Hilbertianizable, hence it has the approximation property.
This implies that $\mathscr{N}(E'_\mu)\subset C^\infty_{Fin}(E)$ with induced topology is not complemented, as soon as  $E$ is Banach space without AP, since otherwise $((\mathscr{N}(E'_\mu))^*_\rho)^*_\rho\subset ((C^\infty_{Fin}(E))^*_\rho)^*_\rho$ would be complemented and $((\mathscr{N}(E'_\mu))^*_\rho)^*_\rho=((E'_c)^*_\rho)^*_\rho=E'_c$ would have the approximation property, and this may not be the case. This points out that the change to a different class of smooth function in the next section is necessary to obtain certain models of DiLL. Otherwise, the differential that would give such a complementation cannot be continuous.
\end{remark}

We define $E^*_{\mathscr{C}}$ for $E\in\McS$ as the Mackey completion of $\mathscr{S}_{\mathscr{C}}((\widehat{E}^M)'_\mu)$, i.e. since $\mathscr{S}_{\mathscr{C}}\mathscr{S}=\mathscr{S}_{\mathscr{C}}$:
$$E^*_{\mathscr{C}}=\mathscr{S}_{\mathscr{C}}(E^*_\rho).$$
\begin{proposition}\label{Cdialogue}Let $Fin\subset\mathscr{C}\subset\FDFS$ a small and full Cartesian subcategory. The full subcategory $\mathscr{C}-\Mc\subset \McS$ of objects satisfying $E=\mathscr{S}_\mathscr{C}(E)$ is reflective of reflector $ \mathscr{S}_\mathscr{C}$. $(\mathscr{C}-\Mc^{op},\varepsilon,\K,(\cdot)^*_{\mathscr{C}})$ is a Dialogue category.
\end{proposition}

\begin{proof}
Since $E\to \mathscr{S}_\mathscr{C}(E)$ is the continuous identity map, the first statement about the reflector is obvious. $\mathscr{C}-\Mc$ is                                                                                                  
 stable by $\varepsilon$-product since                               
 $\mathscr{S}_\mathscr{C}(E)\varepsilon\mathscr{S}_\mathscr{C}(F)$ is a projective kernel of $C^\infty_{co}(X)\varepsilon C^\infty_{co}(Y)=C^\infty_{co}(X\times Y)\in \Cin.$
  We use Proposition \ref{rhoDialogue} to get $(\McS,\epsilon,\K,(\cdot)^*_{\rho})$.   One can apply Lemma \ref{lemma:transport_dialoque_categories}  since we have $\mathscr{S}_\mathscr{C}\circ (\cdot)^*_{\rho}=(\cdot)^*_{\mathscr{C}}$ and  $I:\mathscr{C}-\Mc\subset \McS$ satisfies $\mathscr{S}_\mathscr{C}(I(E\varepsilon F))=I(E\varepsilon F)=I(E)\varepsilon I( F).$ This concludes.
\end{proof}
}

\subsection{A general construction for LL models}
We used intensively Dialogue categories from \cite{MelliesTabareau,TabareauPhD} to obtain $*$-autonomous categories, but their notion of models of tensor logic is less fit for our purposes since the Cartesian category they use need not be Cartesian closed. For us  trying to check their conditions involving an adjunction at the level   of the Dialogue category would imply introducing a non-natural category of smooth maps while we have already a good Cartesian closed category. Therefore we propose a variant of their definition using relative adjunctions \cite{Ulmer}.

\begin{definition}\label{lambdaTensorDef}
A linear (resp. and commutative) categorical model of $\lambda$-tensor logic is 
a complete and cocomplete dialogue category $(\mathcal{C}^{op},\parr_\mathcal{C},I,\neg)$  with a (resp. commutative and idempotent) continuation monad $T=\neg\neg$, jointly with a Cartesian category $(\mathcal{M},\times,0)$, a symmetric strongly monoidal functor $NL:\mathcal{M}\to \mathcal{C}^{op}$ having a right $\neg$-relative adjoint $U$. The model is said to be a \textit{Seely} model if $U$ is bijective on objects. 
\end{definition}

This definition is convenient for its concision, but it does not emphasize that $\mathcal{M}$ must be Cartesian closed. Since our primitive objects are functional, we will prefer an  equivalent alternative based on the two relations we started to show in lemma \ref{thm:CartesianClosedBasic}, namely an enriched adjointness of Cartesian closedness and a compatibility with $\parr$.

\begin{definition}\label{LambdaModel}
A (resp. commutative) $\lambda$-categorical model of $\lambda$-tensor logic is 
a complete and cocomplete dialogue category $(\mathcal{C}^{op},\parr_\mathcal{C},1_\mathcal{C}=\K,\neg)$  with a (resp. commutative and idempotent) continuation monad $T=\neg\neg$, jointly with a Cartesian closed category $(\mathcal{M},\times,0,[\cdot,\cdot])$, and a 
functor $NL:\mathcal{M}\to \mathcal{C}^{op}$ having a right $\neg$-relative adjoint $U$, which is assumed  faithful, and compatibility natural isomorphisms in $\mathcal{M},\mathcal{C}$ respectively:  $$\Xi_{E,F}: U(NL(E)\parr_\mathcal{C} F)\to [E,U(F)], \Lambda^{-1}_{E,F,G}: NL(E)\parr_\mathcal{C}\Big( NL(F)\parr_\mathcal{C}G\Big)\to NL(E\times F)\parr_\mathcal{C}G$$
satisfying the following six commutative diagrams (where $Ass^\parr,\rho,\lambda,\sigma^\parr$ are associator, right and left unitors and braiding in $\mathcal{C}^{op}$ and $\Lambda^{\mathcal{M}},\sigma^\times,\ell,r$ are the curry map, braiding and unitors in the Cartesian closed category $\mathcal{M}$) expressing an intertwining between curry maps:
\[
\xymatrix@C=50pt{
U\Big(NL(E)\parr_\mathcal{C} \Big(NL(F)\parr_\mathcal{C} G\Big)\Big) \ar[r]^{\Xi_{E,NL(F)\parr_\mathcal{C} G}} &  [E,U(NL(F)\parr_\mathcal{C} G)]\ar[r]^{[id_{E},\Xi_{F,G}]} &[E,[F,U(G)]] \\
U\Big(NL(E\times F)\parr_\mathcal{C}G\Big)\ar[u]_{U(\Lambda_{E,F,G})} \ar[rr]^{\Xi_{E\times F, G}}&&[E\times F,U(G)]\ar[u]^{\Lambda^{\mathcal{M}}}
}
\]
{compatibility of $\Xi$ with the (relative) adjunctions  (written $\simeq$ and $\varphi$, the characteristic isomorphism of the dialogue category $\mathcal{C}^{op}$):
\[
\xymatrix@C=40pt{
 \mathcal{M}(0,U\Big(NL(E)\parr_\mathcal{C}F\Big)) \ar[r]^{\simeq} & \mathcal{C}(\neg(NL(0)),NL(E)\parr_\mathcal{C}F) \ar[r]^{\varphi_{NL(E),F,NL(0)}^{op}} & \mathcal{C}(\neg(F\parr_\mathcal{C}NL(0)),NL(E)) \\
M(U\Big(NL(E)\parr_\mathcal{C}F\Big))\ar[u]_{M(I_\mathcal{M})} \ar[rr]^{M(\Xi_{E,F})}&&M([E,U(F)])=\mathcal{M}(E,U(F))\ar[u]^{\simeq}
}
\]}
compatibility with associativity:
\[
\xymatrix@C=40pt{
  NL(E)\parr_\mathcal{C} \Big(NL(F)\parr_\mathcal{C} G\Big)\ar[d]|{Ass^{\parr}_{NL(E),NL(F),G}}\ar[r]^{\ \ \ \ \Lambda^{-1}_{E,F,G}}  &  NL(E\times F)\parr_\mathcal{C} G\ar[r]^{\rho_{NL(E\times F)}\parr_\mathcal{C} G\ \ \ \ \ \ \ }
&\Big(NL(E\times F)\parr_\mathcal{C} \K\Big)\parr_\mathcal{C} G\\
\Big(NL(E)\parr_\mathcal{C}NL(F)\Big)\parr_\mathcal{C} G\ar[rr]^{(NL(E)\parr\ \rho_{NL(F)})\parr\   G} & &\Big(NL(E)\parr_\mathcal{C}(NL(F)\parr_\mathcal{C}\K)\Big)\parr_\mathcal{C} G\ar[u]|{\Lambda^{-1}_{E,F,\K}\parr_\mathcal{C}G } }
\]
compatibility with symmetry,
\[
\xymatrix@C=40pt{
 NL(E)\parr_\mathcal{C} NL(F)\ar[d]|{\sigma^{\parr}_{NL(E),NL(F)}}\ar[rr]^{NL(E)\parr\ \rho_{NL(F)}\quad}  & &  NL(E)\parr_\mathcal{C} \Big(NL(F)\parr_\mathcal{C} \K\Big)\ar[r]^{\ \ \ \ \Lambda^{-1}_{E,F,\K}}  &  NL(E\times F)\parr_\mathcal{C} \K\ar[d]|{NL(\sigma_{E, F}^\times)\parr_\mathcal{C} \K}\\
NL(F)\parr_\mathcal{C} NL(E)\ar[rr]^{NL(F)\parr\ \rho_{NL(E)}\quad} &&NL(F)\parr_\mathcal{C} \Big(NL(E)\parr_\mathcal{C} \K\Big)\ar[r]^{\ \ \ \ \Lambda^{-1}_{F,E,\K}} & NL(F\times E)\parr_\mathcal{C} \K } 
\]
and compatibility with unitors for a given canonical isomorphism $\epsilon:\K\to NL(0_\mathcal{M})$:
\[
\xymatrix@C=40pt{
 NL(0_\mathcal{M})\parr_\mathcal{C} NL(F)\ar[rr]^{NL(0_\mathcal{M})\parr\ \rho_{NL(F)}\quad}  & &  NL(0_\mathcal{M})\parr_\mathcal{C} \Big(NL(F)\parr_\mathcal{C} \K\Big)\ar[r]^{\ \ \ \ \Lambda^{-1}_{0_\mathcal{M},F,\K}}  &  NL(0_\mathcal{M}\times F)\parr_\mathcal{C} \K\ar[d]|{NL(\ell_F)\parr_\mathcal{C} \K}\\
\K\parr_\mathcal{C} NL(F)\ar[u]|{\epsilon\parr_\mathcal{C} NL(F)}\ar[rr]^{\lambda_{NL(F)}^{-1}\quad} &&NL(F)\ar[r]^{\ \ \ \ \rho_{NL(F)}} & NL(F)\parr_\mathcal{C} \K } 
\]
\[
\xymatrix@C=40pt{
 NL(E)\parr_\mathcal{C} NL(0_\mathcal{M})\ar[rr]^{NL(E)\parr\ \rho_{NL(0_\mathcal{M})}\quad}  & &  NL(E)\parr_\mathcal{C} \Big(NL(0_\mathcal{M})\parr_\mathcal{C} \K\Big)\ar[r]^{\ \ \ \ \Lambda^{-1}_{E,0_\mathcal{M},\K}}  &  NL(E\times 0_\mathcal{M})\parr_\mathcal{C} \K\ar[d]|{NL(r_E)\parr_\mathcal{C} \K}\\
NL(E)\parr_\mathcal{C}\K \ar[u]|{NL(E)\parr_\mathcal{C} \epsilon}\ar[rrr]^{id\quad} && & NL(E)\parr_\mathcal{C} \K } 
\]
 The model is said to be a \textit{Seely} model if $U$ is bijective on objects. 
\end{definition}

In our examples, $U$ must be thought of as an underlying functor that forgets the linear structure of $\mathcal{C}$ and sees it as a special smooth structure in $\mathcal{M}$. Hence we could safely assume it faithful and bijective on objects. 

\begin{proposition}
A Seely $\lambda$-model of $\lambda$-tensor logic is a Seely  linear model of $\lambda$-tensor logic too 
\end{proposition}
\begin{proof}
Start with a $\lambda$-model. Let $$\mu_{E,F}^{-1}=\rho_{NL(E\times F)}^{-1}\circ\Lambda^{-1}_{E,F,\K}\circ (id_{NL(E)}\parr\rho_{NL(F)}):NL(E)\parr_\mathcal{C}  NL(F)\to NL(E\times F)$$ using the right unitor $\rho$ of $\mathcal{C}^{op}$, and composition in $\mathcal{C}$. The identity isomorphism $\epsilon$ is also assumed given. Since $\mu$ is an isomorphism it suffices to see it makes $NL$ a lax symmetric monoidal functor. The symmetry condition is exactly the diagram of compatibility with symmetry that we assumed and similarly for the unitality conditions. The first assumed diagram with $\Lambda$ used in conjection with $U$  faithful enables to transport any diagram valid in the Cartesian closed category to an enriched version, and the second diagram concerning compatibility with associativity is then the only missing part needed so that $\mu$ satisfies the relation with associators of $\parr,\times$.
\end{proof}

Those models enable to recover models of linear
logic. We get a linear-non-linear adjunction in the sense of \cite{Benton} (see also \cite[def 21 p 140]{Mellies}). 

\begin{theorem}\label{lambdaTensortoLL}
$(\mathcal{C}^{op},\parr_\mathcal{C},I,\neg, \mathcal{M},\times ,0, NL,U)$ a Seely linear model of $\lambda$-tensor logic. Let $\mathcal{D}\subset  \mathcal{C}$ the full subcategory of objects of the form $\neg C, C\in \mathcal{C}$. Then, $\mathcal{N}=U(\mathcal{D})$ is equivalent to $\mathcal{M}$. $\neg\circ NL:\mathcal{N}\to \mathcal{D}$ is left adjoint to $U:\mathcal{D}\to \mathcal{N}$ and forms a linear-non-linear adjunction. Finally $!=\neg\circ NL\circ U$ gives a comonad on $\mathcal{D}$ making it a $*$-autonomous complete and cocomplete Seely category with Kleisli category for $!$ isomorphic to $\mathcal{N}$. 
\end{theorem}

\begin{proof}
This is a variant of \cite[Thm 2.13]{TabareauPhD}.
We already saw in lemma \ref{DialogueToRef} that $\mathcal{D}$ is $*$-autonomous with the structure defined there. Composing the natural isomorphisms for 
$F\in\mathcal{D}, E\in \mathcal{M}$ $$\mathcal{M}(E,U(F))\simeq \mathcal{C}^{op}(NL(E),\neg F)\simeq\mathcal{D}(\neg(NL(E)), F),$$ one gets the stated adjunction. The equivalence is the inclusion with inverse $\neg\neg:\mathcal{M}\to\mathcal{N}$ which is based on the canonical map in $\mathcal{C}$, $\eta_E:\neg\neg E\to E$ which is mapped via $U$ to a corresponding natural transformation in $\mathcal{M}$. It is an isomorphism in $\mathcal{N}$ since any element is image of $U$ enabling to use the $\neg$-relative adjunction  for $E\in \mathcal{C}$: $$\mathcal{M}(U(E),U(\neg\neg E))\simeq\mathcal{C}^{op}(NL(U(E)), \neg \neg\neg E)\simeq \mathcal{C}^{op}(NL(U(E)), \neg E)\simeq\mathcal{M}(U(E),U( E)).$$
Hence the element corresponding to identity gives the inverse of $\eta_E$. Since $\mathcal{D}$ is coreflective in $\mathcal{C}$, the coreflector preserves limits enabling to compute them in $\mathcal{D}$, and by $*$-autonomy, it therefore has colimits (which must coincide with those in $\mathcal{C}$).
By \cite[Prop 25 p 149]{Mellies}, since $U:\mathcal{D}\to \mathcal{N}$ is still a bijection on objects, the fact that $\mathcal{D}$ is a Seely category follows and the computation of its Kleisli category too. The co-unit and co-multiplication of the co-monad $!$ come from the relative adjunction $U \dashv_{neg} NL$, and correspond respectively to  the identity on $E$  in $M$, and to the composition of the unit of the adjunction by $!$ on the left and $U$ on the right. 
\end{proof}

\begin{remark}\label{FaithfullnessLambdaTensor}
In the previous situation, we checked that $U(E)\simeq U(\neg\neg E)$ in $\mathcal{M}$ and we even obtained a natural isomorphism  $U\circ \neg\neg \simeq U$ and this has several consequences we will reuse. First $\neg$ is necessarily faithful on $\mathcal{C}$ since if $\neg(f)=\neg(g)$ then $U\circ \neg\neg(f)=U\circ \neg\neg(g)$  hence $ U(f)=U(g)$ and $U$ is assumed faithful hence $f=g$. Let us see that as a consequence, as for $\varepsilon$, $\parr_\mathcal{C}$ preserves monomorphisms. Indeed if $f:E\to F$ is a monomorphism, $\neg\neg(f\parr_\mathcal{C} id_G)$ is the application of the $\neg\neg(\cdot)\parr G$ for the $*$-autonomous continuation category, hence a right adjoint functor, hence $\neg\neg(f\parr_\mathcal{C} id_G)$ is a monomorphism since right adjoints preserve monomorphisms. 
 Since  $\neg\neg$ is faithful one deduces $f\parr_\mathcal{C} id_G$ is a monomorphism too.
\end{remark}

\subsection{A class of examples of LL models}
We now fix $Fin\subset\mathscr{C}\subset \FDFS.$ Recall that then $C^\infty_{\mathscr{C}}(E,F)=C^\infty(E,F)$ algebraically  for any \lcs $E,F$. The index $\mathscr{C}$ remains to point out the different topologies.

\begin{definition}
We define $E^*_{\mathscr{C}}$ as the Mackey completion of $\mathscr{S}_{\mathscr{C}}((\widehat{E}^M)'_\mu)$. Thus we can define $\mathscr{C}$-reflexive spaces as satisfying $E=(E^*_{\mathscr{C}})^*_{\mathscr{C}}$. We denote by $\Cref$ the category of $\mathscr{C}$-reflexive spaces and linear maps.
\end{definition}

The dialogue category $\mathscr{C}-\Mc$ enables to give a situation similar to $\rRef.$ First for any \lcs E, $(E^*_{\mathscr{C}})' = \mco{E}$ algebraically from lemma \ref{SCdual}. 

\begin{corollary}
\label{prop:Cref}
For any \lcs $E$ , $E^*_{\mathscr{C}}$ is $\mathscr{C}$-reflexive, and $(E^*_\mathscr{C})'_\mathscr{C}$ is Mackey-complete, hence equal to $(E^*_\mathscr{C})^*_\mathscr{C}$
\end{corollary}

\begin{proof}We saw $E^*_{\mathscr{C}}=\mathscr{S}_{\mathscr{C}}(E^*_\rho)$ but from lemma \ref{SCdual} and commutation of $\mathscr{S}_{\mathscr{C}}=\mathscr{S}\circ\mathscr{S}_{\mathscr{C}}$ with Mackey completions, $[\mathscr{S}_{\mathscr{C}}(E)]^*_\rho=E^*_\rho$. Hence composing and using Theorem \ref{rhoref}, one gets the claimed reflexivity: $$((E^*_{\mathscr{C}})^*_{\mathscr{C}})^*_{\mathscr{C}}
=\mathscr{S}_{\mathscr{C}}\Big[((E^*_{\rho})^*_{\rho})^*_{\rho}\Big]
=\mathscr{S}_{\mathscr{C}}\Big[E^*_{\rho}\Big]=E^*_{\mathscr{C}}$$
Similarly $(E^*_\mathscr{C})'_\mathscr{C}=\mathscr{S}_{\mathscr{C}}((E^*_\rho)'_\RR)$ which is Mackey-complete by the same result.
\end{proof}

\begin{theorem}\label{th:CRef}
Let $Fin\subset\mathscr{C}\subset \FDFS$
 a full Cartesian small subcategory. $\Cref$ is a complete and cocomplete $*$-autonomous category with tensor product $E {\o}_{\mathscr{C}} F=(E^*_{\mathscr{C}}\varepsilon F^*_{\mathscr{C}})^*_{\mathscr{C}}$ and dual $(.)^*_{\mathscr{C}}$ and dualizing object $\K$.  It is stable by arbitrary products. It is equivalent to  the Kleisli category of $\mathscr{C}-\Mc$ and to  $\rho$-\textbf{Ref} as a $*$-autonomous category via the inverse functors: $\mathscr{S}_{\mathscr{C}}:\rho$-\textbf{Ref}$\to \mathscr{C}$-\textbf{Ref} and $\mathscr{S}([.]_\mu):\mathscr{C}$-\textbf{Ref}$\to \rho$-\textbf{Ref}.
\end{theorem}

\begin{proof}
This is a consequence of lemma \ref{DialogueToRef} applied to the Dialogue category $(\mathscr{C}-\Mc^{op},\epsilon,\K,(\cdot)^*_{\mathscr{C}})$ from proposition \ref{Cdialogue}.
Recall from the previous proof that $(\mathscr{S}_{\mathscr{C}}(E))'_\mu=E'_\mu$ and  $((E)^*_{\mathscr{C}})^*_{\mathscr{C}}=\mathscr{S}_{\mathscr{C}}((E^*_\rho)^*_\rho).$ This implies the two functors are inverse of each other as stated. 

We show they intertwine the other structure. We already noticed $E^*_{\mathscr{C}}=\mathscr{S}_{\mathscr{C}}(E^*_\rho)$. 
We computed in lemma \ref{DualArensMc}:
$$(E^*_{\rho}\varepsilon F^*_{\rho})'_\mu\simeq(E^*_{\rho}\eta F^*_{\rho})'_\mu\simeq  (E^*_{\rho})'_\mu\widehat{\o}^M_i(F^*_{\rho})'_\mu\simeq(E^*_{\mathscr{C}}\varepsilon F^*_{\mathscr{C}})'_\mu$$
Since $\varepsilon$ product keeps Mackey-completeness, one can compute $(\cdot)^*_{\mathscr{C}}$ and $(\cdot)^*_{\rho}$ by applying respectively $\mathscr{S}_{\mathscr{C}}(\widehat{\cdot}^M)$ and $\mathscr{S}(\widehat{\cdot}^M)$, which gives the missing topological identity:
$$\mathscr{S}_{\mathscr{C}}\Big((E^*_{\rho}\varepsilon F^*_{\rho})^*_\rho\Big)\simeq(E^*_{\mathscr{C}}\varepsilon F^*_{\mathscr{C}})^*_{\mathscr{C}}.$$
\end{proof}

Let $\Cref_\infty, \mathscr{C}-\Mc_\infty$ the Cartesian categories with same spaces as $\Cref, \mathscr{C}-\Mc$ and $\mathscr{C}$-smooth maps, namely conveniently smooth maps. Let $U:\Cref\to \Cref_\infty$  the inclusion functor  (forgetting linearity and continuity of the maps). Note that, for $\mathscr{C}\subset \mathscr{D}$, $ \mathscr{C}-\Mc_\infty\subset\mathscr{D}-\Mc_\infty$ is a full subcategory.
\begin{theorem}\label{th:CRefSeely}
Let $Fin\subset\mathscr{C}\subset \FDFS$
 as above. $\Cref$ is also a Seely category with structure extended by the comonad $!_\mathscr{C}(\cdot)=(C_\mathscr{C}^\infty(\cdot))^*_\mathscr{C}$ associated to the adjunction with left adjoint $!_\mathscr{C}:\Cref_\infty\to \Cref$ and right adjoint $U$.
\end{theorem}
\begin{proof}
We apply Theorem \ref{lambdaTensortoLL} to $\mathcal{C}=\mathscr{C}-\Mc$ so that $\mathcal{D}=\Cref$ and $\mathcal{N}=\Cref_\infty$. For that we must check the assumptions of a $\lambda$-categorical model for $\mathcal{M}=\mathscr{C}-\Mc_\infty$. Lemma \ref{thm:CartesianClosedBasic} shows that $\mathcal{M}$ is a Cartesian closed category since the internal hom functor $\CinC(E,F)$ is almost by definition in $\mathscr{C}-\Mc$. Indeed it is a projective limit of $\Cin_{co}(X)\varepsilon F$ which is a projective kernel of  $\Cin_{co}(X)\varepsilon\Cin_{co}(Y)=\Cin_{co}(X\times Y)$ with $X,Y\in \mathscr{C}$ as soon as $F\in \mathscr{C}-\Mc.$ The identity in lemma \ref{Ccomplete} gives the natural isomorphisms for the $(\cdot)^*_\mathscr{C}$-relative adjunction (the last one algebraically using $\CinC(E)\in\mathscr{C}-\Mc$):
$$\CinC(E,F)\simeq \CinC(E)\varepsilon F\simeq L(F'_c, \CinC(E))=L(F^*_\mathscr{C}, \CinC(E))=\mathcal{C}^{op}( \CinC(E), F^*_\mathscr{C})$$
It remains to see that $\CinC:\mathcal{M}\to \mathcal{C}$ is a symmetric unital functor satisfying the extra assumptions needed for a $\lambda$-categorical model. Note that Lemmas \ref{thm:CartesianClosedBasic} and \ref{Ccomplete} also provide the definitions of the map $\Lambda,\Xi$ respectively, the second diagram for $\Xi$. The diagram for $\Xi$ comparing the internal hom functors is satisfied by definition of the map $\Lambda$ which is given by a topological version of this diagram. Note that unitality and functoriality of $\CinC$ are obvious and that $\Lambda_{E,F,G}$ is even defined for any $G\in \Mc$.
It remains to prove symmetry and the second diagram for $\Lambda$. We first reduce it to $\mathscr{C}$ replaced by $Fin$.
For, note that, by  their definition as projective limit, there is a continuous identity map $\CinC(E)\to \Cin_{Fin}(E)$ for any \lcs E, and since smooth curves only depend on the bornology, $\Cin_{Fin}(E)\simeq\Cin_{Fin}(\mathscr{N}(E))$ topologically (recall $\mathscr{N}=\mathscr{S}_{Fin}$ is the reflector of  $I:Fin-\Mc\subset \mathscr{C}-\Mc$, which is a Cartesian functor \cite{Junek}, and thus also of $I_\infty:Fin-\Mc_\infty\subset \mathscr{C}-\Mc_\infty$ by this very remark.) Composing both, one gets easily a natural transformation $J_{\mathscr{C},Fin}:\CinC\to I\circ\Cin_{Fin}\circ\mathscr{N}$. It intertwines the Curry maps $\Lambda$ as follows for $G\in \Mc$:\[
\xymatrix@C=100pt{
\CinC(E)\varepsilon \Big(\CinC(F)\varepsilon G\Big) \ar[r]^{J_{\mathscr{C},Fin}(E)\varepsilon (J_{\mathscr{C},Fin}( F)\varepsilon G)\qquad \  } &  \Cin_{Fin}(\mathscr{N}(E))\varepsilon \Big(\Cin_{Fin}(\mathscr{N}(F))\varepsilon G\Big)\\
\CinC(E\times F)\varepsilon G\ar[u]_{\Lambda_{E,F,G}} \ar[r]^{J_{\mathscr{C},Fin}(E\times F)\varepsilon G\qquad }&\Cin_{Fin}(\mathscr{N}(E)\times \mathscr{N}(F))\varepsilon G\ar[u]_{\Lambda_{\mathscr{N}(E),\mathscr{N}(F),G}}
}
\]
Now, the associativity, symmetry and unitor maps are all induced from $\McS$, hence, it suffices to prove the compatibility diagrams for $\Lambda$ in the case of  $\Cin_{Fin}$ with $G\in \McS.$ In this case, we can further reduce it using that from naturality of associator, unitor and braiding, they commute with projective limits as $\varepsilon$ does, and from its construction in lemma \ref{thm:CartesianClosedBasic}  $\Lambda_{E,F,G}$ is also a projective limit of maps, hence the projective limit description of $\Cin_{Fin}$ reduces those diagrams to $E,F$ finite dimensional. Note that for the terms with products $E\times F$ the cofinality of product maps used in the proof of lemma \ref{thm:CartesianClosedBasic}) enables to rewrite the projective limit for $E\times F$ with the product of projective limits  for $E,F$ separately. The key to check the relations is to note that the target space of the diagrams is a set of multilinear maps on $(C^\infty(\R^n\times \R^m))', G'$ and to prove equality of the evaluation of both composition on an element in the source space, by linearity continuity and since $\overline{Vect e_{\R^{n+m}}( \R^{n+m})}=(C^\infty(\R^n\times \R^m))'$, it suffices to evaluate the argument in  $(C^\infty(\R^n\times \R^m))'$ on Dirac masses which have a product form. Then when reduced to a tensor product argument, the associativity and braiding maps are canonical and the relation is obvious to check.
\end{proof}

\begin{remark}\label{TwoComonads} In $\Cref$ we defined $!_{\mathscr{C}}E=((C^\infty_{\mathscr{C}}(E))^*_{\mathscr{C}})$ so that moving it back to $\rho$-\textbf{Ref} via the isomorphism of $*$-autonomous category of Theorem \ref{th:CRef}, one gets $\mathscr{S}([!_{\mathscr{C}}E]_\mu)=((C^\infty_{\mathscr{C}}(E))^*_{\rho})$. Let us apply Lemma \ref{WhyNotAP} and Example \ref{Ex:notAP}. For $\mathscr{C}=Fin$ one gets a space with its $\rho$-dual having the approximation property, whereas for $\mathscr{C}=Ban$, one may get one without it since $(!_{\mathscr{C}}\mathscr{S}(E))^*_\rho=((C^\infty_{co}(E))^*_\rho)^*_\rho$ if $E$ is a Banach space (since we have the topological identity $C^\infty_{\mathscr{C}}(E,\K)\simeq C^\infty_{\mathscr{C}}(\mathscr{S}(E),\K)$ coming from the identical indexing set of curves coming from the algebraic equality $C^\infty_{co}(X,E)=C^\infty_{Ban}(X,E)=C^\infty(X,E)=C^\infty_{co}(X,\mathscr{S}(E))$). Therefore, if $E$ is a Schwartz space associated to a Banach space in $Ban$ without the approximation property: $$\mathscr{S}([!_{Fin}E]_\mu)\ \mathlarger{\varsubsetneq} \ !_{Ban}E$$ (since both duals are algebraically equal to $C^\infty(E,\K)$, the difference of topology implies different duals algebraically). It is natural to wonder if there are infinitely many different exponentials obtained in that way for different categories $\mathscr{C}$. It is also natural to wonder if one can characterize $\rho$-reflexive spaces (or even Banach spaces) for which there is equality $\mathscr{S}([!_{Fin}E]_\mu)= \ !_{Ban}E.$
\end{remark}

\subsection{A model of LL : a Seely category}
\label{subsec:modelLL}

We referred to \cite{Mellies} in order to produce a Seely category. Towards extensions to DiLL models it is better to make more explicit the structure we obtained. 
First recall the various functors. When $f : E \to F $ is a continuous linear map with $E,F\in \mathscr{C}-\Mc$, we used $!_\mathscr{C} f : !_\mathscr{C}E \to !_\mathscr{C}F $ defined as $(\cdot\circ f)^*_\mathscr{C}$.
Hence $!_\mathscr{C}$ is indeed a functor from $\mathscr{C}-\Mc$ to $\mathscr{C}-\Ref$. 

Since $\CinC$ is a functor too on $\mathscr{C}-\Mc_\infty$, the above functor is decomposed in a adjunction as follows. For $F:E\to F$ $\mathscr{C}$-smooth, $\CinC(F)(g)=g\circ F, g\in \CinC(F,\R)$ and for a linear map $f$ as above, $U(f)$ is the associated smooth map, underlying the linear map. Hence we also noted $!_\mathscr{C} F=(\CinC(F))^*_\mathscr{C}$ gives the functor, left adjoint to $U:\mathscr{C}-\Mc\to \mathscr{C}-\Mc_\infty$ and our previous $!_\mathscr{C}$ is merely the new $!_\mathscr{C}\circ U$.

For any $E\in \Cref$, we recall the continuous isomorphism from $E$ to $(E^*_\mathscr{C})^*_\mathscr{C}=\mathscr{S}((E'_\mu)'_\mu)$

\begin{equation*}
 ev_E :
\left\lbrace
\begin{split}
E & \rightarrow (E^*_{\mathscr{C}})^*_{\mathscr{C}}= E\\
x & \mapsto ( l \in E^*_{\mathscr{C}} \mapsto l(x))
\end{split}
\right.
\end{equation*}

Note that if $E$ is only Mackey-complete, the linear isomorphism above is still defined, in the sense that we take the extension to the Mackey-completion of  $l\mapsto l(x)$, but it is only bounded/smooth algebraic isomorphism (but not continuous) by Theorem \ref{rhoref}.
However, $ev_E^{-1}$ is always linear continuous in this case too.

 We may still use the notation $e_E$ for any separated locally convex space $E$ as the bounded linear injective map, obtained by composition of the canonical map $E\to \widehat{E}^M$ and $ev_{\widehat{E}^M}$. We also consider the similar canonical maps: 

\begin{lemma}
\label{lem:evsmooth}
For any space $E\in  \mathscr{C}-\Mc$, there is a smooth map (the Dirac mass map):
\begin{equation*}
\delta_E :
\left\lbrace
\begin{split}
E & \rightarrow  (\CinC(E))'\subset !_\mathscr{C} E \\
x & \mapsto ( f \in \CinC(E,\K) \mapsto f(x)=\delta_E(x)(f)),
\end{split}
\right.
\end{equation*}
\end{lemma}

\begin{proof}
We could see this directly using convenient smoothness, but it is better to see it comes from our $\lambda$-categorical model  structure. We have an adjunction:
$$\CinC(E,!_\mathscr{C} E)\simeq \mathscr{C}-\Mc^{op}(\CinC(E),(!_\mathscr{C} E)^*_\mathscr{C})
=\mathscr{C}-\Mc((\CinC(E)^*_\mathscr{C})^*_\mathscr{C},\CinC(E))$$
and $\delta_E$ is the map in the first space,  associated to $ev_{\CinC(E)}^{-1}$ in the last.
\end{proof}

Hence, $\delta_E$ is nothing but the unit of the adjunction giving rise to $!_\mathscr{C}$, considered on the opposite of the continuation category.

As usual, see e.g. \cite[section 6.7]{Mellies}, the adjunction giving rise to $!_\mathscr{C}$ produces a comonad structure on this functor. The counit implementing the dereliction rule is the continuous linear map ${\Der E}:!_\mathscr{C}(E)\to E$ obtained in looking at the map corresponding to identity in the adjunction:
$$\CinC(E, E)\simeq \mathscr{C}-\Mc^{op}(\CinC(E),(E)^*_\mathscr{C})
=\mathscr{C}-\Mc(E^*_\mathscr{C},\CinC(E))\simeq \mathscr{C}-\Mc((\CinC(E))^*_\mathscr{C},E)$$

The middle map $\epsilon^{\CinC}_E\in \mathscr{C}-\Mc(E^*_\mathscr{C},\CinC(E))$ is the counit of the $(\cdot)^*_\mathscr{C}$-relative adjunction and it gives 
${\Der E}=ev_E^{-1}\circ (\epsilon^{\CinC}_E)^*_\mathscr{C}$ when $E\in\Cref$. 
The comultiplication map implementing the promotion rule  is obtained as ${\Digg E}=!_\mathscr{C}(\delta_E)=(\CinC(\delta_E))^*_\mathscr{C}$.

We can now summarize the structure. Note, that we write the usual $\top$, unit for $\times$ as $0$, for the $\{0\}$ vector space.

\begin{proposition}
  The functor $!_\mathscr{C}$ is an exponential modality for the Seely category of Theorem \ref{th:CRefSeely} in the following way:
  \begin{itemize}
  \item $(!_\mathscr{C},{\Digg {}},{\Der {}})$ is a comonad, with ${\Der E}=ev_E^{-1}\circ (\epsilon^{\CinC}_E)^*_\mathscr{C}$ and ${\Digg E}=!_\mathscr{C}(\delta_E)=(\CinC(\delta_E))^*_\mathscr{C}.$

\item $!_\mathscr{C}:(\Cref, \times,0)\rightarrow (\Cref,\otimes, \K)$ is a strong and symmetric monoidal functor, thanks to the isomorphisms $m^0:\K\simeq !_\mathscr{C}(0)$ and (the map composing tensor strengths and adjoints of $\Xi,\Lambda$ of $\lambda$-tensor models):
\begin{align*}m^2_{E,F} : !_\mathscr{C}E\otimes !_\mathscr{C}!F &=\Big((\CinC(E)^*_\mathscr{C})^*_\mathscr{C}\parr_\mathscr{C}(\CinC(F)^*_\mathscr{C})^*_\mathscr{C}\Big)^*_\mathscr{C}
\simeq\Big(\CinC(E, \K)\parr_\mathscr{C}\CinC(F, \K)\Big)^*_\mathscr{C}\\&\simeq\Big(\CinC(E,\CinC(F, \K))\Big)^*_\mathscr{C} \simeq (C^{\infty}(E \times F, \K))^*_\mathscr{C} \simeq !_\mathscr{C}(E\times F)\end{align*}

\item the following diagram commute:
$$\xymatrix{
!_\mathscr{C}E\otimes_\mathscr{C} !_\mathscr{C}F \ar[r]^{m^2_{E,F}} \ar[d]_{{\Digg E}\otimes_\mathscr{C} {\Digg F}} & !(E\times F) \ar[r]^{{\Digg {E\times F}}} & !_\mathscr{C}!_\mathscr{C}(E\times F) \ar[d]^{!_\mathscr{C}\langle!_\mathscr{C}\pi_1,!_\mathscr{C}\pi_2\rangle}\\
!_\mathscr{C}!_\mathscr{C}E\otimes_\mathscr{C} !_\mathscr{C}!_\mathscr{C}F\ar[rr]_{m^2_{!_\mathscr{C}E,!_\mathscr{C}F}} &&!_\mathscr{C}(!_\mathscr{C}E\times !_\mathscr{C}F)
}$$
  \end{itemize}
\end{proposition}

\noindent Moreover, the comonad induces a structure of bialgebra on every space $!_\mathscr{C}E$ and this will be crucial to obtain models of DiLL \cite{Ehrhard16}. We profit of this section for recalling how all the diagrams there not involving codereliction are satisfied. In general, we have maps giving a commutative comonoid structure (this is the coalgebra part of the bialgebra, but it must not be confused with the coalgebra structure from the comonad viewpoint):

\begin{itemize}
\item $ {\Contr E} : \occ E \rightarrow \occ ( E \times E)\simeq \occ E \otimes \occ E  $ given by  $ {\Contr E}=(m^2_{E,E})^{-1}\circ !_\mathscr{C}(\Delta_E)$ with $\Delta_E(x)=(x,x)$ the canonical diagonal map of the Cartesian category.
\item
${\Weak E}= (m^0)^{-1}\circ !_\mathscr{C}(n_E): \occ E \rightarrow \occ 0\simeq \R $  with $n_E:E\to 0$ the constant map, hence more explicitly ${\Weak E}(h)=h(1)$ for $h\in \occ E$ and $1\in \CinC(E)$ the constant function equal to $1$.
\end{itemize}

This is exactly the structure considered in \cite[Chap 4 \S 6]{Bierman} giving a Seely category (in his terminology a new-Seely category) the structure of a Linear category (called $\mathcal{L}^{!}_{\otimes}$-model in \cite{Fiore}) from his Definition  35 in his Thm 25. See also \cite[7.4]{Mellies} for a recent presentation. This especially also contains the compatibility diagrams of \cite[2.6.1]{Ehrhard16}. Especially, ${\Digg E}:(\occ E,{\Weak E},{\Contr E})\to (\occ \occ E,{\Weak {\occ E}},{\Contr {\occ E}})$ is a comonoid morphism as in \cite[2.6.3]{Ehrhard16}. Also $!_\mathscr{C}$ is given the structure of a symmetric monoidal endofunctor on $\Cref$, $(!_\mathscr{C},\mu^0,\mu^2)$ making ${\Weak E},{\Contr E}$ coalgebra morphisms. For instance,  $\mu^0:\R\to !_\mathscr{C}(\R)$ (the space of distributions) is given by \cite[Chap 4 Prop 20]{Bierman} as $!_\mathscr{C}(v_{\R})\circ m^0$, i.e. {$\mu^0(1)=\delta_1$ with $v_{\R}:0\to \R$ the map with $u_R(0)=1$.} 
By \cite{Bierman}, a Linear category with products is actually the same thing as a Seely category. This is what is called in \cite{Fiore} a $\mathcal{L}^{!}_{\otimes,\times}$-model. So far, this structure is available in the setting of Theorem \ref{lambdaTensortoLL}, and we will use it in this setting later.

As explained in \cite{Fiore}, the only missing piece of structure to get a bicomoid structure on every $!_\mathscr{C}E$ is a biproduct compatible with the symmetric monoidal structure, or equivalently a \textbf{Mon}-enriched symmetric monoidal category, where \textbf{Mon} is the category of monoids.
 This is what he calls a $\mathcal{L}^{!}_{\otimes,*}$-model.

His Theorem 3.1 then provides us with the two first compatibility diagrams in \cite[2.6.2]{Ehrhard16} and the second diagram in \cite[2.6.4]{Ehrhard16}.

In our case $ \bigtriangledown_E:E\times E\to E$ is the sum when seeing $E\times E=E\oplus E$ as coproduct and its unit $u:0\to E$ is of course the $0$ map. Hence $(\mathscr{C}-Mc,0,\times, u, \bigtriangledown;n,\Delta)$ is indeed a biproduct structure. And compatibility with the monoidal structure, which boils down to biadditivity of tensor product, is obvious. One gets cocontraction and coweakening maps:

\begin{itemize}
\item $ {\Cocontr E} : \occ E \otimes \occ E\simeq \occ(E\times E) \rightarrow \occ E$ is the convolution product, namely it corresponds to $\occ(\bigtriangledown_E)$. 
\item ${\Coweak E} : \R\simeq \occ(0) \rightarrow \occ E $ is given by ${\Coweak E} (1) = (ev_0)^*_\mathscr{C}$ with $ev_0=\CinC(u_E)$ i.e. $ev_0(f)=f(0)$. 
\end{itemize}
From \cite[Prop 3.2]{Fiore} $(\occ E,{\Contr E},{\Weak E}, {\Cocontr E},{\Coweak E})$ is a commutative bialgebra. The remaining first diagram in \cite[2.6.4]{Ehrhard16} is easy and comes in our case for $f\in \CinC(!E)$ from $${[\occ (\delta_E\circ u_E)]}(f)=\delta_{{\delta_0}}(f)={\delta_1}(\lambda\mapsto f(\lambda(\delta_0))=[\occ {\Coweak E}({\delta_1})](f)=[\occ {\Coweak E}(\mu^0(1))](f).$$

To finish checking the assumptions in \cite{Ehrhard16}, it remains to check the assumptions in 2.5 and 2.6.5. As \cite{Fiore} is a conference paper, they were not explicitly written there. 
 \begin{equation}
   \begin{array}{ccc}
          \xymatrix @R=0.8em @C=3pc 
     {E & \\ & \Excl E\ar[ul]_{\Der E} \\ \One\ar[ur]_{\Coweak E}\ar[uu]^0 &}
     &\quad&
     \xymatrix @R=0.8em @C=3pc
     {E & \\ & \Excl E\ar[ul]_{\Der E} \\
       \Tens{\Excl E}{\Excl E}\ar[ur]_{\Cocontr E}
          \ar[uu]^{\rho_E^{-1}\circ(\Tens{\Der E}{\Weak E})+\lambda_E^{-1}\circ(\Tens{\Weak E}{\Der E})} &}  
   \end{array}
 \end{equation}
 \begin{equation}
   \xymatrix @R=1.2em @C=3pc{
     \Excl E\ar[rr]^-{\Excl 0}
     \ar[dr]_-{\Weak E}&&\Excl E\\
     &\One\ar[ru]_-{\Coweak E}&
   }
  \quad
   \xymatrix @R=1.2em @C=3pc{
     \Excl E\ar[r]^-{\Excl{(f+g)}}\ar[d]_-{\Contr E}
     &\Excl F\\
     \Tens{\Excl E}{\Excl E}\ar[r]^-{\Tens{\Excl f}{\Excl g}}
     &\Tens{\Excl F}{\Excl F}\ar[u]_-{\Cocontr F}    
     }
 \end{equation}

The first is $ev_E^{-1}\circ (\epsilon^{\CinC}_E)^*_\mathscr{C}\circ (\CinC(u_E))^*_\mathscr{C}=ev_E^{-1}\circ (\CinC(u_E)\circ\epsilon^{\CinC}_E )^*_\mathscr{C}=ev_E^{-1}\circ ((u_E )^*_\mathscr{C})^*_\mathscr{C}=u_E=0$ as expected.
The second is $ev_E^{-1}\circ (\epsilon^{\CinC}_E)^*_\mathscr{C}\circ (\CinC(\bigtriangledown_E))^*_\mathscr{C}=ev_E^{-1}\circ (\CinC(\bigtriangledown_E)\circ\epsilon^{\CinC}_E )^*_\mathscr{C}=ev_E^{-1}\circ ((\bigtriangledown_E )^*_\mathscr{C})^*_\mathscr{C}\circ \Tens{(\epsilon^{\CinC}_E)^*_\mathscr{C}}{(\epsilon^{\CinC}_E)^*_\mathscr{C}}$ which gives the right value  since $\bigtriangledown_E=r_E^{-1}\circ(id_E\times n_E)+ \ell_E^{-1}\circ(n_E\times id_E).$

The third diagram comes from $n_Eu_E=0$ and the last diagram from $\bigtriangledown_Y\circ (f\times g)\circ \Delta_X=f+g$ which is the definition of the additive structure on maps.

\subsection{Comparison with the convenient setting of Global analysis and Blute-Ehrhard-Tasson}
\label{sec:ConvenientSetting}
In \cite{BET12}, the authors use the Global setting of convenient analysis \cite{FrolicherKriegel,KrieglMichor} in order to produce a model of Intuitionistic differential Linear logic. They work on the category $\Conv$ of convenient vector spaces, i.e. bornological Mackey-complete (separated) lcs, with continuous (equivalently bounded), linear maps as morphisms. Thus, apart for the bornological requirement, the setting seems really similar to ours. It is time to compare them.

First any bornological space has its Mackey topology, let us explain why $\mathscr{S}:\Conv\to \rRef$ is an embedding giving an isomorphic full subcategory (of course with inverse $(\cdot)_\mu$ on its image). Indeed, for $E\in \Conv$ we use 
Theorem \ref{MackeyCaractrhoRef} in order to see that $\mathscr{S}(E)\in \rRef$ and it only remains to note that $E'_\mu$ is Mackey-complete. 

As in Remark \ref{exValdivia2}, $E$ bornological Mackey-complete, thus ultrabornological, implies $E'_\mu$ and even $\mathcal{S}(E'_\mu)$ complete hence Mackey-complete (and $E'_c$ $k$-quasi-complete).

Said otherwise, the bornological requirement ensures a stronger completeness property of the dual than Mackey-completeness, the completeness of the space, our functor $(\widehat{(\cdot)'_\mu})'_\mu$ should thus be thought of as a replacement of the bornologification functor in \cite{FrolicherKriegel} and $((\cdot)^*_\rho)^*_\rho$ is our analogue of their Mackey-completion functor in \cite{KrieglMichor} (recall that their Mackey completion is what we would call Mackey-completion of the bornologification). Of course, we already noticed that we took the same smooth maps and $\mathscr{S}:\Conv_\infty\to \rRef_\infty$ is even an equivalence of categories. Indeed, $E\to E_{born}$ is smooth and gives the inverse for this equivalence. 

Finally note that $E^*_\rho\varepsilon F=L_\beta(E,F)$ algebraically if $E\in \Conv$ since $E^*_\rho\varepsilon F=L_\epsilon((E^*_\rho)'_c,F)=L_\RR(E,F)$ topologically and the space of continuous and bounded linear maps are the same in the bornological case. $L_\beta(E,F)\to L_\RR(E,F)$ is clearly continuous hence so is $\mathscr{S}(L_\beta(E,F))\to \mathscr{S}(L_\RR(E,F))=L_\RR(E,F)$. 

But the closed structure in $\Conv$ is given by $(L_\beta(E,F))_{born}$ which uses a completion of the dual and hence we only have 
a lax closed functor property for $\mathscr{S}$, in form (after applying $((\cdot)^*_\rho)^*_\rho$) of a continuous map: \begin{equation}\label{ConvSettingClosedFunctor}\mathscr{S}((L_\beta(E,F))_{born})\to ((E^*_\rho\varepsilon F)^*_\rho)^*_\rho.\end{equation} Similarly, most of the linear logical structure is not kept by the functor $\mathscr{S}$.

\section{Models of DiLL}
\label{sec:dill}

Smooth linear maps in the sense of Frölicher are bounded but not necessarily continuous. Taking the differential at $0$ of functions in $\Cin(E,F)$ thus would not give us a morphisms in $\mathbf{k-Ref}$, thus we have no interpretation for the codereliction $\CODER$ of DiLL. We first introduce a general differential framework fitting Dialogue categories, and show that the variant of smooth maps introduce in section \ref{ksmooth} allows for a model of DiLL. 

\subsection{An intermediate notion : models of differential $\lambda$-Tensor logic.}
We refer to \cite{Ehr11,Ehrhard16} for surveys on differential linear logic.

According to Fiore and Ehrhard \cite{Fiore,Ehrhard16}, models of differential linear logic are given by Seely $*$-autonomous complete categories $\mathcal{C}$ 
 with a biproduct structure and either a creation operator natural transformation $\partial_E:  !E\o E\to !E$ or a creation map/codereliction natural transformation ${\Coder E}:E\to !E$ satisfying proper conditions. We recalled in subsection \ref{subsec:modelLL} the structure available without codereliction. Moreover, in the codereliction picture, one requires
the following diagrams to  commute 
\cite[2.5, 2.6.2,2.6.4]{Ehrhard16}: 

 \begin{equation}\label{Fior}
   \begin{array}{ccccc}
     \xymatrix @R=0.8em @C=3pc
     {E \ar[dr]^{\Coder E} \ar[dd]_0 & \\
       & \Excl E\ar[dl]^{\Weak E}\\ \One &}
     &\quad&
     \xymatrix @R=0.8em @C=3pc
     {E \ar[dr]^{\Coder E} \ar[dd]_{(\Tens{\Coder E}{\Coweak E})\circ \rho_E+(\Tens{\Coweak E}{\Coder E})\circ \lambda_E} & \\
     & \Excl E\ar[dl]^{\Contr E}\\  \Tens{\Excl E}{\Excl E} &}&\quad&\xymatrix@C=10pt{
& ! E \ar[dr]^{{\Der E}
}& \\
E \ar[ur]^{{\Coder E}
} \ar[rr]_{id_E}&&E 
}
   \end{array}
 \end{equation}

\begin{equation}\label{FioreStrength}
\xymatrix@C=50pt{
E \otimes ! F \ar[r]^{{\Coder E} \otimes !F}
\ar[dr]_{E\otimes {\Der E}}& 
! E \otimes! F \ar[r]^{\mu^2_{E,F}}&
! (E \otimes F)\\
& E\otimes F \ar[ur]_{{\Coder {E \otimes F}}}
}
\end{equation}
 \begin{equation}\label{FioreChainRule}
   \xymatrix @R=1.2em @C=3pc{
     E\ar[r]^-{\Coder E}
     \ar[d]_-{\lambda_E}
     & \Excl E\ar[r]^-{\Digg E}
     & \Excl{\Excl E}\\
     \Tens\One E\ar[r]^-{\Tens{\Coweak E}{\Coder E}}
     &
     \Tens{\Excl E}{\Excl E}\ar[r]^-{\Tens{\Digg E}{\Coder{\Excl E}}}
     & \Tens{\Excl{\Excl E}}{\Excl{\Excl E}}\ar[u]_-{\Cocontr{\Excl E}}
     }
  \end{equation}

Then from \cite[Thm 4.1]{Fiore} (see also \cite[section 3]{Ehrhard16}) the creation operator  $\partial_E={\Cocontr E}\circ (\Tens{!E}{\Coder E})$

We again need to extend this structure to a Dialogue category context. In order to get a natural differential extension of Cartesian closed category, we use differential $\lambda$-categories from \cite{BuccEhrMan}. This notion gathers the maybe very general Cartesian differential categories of Blute-Cockett-Seely to Cartesian closedness, via the key axiom (D-curry), relating applications of the differential operator $D$ and the curry map $\Lambda$ for $f:C\times A\to B$ (we don't mention the symmetry of Cartesian closed category $(C\times C\times A)\times A\simeq (C\times A)\times (C\times A)$):
$$D(\Lambda(f)) =\Lambda\Big(D(f)\circ \langle(\pi_1\times 0_A),\pi_2\rangle\Big): (C\times C)\to [A, B].$$

 We also use $Diag(E)= E\times E$ the obvious functor. We also suppose that the Cartesian structure is a biproduct, a supposition that is equivalent to supposing a \textbf{Mon}-enriched category as shown by Fiore \cite{Fiore}.

The idea is that while $D$ encodes the usual rules needed for differential calculus
, $d$ encodes the fact that we want the derivatives to be smooth, that is compatible with the linear duality structure we had before. 

\begin{definition}\label{DiffLambdaTensorDef}
A (resp. commutative) 
model of differential $\lambda$-tensor logic is a (resp. commutative) $\lambda$-categorical model of  $\lambda$-tensor logic with dialogue category $(\mathcal{C}^{op},\parr_\mathcal{C},1_\mathcal{C}=\K,\neg)$ {with a biproduct structure compatible with the symmetric monoidal structure}, a Cartesian closed category $(\mathcal{M},\times,0,[\cdot,\cdot])$, which is a differential $\lambda$-category  with operator $D$ internalized as a natural transformation $D_{E,F}:[E,F]\to [Diag(E),F]$ (so that $D$ in the definition of those categories is given by $M(D_{E,F}):\mathcal{M}(E,F)\to \mathcal{M}(E\times E,F)$ with $M$ the basic functor to sets of the closed category $\mathcal{M}$). {We assume $U:\mathcal{C}\to \mathcal{M}$ {and $\neg:\mathcal{C}^{op}\to\mathcal{C}$}  are $\textbf{Mon}$-enriched functors.} We also assume given an internalized differential operator, given by a natural transformation $$d_{E,F}:NL(U(E))\parr_\mathcal{C}F\to NL(U(E))\parr_\mathcal{C}(\neg E \parr_\mathcal{C}F)$$ 
satisfying the following  commutative diagrams (with the opposite of the counit of the relative adjointness relation, giving a map in $\mathcal{C}$ written: $\epsilon^{NL}_E:\neg E\to NL(U(E))\equiv NL_E$) expressing compatibility of the two differentials. We have a first diagram in $\mathcal{M}$ :
\[
\xymatrix@C=20pt{
U\Big(NL_E\parr_\mathcal{C} F\Big)\ar[d]_{\Xi_{E,F}} \ar[r]^{U(d_{E,F})\ \ \ \ \ \ \ \ } &  U\Big(NL_E\parr_\mathcal{C}(\neg E \parr_\mathcal{C}F)\Big)\ar[rrr]^{U(NL_E)\parr_\mathcal{C}(\epsilon^{NL}_E \parr_\mathcal{C}F))\ \ \ \ } &&&U\Big(NL_E\parr_\mathcal{C}(NL_E \parr_\mathcal{C}F)\Big)\ar[d]|{[id_{U(E)},\Xi_{E,F}]\circ \Xi_{E,NL( E )\parr_\mathcal{C}F}} \\
[U(E),U(F)] \ar[r]^{D_{U(E),U(F)}}&[U(E\times E),U(F)]\ar[rrr]^{\Lambda^{\mathcal{M}}_{U(E), U(E),U(F)}
}&&&[U(E),[U(E),U(F)]]
}
\]
and weak differentiation property diagram in $\mathcal{C}$ :
\[
\xymatrix@C=20pt{
NL(U(E))\parr_\mathcal{C} F\ar[d]_{\rho_{NL(U(E))}\parr_\mathcal{C}F} \ar[r]^{d_{E,F}\ \ \ \ } &  NL(U(E))\parr_\mathcal{C}(\neg E \parr_\mathcal{C}F)\ar[rr]^{Ass^{\parr_\mathcal{C}}_{NL(U(E)),\neg E,F}\ \ \ \ } &&(NL(U(E))\parr_\mathcal{C}\neg E )\parr_\mathcal{C} F\ar[d]|{NL(U(E))\parr_\mathcal{C}\rho_{\neg E}\parr_\mathcal{C}F} \\
(NL(U(E))\parr_\mathcal{C}\K)\parr_\mathcal{C} F \ar[rrr]^{d_{E,\K}\parr_\mathcal{C}F\quad}&&&(NL(U(E))\parr_\mathcal{C}(\neg E \parr_\mathcal{C}\K))\parr_\mathcal{C} F
}
\]
\end{definition}

The idea is that while $D$ encodes the usual rules needed for a Cartesian Differential Category, $d$ encodes the fact that we want the derivatives to be smooth, that is compatible with the linear duality structure we had before.

\begin{theorem}\label{DifflambdaTensortoDiLL}
$(\mathcal{C}^{op},\parr_\mathcal{C},I,\neg, \mathcal{M},\times ,0, NL,U,D,d)$ a Seely model of differential $\lambda$-tensor logic. Let $\mathcal{D}\subset  \mathcal{C}$ the full subcategory of objects of the form $\neg C, C\in \mathcal{C}$, equipped with $!=\neg\circ NL\circ U$ as comonad on $\mathcal{D}$ making it a $*$-autonomous complete and cocomplete Seely category with Kleisli category for $!$ isomorphic to $\mathcal{N}=U(\mathcal{D})$. With the dereliction $\Der {E}$ as in subsection \ref{subsec:modelLL} and the codereliction  interpreted by : ${\Coder E}=\neg\Big((NL(E)\parr_\mathcal{C}\rho_{\neg E}^{-1})\circ d_{E,\K}\circ \rho_{NL(E)}\Big)\circ (\Tens{\Coweak E}{Id_{E}})\circ {\lambda_E}$
, this makes
$\mathcal{D}$ a model of differential Linear Logic.
\end{theorem}

\begin{proof}

The setting comes from Theorem \ref{lambdaTensortoLL} giving already a model of Linear logic. Recall from subsection \ref{subsec:modelLL} that we have already checked all diagrams not involving codereliction. We can and do fix $E=\neg C$ so that $\epsilon^{\neg\neg}_E:\neg\neg E\to E$ is an isomorphism that we will ignore safely in what follows.
\setcounter{Step}{0}

\begin{step} Internalization of D-curry from \cite{BuccEhrMan}\end{step}
Let us check:
\begin{align*}
\xymatrix@C=30pt{
 NL(E)\parr NL(E) \parr(\neg E \parr\K)
& 
NL( E) \parr (NL( E) \parr\K)\ar[l]_{\quad NL(E)\parr d_{E,\K}}&&
NL( E\times E)\parr \K \ar[ll]_{\quad \Lambda_{E,E,\K}}\ar[d]_{d_{E\times E,\K}}\\
& NL( E\times E)\parr(\neg E\parr \K) \ar[ul]^{ \Lambda_{E,E,\neg E\parr \K}}&&NL( E\times E)\parr((\neg E)^{\oplus 2})\parr \K) \ar[ll]_{NL( E\times E)\parr(\pi_2\parr \K)}
}
\end{align*}
Indeed using compatibility with symmetry from definition of $\lambda$-categorical models, it suffices to check a flipped version with the derivation acting on the first term. Then applying the faithful $U$, intertwining with $\Xi$ and using the compatibility with $D_{E,F}$ the commutativity then follows easily from  D-curry.

\begin{step} Internalization of chain rule D5 from \cite{BuccEhrMan}\end{step}
For $g\in \mathcal{M}(U(E),U(F))=\mathcal{M}(0,[U(E),U(F)])\simeq\mathcal{M}(0,U(NL(E)\parr_\mathcal{C} F)\simeq \mathcal{C}(\neg(NL(0)),NL(E)\parr_\mathcal{C} F\parr_\mathcal{C}K) $, which gives a map $h:\neg(\K)\to NL(E)\parr_\mathcal{C} F\parr_\mathcal{C} \K$. One gets $d_{E,F}\circ h:\neg(\K)\to NL(E)\parr_\mathcal{C}\neg E\parr_\mathcal{C} F\parr_\mathcal{C} \K$ giving by characteristic diagram of dialogue categories (for $\mathcal{C}^{op}$, recall our maps are in the opposite of this dialogue category) a map $dH:\neg F\to NL(E)\parr_\mathcal{C}\neg E\parr_\mathcal{C} \K$
. We leave as an exercise to the reader to check that D5 can be rewritten as before:
\[   \xymatrix @R=1.2em @C=2pc{
      {NL(U(E))}\parr_\mathcal{C} (\neg E\parr_\mathcal{C}\K)
          & NL(U(E))\parr_\mathcal{C}\K\ar[l]_{\ \ \ \  d_{E,\K}}
     & \ar[l]_{NL(g)\parr_\mathcal{C}\K} NL(U(F))\parr_\mathcal{C}\K\ar[d]_-{ d_{F,\K}}\\
     {NL(U(E))}\parr_\mathcal{C}{NL(U(E))}\parr_\mathcal{C}  (\neg E\parr_\mathcal{C}\K)\ar[u]_-{\Big(
     NL(\Delta_{U(E)})\Lambda_{E,E}^{-1}
     \big)\parr_\mathcal{C} (\neg E\parr_\mathcal{C}\K)}
     &&{NL(U(F))}\parr_\mathcal{C}{(\neg (F)\parr_\mathcal{C}\K)}\ar[ll]^-{NL(g)\parr_\mathcal{C}dH}
     }
\]

Note that we can see $dH$ in an alternative way using our weak differentiation property. Composing with a minor isomorphism, if we see $h:\neg(\K)\to (NL(E)\parr_\mathcal{C} \K)\parr_\mathcal{C}F $ then one can consider $(d_{E,\K}\parr F)\circ h$ and it gives $d_{E,F}\circ h$ after composition by a canonical map. But if $H:\neg(F)\to (NL(E)\parr_\mathcal{C} \K)$ is the map associated to $h$ by the map $\varphi$ of dialogue categories, the naturality of this map gives exactly $dH= d_{E,\K}\circ H.$ Note that if $g=U(g')$, by the naturality of the isomorphisms giving $H$, it is not hard to see that $H=\epsilon^{NL}_{E}\circ\neg (g').$

\begin{step}Two first diagrams in  \eqref{Fior}.\end{step}

For the first diagram,
by functoriality, it suffices to see $d_{E,\K}\circ (NL(n_E)\parr \K)=0.$ Applying step 2 to $g=n_E$, one gets $H=u_{NL(E)\parr \K}$ hence $dH= d_{E,\K}\circ H=0$ as expected thanks to axiom D1 of \cite{BuccEhrMan} giving $D(0)=0.$

For the second diagram, we compute ${\Contr E}{\Coder E}=(m^2_{E,E})^{-1}\circ \neg(NL(\Delta_{U(E)}))\circ \neg\Big((NL(u_E)\parr_\mathcal{C}\rho_{\neg E}^{-1})\circ d_{E,\K}\circ \rho_{NL(E)}\Big)
\circ {\lambda_E}$. We must compute $[NL(u_E)\parr_\mathcal{C}(\neg E\parr_\mathcal{C}\K) ]\circ d_{E,\K}\circ (NL(\Delta_{U(E)})\parr \K)$ using step 2 again with $g=\Delta_{U(E)}=U(\Delta_E),$ hence $H=\epsilon^{NL}_{E}\circ\neg (\Delta_E)=\epsilon^{NL}_{E}\circ\bigtriangledown_{\neg E}.$

Using \eqref{FirstFiore15} below, one gets $dH=( NL(n_E)\parr \bigtriangledown_{\neg E}\parr \K)\circ Isom$, so that, using $Isom\circ(\Delta_{U(E)}\times n_E)\circ \Delta_{U(E)}=\Delta_{U(E)} $, one obtains $$CD_E:=[NL(u_E)\parr_\mathcal{C}(\neg E\parr_\mathcal{C}\K) ]\circ d_{E,\K}\circ (NL(\Delta_{U(E)})\parr \K)=(NL(\Delta_{U(E)}\circ u_E)\parr\triangledown_{\neg E} \parr \K)\circ d_{E^2,\K}.$$
Hence,  noting that by naturality $NL(0)\parr\triangledown_{\neg E} \parr \K)=\triangledown_{NL(0)\parr\neg E\parr \K}$ and using the formula in step 1, \begin{align*}CD_E&\Lambda^{-1}_{E,E,\K}=\triangledown_{NL(0)\parr\neg E\parr \K}\bigoplus_{i=1,2}(NL( u_{E^2})\parr\pi_i \parr \K)\circ d_{E^2,\K}\Lambda^{-1}_{E,E,\K}\\&=\triangledown_{NL(0)\parr\neg E\parr \K}\Big[( NL(u_E)\parr_\mathcal{C}\Big[(NL(u_E)\parr \neg E\parr \K))d_{E,\K}\Big]) ,( \Big[(NL(u_E)\parr \neg E\parr \K))d_{E,\K}\Big]\parr_\mathcal{C}NL(u_E))\Big].\end{align*}

On the other hand, we can compute $(\Tens{\Coweak E}{\Coder E})\circ \lambda_E=\neg\Big((NL(u_E)\parr_\mathcal{C}\rho_{\neg E}^{-1})\circ\lambda^{-1} \circ( NL(u_E)\parr_\mathcal{C}d_{E,\K})\circ (NL(E)\parr_\mathcal{C}\rho_{NL(E)})\Big)
\circ\lambda_E.$
From the symmetric computation, one sees (in using $\neg$ is additive) that our expected equation reduces to proving the formula which reformulates our previous result: $$CD_E\circ \Lambda^{-1}_{E,E,\K}=\lambda^{-1} \circ( NL(u_E)\parr_\mathcal{C}\Big[(NL(u_E)\parr \neg E\parr \K))d_{E,\K}\Big])+\rho^{-1} \circ( \Big[(NL(u_E)\parr \neg E\parr \K))d_{E,\K}\Big]\parr_\mathcal{C}NL(u_E))$$

\begin{step} Final Diagrams for codereliction.\end{step}
To prove \eqref{FioreChainRule},\eqref{FioreStrength}, one can use \cite[Thm 4.1]{Fiore} (and the note added in proof making (14) redundant, but we could also check it in the same vein as below using step 2) and only check (16) and the second part of his diagram (15) on $\partial_E=\neg\Big((NL(E)\parr_\mathcal{C}\rho_{\neg E}^{-1})\circ d_{E,\K}\circ \rho_{NL(E)}\Big)$. Indeed, 
our choice ${\Coder E}=(\partial_E)\circ (\Tens{\Coweak E}{Id_{E}})\circ {\lambda_E}$ is exactly the direction of this bijection producing the codereliction.

 One must check:
 \[   \xymatrix @R=1.2em @C=3pc{
      {\Excl E}\o E\ar[r]^-{\partial_E}
     \ar[d]_-{{\Contr E}\o E}
     & \Excl E\ar[r]^-{\Digg E}
     & \Excl{\Excl E}\\
     {\Excl E}\o{\Excl E}\o E\ar[rr]^-{\Tens{\Digg E}{\partial_E}}
     &&\Tens{\Excl{\Excl E}}{{\Excl E}}\ar[u]_-{\partial_{\Excl E}}
     }
\]
and recall 
$ {\Contr E}=(m^2_{E,E})^{-1}\circ \neg(NL(\Delta_E)),\ $
${\Digg E}=\neg(NL(\delta_E)),$
and $(m^2_{E,E})^{-1}=\neg (\Lambda_{E,E}^{-1}(\epsilon^\neg\parr\epsilon^\neg))$,with $\Lambda_{E,E}^{-1}=\rho_{NL(E\times E)}\Lambda_{E,E,\K}^{-1}(NL(E)\parr \rho_{NL( E)})$, $\epsilon^\neg: \neg\neg E\to E$ the counit of self-adjunction.

Hence our diagram will be obtained by application of $\neg$ (after intertwining with $\rho$) if we prove:
\[   \xymatrix @R=1.2em @C=2pc{
      {NL(U(E))}\parr_\mathcal{C} (\neg E\parr_\mathcal{C}\K)
          & NL(U(E))\parr_\mathcal{C}\K\ar[l]_{\ \ \ \  d_{E,\K}}
     & \ar[l]_{NL(\delta_E)\parr_\mathcal{C}\K} NL(U({\Excl E}))\parr_\mathcal{C}\K\ar[d]_-{ d_{{\Excl E},\K}}\\
     {NL(U(E))}\parr_\mathcal{C}{NL(U(E))}\parr_\mathcal{C}  (\neg E\parr_\mathcal{C}\K)\ar[u]_-{\Big(
     NL(\Delta_E)\Lambda_{E,E}^{-1}
     \big)\parr_\mathcal{C} (\neg E\parr_\mathcal{C}\K)}
     &&{NL(U({\Excl E}))}\parr_\mathcal{C}{(\neg ({\Excl E})\parr_\mathcal{C}\K)}\ar[ll]^-{NL(\delta_E)\parr_\mathcal{C}{(d_{E,\K}}\circ \epsilon^\neg)}
     }
\]

This is the diagram in step 2 for $g=\delta_E$ if we see that $dH=(d_{E,\K}\circ \epsilon^\neg)$. For, it suffices to see $H= \epsilon^\neg$, which is essentially the way $\delta_E$ is defined as in proposition \ref{lem:evsmooth}.


We also need to check the diagram \cite[(16)]{Fiore}
which will follow if we check the (pre)dual diagram:
\begin{align*}
\xymatrix@C=50pt{
 NL(E\times E) \parr(\neg E \parr\K)
& 
NL( E) \parr (NL( E) \parr\K)\ar[l]_{\quad NL(E)\parr d_{E,\K}}&
NL( E\times E)\parr \K \ar[l]_{\quad \Lambda_{E,E,\K}}\\
& NL(E)\parr (\neg E \parr \K) \ar[ul]^{ \Lambda_{E,E,\neg E\parr \K}\circ(NL(U(\bigtriangledown_E))\parr(\neg E \parr \K))\qquad \qquad}& NL( E)\parr \K\ar[l]_{d_{E,\K}}\ar[u]_{ \ NL(U(\bigtriangledown_E))\parr \K}
}
\end{align*}

Using step 1 
and step 2 with $g=U(\bigtriangledown_E)$, it reduces to:
\begin{align*}
\xymatrix@C=60pt{
 NL(E^2) \parr(\neg E \parr\K)
&
NL( E^2)\parr((\neg E)^{\oplus 2})\parr \K)\ar[l]_{NL( E^2)\parr(\pi_2\parr \K)\quad}
\\
 NL(E)\parr (\neg E \parr \K) \ar[u]^{ (NL(U(\bigtriangledown_E))\parr(\neg E \parr \K))}\ar[r]^{NL(U(\bigtriangledown_E))\parr dH \qquad \qquad}& NL( E^2)\parr NL( E^2)\parr((\neg E)^{\oplus 2})\parr \K) \ar[u]_{(
     NL(\Delta_{E^2})\Lambda_{E^2,E^2}^{-1}
     )\parr_\mathcal{C} ((\neg E)^{\oplus 2}\parr_\mathcal{C}\K)}\qquad
}
\end{align*}
Recall that here, from step 2, $dH= d_{E^2,\K}\circ H .$ In our current case, we noticed that $H=\epsilon^{NL}_{E^2}\circ\neg (\bigtriangledown_E).$ Using 
\eqref{FirstFiore15} with $E^2$ instead of $E$, and $Isom\circ (id_{E^2}\times n_{E^2})\Delta_{E^2}=id_{E^2}$, the right hand side of the diagram we must check reduces to the map $NL(U(\bigtriangledown_E))\parr (\pi_2\circ\neg ( \bigtriangledown_E))\parr\K = NL(U(\bigtriangledown_E))\parr \neg E\parr\K$ as expected, using only the defining property of $\bigtriangledown_E$ from the coproduct.

Let us turn to proving the first diagram in \cite[(15)]{Fiore}, which will give at the end  $\Coder E \circ \Der E = Id_E .$
Modulo applying $\neg$ and intertwining with canonical isomorphisms, it suffices to see:
\begin{equation}\label{FirstFiore15}\xymatrix@C=80pt{
& NL(U( E))\parr \K \ar[dl]_{d_{E,\K}}
& \\
NL(U(E))\parr \neg E\parr \K  & NL(U(0))\parr \neg E\parr \K\ar[l]^{NL(n_E)\parr \neg E\parr \K\quad }&\neg E\parr \K \ar[ul]_{\epsilon^{NL}_E\parr \K
}\ar[l]_{\simeq}
}\end{equation}

For it suffices to get the diagram after precomposition by any $h:\neg \K\to \neg E$ (using the $\mathcal{D}$ is closed with unit $\neg\K$ for the closed structure). Since $E\in \mathcal{D}$ this is the same thing as $g=\neg h: E\simeq \neg\neg E\to \K$ so that one can apply naturality in $E$ of all the maps in the above diagram.
This reduces the diagram to the case $E=\K$.

But from axiom D3 of \cite{BuccEhrMan}, we have $D(Id_{U(E)}) = \pi_2$, projection on the second element of a pair, for $E \in \mathcal{C}$. When we apply the compatibility diagram between $D$ and $d_{E,E}$ to $\Der {\neg E}$, which corresponds through $\Xi$ to $Id_U(E)$, we have (for $\pi_2\in M(E\times E,E)$ the second projection): 
$$M\Big[ Isom\circ U\big(NL_E\parr_\mathcal{C}(\epsilon^{NL}_E \parr_\mathcal{C}E)\big) \circ U(d_{E,E}) \circ U((\epsilon^{NL}_E\parr \K)\circ(I_E))\Big] =   \pi_2$$

Here we used $I_E:\neg \K\to \neg E\parr E$ used from the axiom of dialogue categories corresponding via $\varphi$ to $id_{\neg E}$ and where we use { $M(\Xi_{E,E}\circ (U((\epsilon^{NL}_E\parr \K)\circ (I_E))))=Id_{U(E)}$. This comes via naturality for $\varphi$ from the association via $\varphi$ of $(\epsilon^{NL}_E\parr \K)\circ (I_E)$ to the map $\epsilon^{NL}_E:\neg E\to NL( E)$, and then from the use of the compatibility of $\Xi$ with adjunctions in definition \ref{LambdaModel} jointly with the definition of $\epsilon^{NL}$ as counit of adjunction, associating it to $Id_{U(E)}$}. Thus applying this to $E=\K$ and since we can always apply the faithful functors $U,M$ to our  relation and compose it with the monomorphism applied above after $U(d_{E,E})$ and on the other side to $U(I_\K)\simeq Id,$ it is easy to see that the second composition is also $\pi_2.$ 
\end{proof}

\subsection{A general construction for DiLL models}
\label{sec:GeneralDiLL}

Assume given the situation of Theorem \ref{lambdaTensortoLL}, with $\mathcal{C}$ having a biproduct structure with $U,\neg$ $\textbf{Mon}$-enriched and assume that $\mathcal{M}$ is actually given the structure of a differential $\lambda$-category  with operator internalized as a natural transformation $D_{E,F}:[E,F]\to [Diag(E),F]$ (so that $D$ in the definition of those categories is given by $M(D_{E,F}):\mathcal{M}(E,F)\to \mathcal{M}(E\times E,F)$ with $M$ the basic functor to sets of the closed category $\mathcal{M}$) and $U$ bijective on objects. Assume also that there is a map $D'_{E,F}:NL(E)\parr_\mathcal{C} F\to NL(E\times E)\parr_\mathcal{C} F$ in $\mathcal{C}$, {natural in $E$} such that 
$$\Xi_{E\times E,F}\circ U(D'_{E,F})= D_{U(E),U(F)}\circ \Xi_{E,F}.$$

and \begin{equation}\label{DprimeWeak}
D'_{E,F}=\Big(\rho_{NL(E^2)}^{-1}\circ D'_{E,\K}\circ \rho_{NL(E)}\Big)\parr_\mathcal{C} F.
\end{equation}
 Our non-linear variables are the first one after differentiation.
 
{We assume 
 $\parr_\mathcal{C}$ commutes with limits and finite coproducts  in $\mathcal{C}$ and recall from remark \ref{FaithfullnessLambdaTensor} that it preserves monomorphisms and that $\neg$ is faithful. Note that since $\mathcal{C}$ is assumed complete and cocomplete, it has coproducts $\oplus=\times$, by the biproduct assumption, and that $\neg(E\times F)=\neg(E)\oplus \neg(F)$ since $\neg: \mathcal{C}^{op}\to \mathcal{C}$ is left adjoint to its opposite functor  $\neg :\mathcal{C}\to \mathcal{C}^{op}$ which therefore preserves limits.} 
We will finally need the following :
\begin{equation}\label{epsionNL0}
\xymatrix@C=40pt{
\neg(E\times F)\parr_\mathcal{C} G\ar[d]|{\simeq}\ar[r]^{\ \ \ \ \epsilon^{NL}_{E\times F}\parr_\mathcal{C} G
  }  &  NL(E\times F)\parr_\mathcal{C} G\ar[rr]^{\qquad NL(U((id_E\times 0_F)\circ r))\parr_\mathcal{C} G\ \ \ \ \ \ \ }
&&NL(E)\parr_\mathcal{C} G\\
\Big(\neg E\parr_\mathcal{C} G\Big)\times\Big(\neg F\parr_\mathcal{C} G\Big) \ar[rrr]^{\pi_1} & & &\neg E\parr_\mathcal{C} G\ar[u]|{\epsilon^{NL}_{E}\parr_\mathcal{C} G } }
\end{equation}

This reduces to the case $G=\K$ by functoriality and then, this is a consequence of naturality of $\epsilon^{NL}$ since the main diagonal of the diagram taking the map via the lower left corner is nothing but $\neg((id_E\times 0_F)\circ r)$ with $r:E\to E\times 0$ the right unitor for the Cartesian structure on $\mathcal{C}.$

We want to build from that data a new category $\mathcal{M}_\mathcal{C}$ giving jointly with $\mathcal{C}$ the structure of a model of differential $\lambda$-tensor logic.

$\mathcal{M}_\mathcal{C}$ has the same objects as $\mathcal{M}$ (and thus as $\mathcal{C}$ too) but new morphisms that will have as derivatives maps from $\mathcal{C}$, or rather from its continuation category. Consider the category $Diff_{\N}$ with objects $\{0\}\times \N\cup \{1\} \times\N^*$ generated by the following family of morphisms without relations: one morphism $d=d_i:(0,i)\to (0,i+1)$ for all $i\in \N$ which will be mapped to a differential and one morphism $j=j_i:(1,i+1)\to (0,i+1)$ for all $i\in \N$ which will give an inclusion. Hence all the morphism are given by $d^k:(0,i)\to (0,i+k)$, $d^k\circ j:(1,i+1)\to (0,i+k+1)$.

We must define the new Hom set. We actually define an internal Hom. Consider, for $E,F\in\mathcal{C}$  the functor $Diff_{E,F},Diff_E:Diff_{\N}\to \mathcal{C}$ on objects by 
$$Diff_{E}((0,i))=NL(U(E)^{i+1})\parr_\mathcal{C} \K, \quad Diff_{E}((1,i+1))=(NL(U(E))\parr_\mathcal{C}((\neg E)^{\parr_\mathcal{C} i+1}\parr_\mathcal{C}\K))$$ with the obvious inductive definition $(\neg E)^{\parr_\mathcal{C} i+1}\parr_\mathcal{C}\K=\neg E\parr_\mathcal{C}\Big[(\neg E)^{\parr_\mathcal{C} i}\parr_\mathcal{C} \K \Big].$ Then we define $Diff_{E,F}=Diff_E\parr_\mathcal{C} F.$

The images of the generating morphisms are defined as follows: 
$$Diff_{E}(d_i)=\Lambda_{U(E)^2,U(E)^i,\K}^{-1}\circ(D'_{E,NL(U(E)^i)\parr_\mathcal{C}\K})\circ \Lambda_{U(E),U(E)^i,\K},$$
$$ Diff_{E}(j_{i+1})=\Big[\rho_{NL(U(E)^{i+1})}\circ\Lambda_{U(E),U(E)^{i+1},\K}^{-1}\circ \Big(NL(U(E))\parr_\mathcal{C}\Big[\Lambda_{U(E);i+1,\K}^{-1}\circ((\epsilon^{NL}_E)^{\parr_\mathcal{C} i+1}\parr_\mathcal{C} \K)\Big]\Big)\Big]
$$

where we wrote  $$\Lambda_{U(E);i+1,F}^{-1}=\Lambda_{U(E),U(E^i),F}^{-1}\circ\cdots\circ (NL(U(E))^{\parr_\mathcal{C} i-1}\parr_\mathcal{C}\Lambda_{U(E),U(E),F}^{-1})$$

Since $\mathcal{M}$ has all small limits, one can consider the limit of the functor $U\circ Diff_{E,F}$ and write it $[U(E),U(F)]_{\mathcal{C}}.$ Since $U$ bijective on objects, this induces a Hom set:  $$\mathcal{M}_\mathcal{C}(U(E),U(F))=M([U(E),U(F)]_{\mathcal{C}}).$$

We define $NL_\mathcal{C}(U(E))$ as the limit in $\mathcal{C}$ of $Diff_{E}$. Note that, since  $\parr_\mathcal{C}$ commutes with limits in $\mathcal{C}$, $NL_\mathcal{C}(U(E))\parr_\mathcal{C} F$ is the limit of $Diff_{E,F}=Diff_{E}\parr_\mathcal{C} F$. 

From the universal property of the limit, it comes with canonical maps $$D^k_{E,F}:NL_\mathcal{C}(U(E))\parr_\mathcal{C} F\to Diff_{E,F}((1,k)), \ \ \ j=j_{E,F}:NL_\mathcal{C}(U(E))\parr_\mathcal{C} F\to Diff_{E,F}((0,0)).$$

Note that $j_{E,F}=j_{E,\K}\parr_\mathcal{C} F$ is a monomorphism since for a pair of maps $f,g$ with target $NL_\mathcal{C}(U(E))\parr_\mathcal{C} F$, using that lemma \ref{epsilonNLMono} below implies that all $Diff_{E}(j_{i+1})$ are monomorphisms, one deduces that all the compositions with all maps of the diagram are equal, hence, by the uniqueness in the universal property of the projective limit, $f,g$ must be equal.

Moreover, since $U:\mathcal{C}\to \mathcal{M}$ is right adjoint to $\neg\circ NL$, it preserves limits, so that one gets an isomorphism $\Xi_{U(E),F}^{\mathcal{M}_\mathcal{C}}:U(NL_\mathcal{C}(U(E))\parr_\mathcal{C} F)\simeq U(\lim Diff_{E,F})\simeq [U(E),U(F)]_{\mathcal{C}}.$ 
It will remain to build $\Lambda^{\mathcal{M}_\mathcal{C}}$ but we can already obtain $d_{E,F}$. 
 
 We build it by the universal property of  limits, consider the maps (obtained using canonical maps for the monoidal category $\mathcal{C}^{op}$) $$D^{(1,k)}_{E,F}:\xymatrix@C=30pt{
NL_\mathcal{C}(U(E))\parr_\mathcal{C}F\ar[r]^{D^{k+1}_{E,F}\qquad \qquad \qquad } &  
(NL(U(E))\parr_\mathcal{C}((\neg E)^{\parr_\mathcal{C} k+1}\parr_\mathcal{C}\K))\parr_\mathcal{C}F\ar[r]^{\qquad \qquad \simeq } &Diff_{E,\neg E\parr_\mathcal{C}F}((1,k))
}$$
 $$J^1:\xymatrix@C=30pt{
NL_\mathcal{C}(U(E))\parr_\mathcal{C}F\ar[r]^{D^{1}_{E,F}\qquad \qquad } &  
(NL(U(E))\parr_\mathcal{C}((\neg E)\parr_\mathcal{C}\K))\parr_\mathcal{C}F\ar[r]^{ \qquad \simeq } &Diff_{E,\neg E\parr_\mathcal{C}F}((0,0))
}$$

Those maps extends uniquely to a cone enabling to get by the universal properties of limits our expected map:$d_{E,F}$. This required checking the identities $$Diff_{E,\neg E\parr_\mathcal{C}F}(d^k)\circ J^1=Diff_{E,\neg E\parr_\mathcal{C}F}(j_k)\circ D^{(1,k)}_{E,F}$$
that comes from $Diff_{E,F}(d^k\circ j_1)\circ D^{1}_{E,F}=Diff_{E,F}(j_{k+1})\circ D^{1+k}_{E,F}$ (by definition of $D^{1+k}_{E,F}$ as map coming from a limit) which is exactly the previous identity after composition with structural isomorphisms and $NL(E^{k+1})\parr_\mathcal{C}\epsilon^{NL}_E\parr_\mathcal{C}F$ which is a monomorphism, hence the expected identity, thanks to the next:

\begin{lemma}\label{epsilonNLMono}
In the previous situation, $\epsilon^{NL}_E$ is a monomorphism.
\end{lemma} 
\begin{proof}
 Since $\neg: \mathcal{C}^{op}\to \mathcal{C}$ is faithful, it suffices to see $\neg(\epsilon^{NL}_E):\neg(NL(U(E))\to \neg\neg E$ is an epimorphism. But its composition with the epimorphism $\neg\neg E\to E$, as counit of an adjunction with faithful functors $\neg$, is also the counit of  $\neg\circ NL$ with right adjoint $U$ which is faithful too, hence the composition is an epimorphism too. 
But $U(\neg\neg E)\simeq U(E)$ by the proof of Theorem \ref{lambdaTensortoLL}, thus $U(\neg(\epsilon^{NL}_E))$ is an epimorphism and $U$ is also faithful so reflects epimorphisms.
\end{proof}

\begin{theorem}\label{LLtoDiLL}
In the above situation, $(\mathcal{C}^{op},\parr_\mathcal{C},I,\neg, \mathcal{M}_\mathcal{C},\times ,0,[.,.]_{\mathcal{C}},NL_{\mathcal{C}},U,D,d)$ has a structure of Seely model of differential $\lambda$-tensor logic.
\end{theorem}

\begin{proof}\setcounter{Step}{0}
For brevity, we call $A_k=(1,k), k>0, A_0=(0,0)=B_0, B_k=(0,k)$

\begin{step} $\mathcal{M}_\mathcal{C}$ is a Cartesian (not full) subcategory of $\mathcal{M}$ and $U:\mathcal{C}\to\mathcal{M}_\mathcal{C}, NL_\mathcal{C}:\mathcal{M}_\mathcal{C}\to \mathcal{C}$ are again  functors, the latter being right $\neg$-relative adjoint of the former.\end{step}

Fix $$g\in \mathcal{M}_\mathcal{C}(U(E),U(F))
=\mathcal{C}(\neg K, NL_\mathcal{C}(U(E))\parr_\mathcal{C}F)\to \mathcal{C}(\neg K, NL(U(E))\parr_\mathcal{C}F)\simeq\mathcal{C}(\neg F, NL(U(E)))\ni d^0g.$$
Similarly, composing with $D^k_{E,F}$ one obtains: 
$$d^kg\in \mathcal{C}(\neg K, Diff_{E,\K}(A_k)\parr_\mathcal{C}F)\simeq\mathcal{C}(\neg F, Diff_{E,\K}(A_k).$$
 We first show that $NL(g)=\cdot\circ g:NL(U(F))\to NL(U(E))$ induces via the monomorphisms $j$ a map $NL_\mathcal{C}(g):NL_\mathcal{C}(F)\to NL_\mathcal{C}(E)$ such that $j_{E,\K}NL_\mathcal{C}(g)=NL(g)j_{F,\K}$. This relation already determines at most one $NL_\mathcal{C}(g)$, one must check such a map exists in using the universal property for $NL_\mathcal{C}(E)$. We must build maps: $$NL_\mathcal{C}^{k}(g): NL_{\mathcal{C}}( F)\to Diff_{E,\K}(A_k)$$
with $NL_\mathcal{C}^{0}(g)=NL(g)j_{F,\K}$ satisfying the relations for $k\geq 0$ ($j_0=id$):
\begin{equation}\label{NLC-k}Diff_{E,\K}(d_k\circ j_k)\circ NL_\mathcal{C}^{k}(g)=Diff_{E,\K}(j_{k+1})\circ NL_\mathcal{C}^{k+1}(g).\end{equation}

An abstract version of Fa\`{a} di Bruno's formula will imply the form of $NL_\mathcal{C}^{k}(g),$ that we will obtain it as sum of $NL_\mathcal{C}^{k,\pi}(g):Diff_{F,\K}(A_{|\pi|})\to Diff_{E,\K}(A_{k})$ for $\pi=\{\pi_1,...,\pi_{|\pi|}\}\in P_k$ the set of partitions of $[\![1,k]\!].$ We define it as $$NL_\mathcal{C}^{k,\pi}(g)=(NL(\Delta^{|\pi|+1}_E)\parr_\mathcal{C} Id)\circ IsomAss_{|\pi|}\circ [d^0g\parr_\mathcal{C}d^{|\pi_1|}g\parr_\mathcal{C}\cdots \parr_\mathcal{C}d^{|\pi_{|\pi|}|}g\parr_\mathcal{C}id_\K]$$ with $IsomAss_k:NL(E)\parr_\mathcal{C}(NL(E)\parr_\mathcal{C}E_1)\parr_\mathcal{C}
\cdots\parr_\mathcal{C}(NL(E)\parr_\mathcal{C}E_k)\simeq NL(E^{k+1})\parr_\mathcal{C}(E_1\parr_\mathcal{C}
\cdots\parr_\mathcal{C}E_k)$ and $\Delta_k:E \to E^k$ the diagonal of the Cartesian category $\mathcal{C}$. 
We will compose it with $d^{P_k}:NL_{\mathcal{C}}( F)\to \prod_{\pi\in P_k} Diff_{F,\K}(A_{|\pi|})$ given by the universal property of product composing to $d^{|\pi|}$ in each projection.

Then using the canonical sum map $\Sigma_E^k:\prod_{i=1}^k E\simeq \oplus_{i=1}^k E\to E$ obtained by universal property of coproduct corresponding to identity maps, one can finally define the map inspired by Fa\`{a} di Bruno's Formula : 
$$NL_\mathcal{C}^{k}(g)=\Sigma_{Diff_{E,\K}(A_{k})}^{|P_k|}\circ\Big(\prod_{\pi\in P_k}NL_\mathcal{C}^{k,\pi}(g) \Big)\circ d^{P_k}.$$
Applying $U$ and composing with $\Xi$, \eqref{NLC-k} is then obtained in using the chain rule D5 on the inductive proof of Fa\`{a} di Bruno's Formula, using also that $U$ is additive. 

Considering $NL_\mathcal{C}(g)\parr_\mathcal{C} G:NL_\mathcal{C}(F)\parr_\mathcal{C} G\to NL_\mathcal{C}(E)\parr_\mathcal{C} G$ which induces a composition on $\mathcal{M}_\mathcal{C}$, one gets that $\mathcal{M}_\mathcal{C}$ is a subcategory of $\mathcal{M}$ (from the agreement with previous composition based on intertwining with $j$) as soon as we see $id_{U(E)}\in \mathcal{M}_\mathcal{C}(U(E),U(E))$. This boils down to building a map in $\mathcal{C}$, $I_{\mathcal{M}_\mathcal{C}}:\neg(\K)\to  NL_\mathcal{C}(U(E))\parr_\mathcal{C}E$ using the universal property such that $j_{E,E}\circ I_{\mathcal{M}_\mathcal{C}}=I_{\mathcal{M}}:\neg(\K)\to  NL(U(E))\parr_\mathcal{C}E$ corresponds to identity map. We define it in imposing $D^k_{E,E}\circ I_{\mathcal{M}_\mathcal{C}}=0$ if $k\geq 2$ and $$D^1_{E,E}\circ I=NL(0_E)\parr_\mathcal{C} i_\mathcal{C}:\neg(\K)\simeq NL(0)\parr_\mathcal{C}\neg(\K)\to NL(E)\parr_\mathcal{C}\neg E\parr_\mathcal{C}E$$  with $i_\mathcal{C}\in \mathcal{C}(\neg(\K), \neg E\parr_\mathcal{C}E)\simeq\mathcal{C}(\neg E, \neg E)$ corresponding to identity via the compatibility for the dialogue category $(\mathcal{C}^{op},\parr_\mathcal{C},\K,\neg)$. This satisfies the compatibility condition enabling to define a map by the universal property of limits because of axiom D3 in \cite{BuccEhrMan} implying (recall our linear variables are in the right contrary to theirs) $D(Id_{U(E)})=\pi_2, D(\pi_2)=\pi_2\pi_2$ (giving vanishing starting at second derivative via D-curry) and of course $(\epsilon^{NL}_E\parr_\mathcal{C}E)\circ i_{\mathcal{C}}=I_{\mathcal{M}}$ from the adjunction defining $\epsilon^{NL}$. 

As above we can use known adjunctions to get the  isomorphism \begin{align}\label{NLCadjunction}\begin{split}\mathcal{M}_\mathcal{C}(U(E),U(F))&=\mathcal{M}(0,[U(E),U(F)]_\mathcal{C})\simeq\mathcal{C}^{op}(NL(0),\neg(NL_\mathcal{C}(U(E))\parr_\mathcal{C} F))\\&\simeq \mathcal{C}(\neg( \K),NL_\mathcal{C}(U(E))\parr_\mathcal{C} F)\simeq
 \mathcal{C}^{op}(NL_\mathcal{C}(U(E)),\neg F)\end{split}\end{align}
where the last isomorphism is the compatibility for the dialogue category $(\mathcal{C}^{op},\parr_\mathcal{C},\K,\neg)$.
Hence the map $id_{U(E)}$ we have just shown to be in the first space gives $\epsilon^{NL_\mathcal{C}}_E:\neg E\to NL_\mathcal{C}(U(E))$ with $j_{E,\K}\circ \rho_{NL_\mathcal{C}(U(E))}\circ \epsilon^{NL_\mathcal{C}}_E=\rho_{NL(U(E))}\circ \epsilon^{NL}_E.$

Let us see that $U$ is a functor too. Indeed $\epsilon^{NL_\mathcal{C}}_E\parr_\mathcal{C} F:\neg E\parr_\mathcal{C} F\to NL_\mathcal{C}(U(E))\parr_\mathcal{C} F$ can be composed with the adjunctions and compatibility for the dialogue category again to get:
$$\mathcal{C}(E,F)\to \mathcal{C}(\neg\neg E,F)\simeq\mathcal{C}(\neg \K,\neg E\parr_\mathcal{C} F)\to \mathcal{C}(\neg NL(0),NL_\mathcal{C}(U(E))\parr_\mathcal{C} F),$$
the last space being nothing but $\mathcal{M}_\mathcal{C}(U(E),U(F))=\mathcal{M}(0,[U(E),U(F)]_\mathcal{C})$ giving the wanted $U(g)$ for $g\in \mathcal{C}(E,F)$ which is intertwined via $j$ with the $\mathcal{M}$ valued one, hence $U$ is indeed a functor too. The previous equality is natural in $F$ via the intertwining with $j$ and the corresponding result for $\mathcal{M}$.

Now one can see that \eqref{NLCadjunction} is natural in $U(E),F$. For it suffices to note that the first equality is natural by definition and all the following ones are already known. Hence the stated $\neg$-relative adjointeness.

This implies $U$ preserve products as right adjoint of $\neg\circ NL_\mathcal{C}$, hence the previous products $U(E)\times U(F)=U(E\times F)$ are still products in the new category, and the category $\mathcal{M}_\mathcal{C}$ is indeed Cartesian.

\begin{step} Curry map\end{step}
It remains a few structures to define, most notably the internalized Curry map: $\Lambda^{\mathcal{C}}_{E,F,G}: NL_{\mathcal{C}}(E\times F)\parr_\mathcal{C}G\to NL_{\mathcal{C}}(E)\parr_\mathcal{C}\Big( NL_{\mathcal{C}}(F)\parr_\mathcal{C}G\Big)$.
We use freely the structure isomorphisms of the monoidal category $\mathcal{C}$. 

For we use the universal property of limits as before, we need to define:
 $$\Lambda^{k}_{E,F,G}: NL_{\mathcal{C}}(E\times F)\parr_\mathcal{C}G\to Diff_{E,(NL_{\mathcal{C}}(F)\parr_\mathcal{C}G)}(A_k)$$
satisfying the relations for $k\geq 0$ ($j_0=id$):
\begin{equation}\label{Lambda-k}Diff_{E,(NL_{\mathcal{C}}(F)\parr_\mathcal{C}G)}(d_k\circ j_k)\circ \Lambda^{k}_{E,F,G}=Diff_{E,(NL_{\mathcal{C}}(F)\parr_\mathcal{C}G)}(j_{k+1})\circ \Lambda^{k+1}_{E,F,G}.\end{equation}
Since $Diff_{E,(NL_{\mathcal{C}}(F)\parr_\mathcal{C}G)}(A_k)\simeq NL_{\mathcal{C}}(F)\parr_\mathcal{C}\Big(NL(E)\parr_\mathcal{C}(\neg E)^{\parr_\mathcal{C} k}\parr_\mathcal{C}\K\Big)\parr_\mathcal{C}G$ we use again the same universal property to define the map $\Lambda^{k}_{E,F,G}$ and we need to define :
$$\Lambda^{k,l}_{E,F,G}:NL_{\mathcal{C}}(E\times F)\parr_\mathcal{C}G\to Diff_{F,\Big(NL(E)\parr_\mathcal{C}(\neg E)^{\parr_\mathcal{C} k}\parr_\mathcal{C}\K\Big)\parr_\mathcal{C}G}(A_l).$$
satisfying the relations:
\begin{equation}\label{Lambda-kl}Diff_{F,\Big(NL(E)\parr_\mathcal{C}(\neg E)^{\parr_\mathcal{C} k}\parr_\mathcal{C}\K\Big)\parr_\mathcal{C}G}(d_l\circ j_l)\circ \Lambda^{k,l}_{E,F,G}=Diff_{F,\Big(NL(E)\parr_\mathcal{C}(\neg E)^{\parr_\mathcal{C} k}\parr_\mathcal{C}\K\Big)\parr_\mathcal{C}G}(j_{l+1})\circ \Lambda^{k,l+1}_{E,F,G}.\end{equation}
But we can consider the map:$$D^{k+l}_{E\times F,G}:NL_{\mathcal{C}}(E\times F)\parr_\mathcal{C}G\to(NL(U(E\times F))\parr_\mathcal{C}((\neg (E\times F))^{\parr_\mathcal{C} k+l}\parr_\mathcal{C}\K))\parr_\mathcal{C}G $$

Let us describe an obvious isomorphism of the space of value to extract the component we need. First, using the assumptions on $\parr_\mathcal{C}$ and $\neg$:
\begin{align*}((\neg (E_1\times E_2))^{\parr_\mathcal{C} k+l}\parr_\mathcal{C}\K))\parr_\mathcal{C}G&\simeq\bigoplus_{i:[\![1,k+l ]\!]\to \{1,2\}} (\neg E_{i_1}\parr_\mathcal{C} \neg E_{i_2}\parr_\mathcal{C} \cdots \parr_\mathcal{C} \neg E_{i_{k+l}})\parr_\mathcal{C} G\\&\simeq \bigoplus_{i:[\![1,k+l ]\!]\to \{1,2\}
} (\neg(E_1)^{\parr_\mathcal{C} (\# f^{-1}(\{1\}))}\parr_\mathcal{C} \neg(E_2)^{\parr_\mathcal{C}( \# i^{-1}(\{2\}))})\parr_\mathcal{C} G.\end{align*}
Hence using also $\Lambda_{E,F,.}$ one gets:
\begin{align*}&\Lambda:(NL(U(E\times F))\parr_\mathcal{C}((\neg (E\times F))^{\parr_\mathcal{C} k+l}\parr_\mathcal{C}\K))\parr_\mathcal{C}G\\&\simeq \bigoplus_{i:[\![1,k+l ]\!]\to \{1,2\}
} Diff_F(A_{\# i^{-1}(\{2\}))})
\parr_\mathcal{C}\Big(Diff_E(A_{\# i^{-1}(\{1\}))})
\parr_\mathcal{C} G\Big).\end{align*}
Composing with $P_{k,l}$ a projection on a term with $\# i^{-1}(\{1\})=k$, one gets the map $P_{k,l}\circ \Lambda\circ D^{k+l}_{E\times F,G}=\Lambda^{kl}_{E,F,G}$ we wanted. One could check this does not depend on the choice of term using axiom (D7) of Differential Cartesian categories giving an abstract Schwarz lemma, but for simplicity we choose $i(1)=\cdots =i(l)=2$ which corresponds to differentiating all variables in $E$ first and then all variables in $F$. The relations we want to check will follow from axiom (D-curry) of differential $\lambda$-categories.

Then to prove the relation \eqref{Lambda-kl} we can prove it after composition by a $\Lambda$ (hence the left hand side ends with application of  $D'_{F,NL(U(F)^l)\parr_\mathcal{C}\K}\parr_\mathcal{C} \Big(Diff_E(A_{k})\parr_\mathcal{C}G\Big)$). We can then apply $Diff_F(B_{l})
\parr_\mathcal{C}\Big(Diff_E(j_k)
\parr_\mathcal{C} G\Big)$ 
 which is a monomorphism and obtain, after decurryfying and applying $U$ and various $\Xi$, maps in $[U(F)^2\times U(F)^l\times U(E)^{k+1},U(G)],$ and finally only prove equality there, the first variable $F$ being a non-linear one. 

 Of course, we start from $Diff_{E\times F,G}(d_{k+l}\circ j_{k+l})\circ D^{k+l}_{E\times F,G}=Diff_{E\times F,G}(j_{k+l+1})\circ D^{k+l+1}_{E\times F,G}$ 
 and use an application of \eqref{epsionNL0}:\begin{align*}&\Big[Diff_F(B_{l})
\parr_\mathcal{C}\Big(Diff_E(j_k)
\parr_\mathcal{C} G\Big)\Big]\circ Diff_{F,Diff_E(A_{k})
\parr_\mathcal{C}G}(j_{l})\circ P_{k,l}\circ \Lambda\\&=Isom\circ \Lambda_{E,F,Diff_F(B_{l})
\parr_\mathcal{C}\big(Diff_E(B_{k})\parr_\mathcal{C} G\big)}\circ NL(0_{l,k})\circ Diff_{E\times F,G}(j_{k+l})\end{align*}
 with $0_{l,k}:U(E\times F)\times U(F)^l\times U(E)^k\simeq U(E\times F)\times U(0\times F)^l\times U(E\times 0)^k\to U(E\times F)^{k+l+1}$
  the map corresponding to $id_{E\times F}\times(0_E\times id_F)^l \times (id_E\times 0_F)^k$. We thus need the following commutation relation:
  \begin{align*}
 NL(0_{l+1,k})\circ Diff_{E\times F,G}(d_{k+l})=Isom\circ NL(0_{1,0})\circ(D'_{E\times F, \big(Diff_{F}(B_{l})\parr_\mathcal{C}Diff_{E}(B_{k})\parr_\mathcal{C}G\big) })\circ
NL(0_{l,k})\end{align*}
  This composition $NL(0_{1,0})\circ D'_{E\times F,.}$ gives exactly after composition with some $\Xi$ the right hand side of (D-curry), hence composing all our identities, and using canonical isomorphisms of $\lambda$-models of $\lambda$-tensor logic, and  this relation gives the expected \eqref{Lambda-kl} at the level of $[U(F)^2\times U(F)^l\times U(E)^{k+1},U(G)]$.
  
Let us turn to checking  \eqref{Lambda-k}. It suffices to check it after composition with the monomorphism $Diff_E(B_k)\parr_\mathcal{C}j_{F,G}$. Then the argument is the same as for \eqref{Lambda-kl} in the case $k=0$ and with $E$ and $F$ exchanged. The inverse of the Curry map is obtained similarly.
  
\begin{step} $\mathcal{M}_\mathcal{C}$ is a differential $\lambda$-category.\end{step}

We first need to check that $\mathcal{M}_\mathcal{C}$ is Cartesian closed, and we already know it is Cartesian. Since we defined the internalized curry map and $\Xi$ one can use the first compatibility diagram in the definition \ref{LambdaModel} to define $\Lambda^{\mathcal{M}_\mathcal{C}}$. To prove the defining adjunction of exponential objects for Cartesian closed categories, it suffices to see naturality after applying the basic functor to sets $M$. From the defining diagram, naturality in $E,F$ of $\Lambda^{\mathcal{M}_\mathcal{C}}:[E\times F,U(G)]_\mathcal{C}\to [E,[F,U(G)]_\mathcal{C}]_\mathcal{C}$ will follow if one checks the naturality of $\Xi_{E,F}^\mathcal{C}$ and 
$\Lambda^{\mathcal{C}}_{E,F,G}$ that we must check anyway while naturality in $U(G)$ and not only $G$ will have to be considered separately.

For $\Lambda^{\mathcal{C}-1}_{E,F,G}$, take $e:E\to E', f:F\to F', g:G'\to G$ the first two in $\mathcal{M}_\mathcal{C}$ the last one in $\mathcal{C}$. We must see $\Lambda^{\mathcal{C}-1}_{E,F,G}\circ[ NL_\mathcal{C}(e)\parr_\mathcal{C} (NL_\mathcal{C}(f)\parr_\mathcal{C} g)]= [NL_\mathcal{C}(e\times f)\parr_\mathcal{C} g] \circ \Lambda^{\mathcal{C}-1}_{E',F',G'}$ and it suffices to see equality after composition with the monomorphism $j_{U^{-1}(E\times F),G}: NL_\mathcal{C}(E\times F)\parr_\mathcal{C}G\to NL(E\times F)\parr_\mathcal{C}G$. But by definition, $j_{U^{-1}(E\times F),G}\Lambda^{\mathcal{C}-1}_{E,F,G}=\Lambda^{-1}_{E,F,G}(NL(E)\parr j_{U^{-1}(F),   G} )j_{U^{-1}(E),   NL_\mathcal{C}(F)\parr G}$ and similarly for $NL_\mathcal{C}$ functors which are also induced from $NL$, hence the relation comes from the one for $\Lambda$ of the original model of $\lambda$-tensor logic we started with. The reasoning is similar with $\Xi$.
Let us finally see that $M(\Lambda^{\mathcal{M}_\mathcal{C}})$
is natural in $U(G)$, but again from step 1 composition with a map $g\in \mathcal{M}_\mathcal{C}(U(G),U(G'))\subset \mathcal{M}(U(G),U(G'))$ is induced by the one from $\mathcal{M}$ and so is $\Lambda^{\mathcal{M}_\mathcal{C}}$ from $\Lambda^{\mathcal{M}}$ in using the corresponding diagram for the original model of $\lambda$-tensor logic we started with and all the previous induced maps for $\Xi,\Lambda^\mathcal{C}$. Hence also this final naturality in $U(G)$ is induced.

Having obtained the adjunction for a Cartesian closed category, we finally see that all the axioms D1--D7 of Cartesian differential categories in \cite{BuccEhrMan} and D-Curry is also induced. Indeed, our new operator $D$ is also obtained by restriction as well as the left additive structure. Note that as a consequence the new $U$ is still a $\textbf{Mon}$-enriched functor.

\begin{step} $(\mathcal{M}_\mathcal{C},\mathcal{C})$ form a  $\lambda$-categorical model of $\lambda$-tensor logic and  Conclusion.\end{step}

We have already built all the data for definition \ref{LambdaModel}, and shown $\neg$-relative adjointness in step 1. It remains to see the four last compatibility diagrams.

But from all the naturality conditions for canonical maps of the monoidal category, one can see them after composing with monomorphisms $NL_\mathcal{C}\to NL$ and induce them from the diagrams for $NL$.

Among all the data needed in definition \ref{DiffLambdaTensorDef}, it remains to build the internalized differential $D^\mathcal{C}_{E,F}$ for $D$ in $\mathcal{M}^\mathcal{c}$ and see the two compatibility diagrams there. From the various invertible maps, one can take the first diagram as definition of $D^\mathcal{C}_{U(E),U(F)}$
and must see that, then $M(D^\mathcal{C}_{U(E),U(F)})$ is indeed the expected restriction of $D$. Let $j_\mathcal{M}^{E,F}:[U(E),U(F)]_\mathcal{C}\to [U(E),U(F)]$ the monomorphism. It suffices to see 
$j_\mathcal{M}^{E\times E,F}\circ D^\mathcal{C}_{U(E),U(F)}\circ \Xi^{\mathcal{C}}_{E,F}= \Xi_{E\times E,F}\circ U(D'_{E,F}\circ j_{E,F})$ {(note that this also gives the naturality in $E,F$ of $d$ from the one of $D'$)}. Hence from the definition of $D^\mathcal{C}$, it suffices to see the following diagram:
\[
\xymatrix@C=20pt{
U\Big(NL_\mathcal{C}(U(E))\parr_\mathcal{C} F\Big)\ar[d]_{\Xi_{E\times E,F}\circ U(D'_{E,F}\circ j_{E,F})} \ar[rrr]^{U(NL_\mathcal{C}(U(E)))\parr_\mathcal{C}(\epsilon^{NL_\mathcal{C}}_E \parr_\mathcal{C}F))\circ U(d_{E,F})\ \ \ \ } &&&U\Big(NL_\mathcal{C}(U(E))\parr_\mathcal{C}(NL_\mathcal{C}(U(E)) \parr_\mathcal{C}F)\Big)\ar[d]|{[id_{U(E)},\Xi_{E,F}^\mathcal{C}]_\mathcal{C}\circ \Xi_{E,NL_\mathcal{C}(U( E) )\parr_\mathcal{C}F}^\mathcal{C}} \\
[U(E\times E),U(F)] &[U(E\times E),U(F)]_\mathcal{C}\ar[l]_{j_\mathcal{M}^{E\times E,F}}&&[U(E),[U(E),U(F)]_\mathcal{C}]_\mathcal{C}\ar[ll]_{(\Lambda^{\mathcal{M}_\mathcal{C}}_{U(E), U(E),U(F)})^{-1}}.
}
\]
First we saw from induction of our various maps that 
the right hand side of the diagram can be written without maps with index $\mathcal{C}$:
$$(\Lambda^{\mathcal{M}}_{U(E), U(E),U(F)})^{-1}\circ [id_{U(E)},\Xi_{E,F}]\circ \Xi_{E,NL( E )\parr_\mathcal{C}F}\circ  U(NL(U(E)))\parr_\mathcal{C}(\epsilon^{NL}_E \parr_\mathcal{C}F))\circ U(j_{E,\neg E \parr_\mathcal{C}F}\circ d_{E,F}).$$

The expected diagram now comes the definition of $d_{E,F}$ by universal property which gives 
$ j_{E,\neg E\parr_\mathcal{C}F}\circ
 d_{E,F}=J^1=Isom\circ D_{E,F}^1$ and similarly $Diff_{E,F}(j_1)\circ D_{E,F}^1=D'_{E,F}\circ j_{E,F}$ so that composing the above diagrams (and an obvious commutation of the map involving $\epsilon^{NL}$ through various natural isomorphisms) gives the result.
 
 For the last diagram in definition \ref{DiffLambdaTensorDef}, since $ j=j_{E,\neg E\parr_\mathcal{C}F}$ is a monomorphism, it suffices to compose $d_{E,F}$ and the  equivalent map stated in the diagram by $j$ and see equality, and from the recalled formula above reducing it to $D'_{E,F}$, this reduces to \eqref{DprimeWeak}.
\end{proof}
\subsection{$\rho$-smooth maps as model of DiLL}

Our previous categories from Theorem \ref{th:CRefSeely}
 cannot give a model of DiLL with $\Cref_\infty$ as category with smooth maps. If one wants to obtain a differential map since the map won't be with  value in $E\multimap_\mathscr{C} F$ but in spaces of bounded linear maps $L_{bd}(E,F)$. We will have to restrict to maps with iterated differential valued in $E^{\otimes_\mathscr{C} k}\multimap_\mathscr{C} F:=E\multimap_\mathscr{C}(\cdots (E\multimap_\mathscr{C} F)\cdots ).$ This is what we did abstractly in the previous subsection that will enable us to obtain efficiently a model.
 
 \begin{lemma}
 The categories of  Theorem \ref{th:CRefSeely} satisfy the assumptions of subsection \ref{sec:GeneralDiLL}.
  \end{lemma}
\begin{proof}
We already saw in Theorem \ref{th:CRefSeely} that the situation of Theorem \ref{lambdaTensortoLL} is satisfied with dialogue category $\mathcal{C}=(\mathscr{C}-\Mc^{op},\epsilon,\K,(\cdot)^*_{\mathscr{C}})$, and $\mathcal{M}=\mathscr{C}-\Mc_\infty$. We already know that $\varepsilon$-product commutes with limits and monomorphisms and the biproduct property is easy. The key is to check that we have an internalized derivative. From \cite{KrieglMichor} we know that we have a derivative $d:\CinC(E,F)\to \CinC(E,L_b(E,F))$ and the space of bounded linear maps $L_{bd}(E,F)\subset \CinC(E,F)$ the set of conveniently smooth maps. Clearly, the inclusion is continuous since all the images by curves of compact sets appearing in the projective kernel definition of $\CinC(E,F)$ are bounded. Thus one gets  $d:\CinC(E,F)\to \CinC(E,\CinC(E,F))\simeq\CinC(E\times E,F)$. It remains to see continuity. For by the projective kernel definition, one must check that for $c=(c_1,c_2)\in \Cin_{co}(X,E\times E)$, $f\mapsto df\circ (c_1,c_2)$ is continuous $\CinC(E,F)\to \Cin_{co}(X,F)$. But consider the curve $c_3:X\times X\times \R\to E$ given by $c_3(x,y,t)=c_1(x)+c_2(y)t$, since $X\times X\times \R\in \mathscr{C}$, we know that $f\circ c_3$ is smooth and $\partial_t (f\circ c_3)(x,y,0)=df(c_1(x))(c_2(y))$  and its derivatives in $x,y$ are controlled by the seminorms for $\CinC(E,F)$, hence the stated continuity. It remains to note that $\mathcal{M}$ is a differential $\lambda$-category since we already know it is Cartesian closed and all the properties of derivatives are well-known for conveniently smooth maps. For instance, the chain rule D7 is \cite[Thm 3.18]{KrieglMichor}.
\end{proof}

Concretely, one can make explicit the stronger notion of smooth maps considered in this case.

We thus consider $d^k$ the iterated (convenient) differential giving $d^k:\CinC(E,F)\to \CinC(E,L_{bd}(E^{\otimes_\beta k},F)).$ Since $E^{\otimes_\mathscr{C} k}\multimap_\mathscr{C} F$ is a subspace of $L_{bd}(E^{\otimes_\beta k},F)$ (unfortunately this does not seem to be in general boundedly embedded), we can consider :
$$C^\infty_{\Cref}(E,F):=\Big\{u\in \CinC(E,F): \forall k\geq 1: d^k(u)\in \CinC(E,E^{\otimes_\mathscr{C} k}\multimap_\mathscr{C} F)\Big\}.$$

\begin{remark}\noindent A map $ f \in C^\infty_{\Cref}(E,F)$ will be called \emph{$\Cref$-smooth}.  In the case $\mathscr{C}=Ban$, we say $\rho$-smooth maps, associated to the category $\rho$-Ref, and write $C^\infty_{\rho}=  C^\infty_{Ban-\mathbf{Ref}}$. Actually, for $Fin\subset\mathscr{C}\subset \FDFS$, from the equivalence of $*$-autonomous categories in Theorem  \ref{th:CRef}, and since the inverse functors keep the bornology of objects, hence don't change the notion of conveniently smooth maps, we have algebraically $$C^\infty_{\Cref}(E,F)=C^\infty_{\rho}(\mathscr{S}(E_\mu),\mathscr{S}(F_\mu)).$$
Hence, we only really introduced one new notion of smooth maps, namely, $\rho$-smooth maps. Of course, the topologies of the different spaces differ. 
\end{remark}



Thus $d^k$ induces a map $C^\infty_{\Cref}(E,F)\to \CinC(E,E^{\otimes_\mathscr{C} k}\multimap_\mathscr{C} F)$ ($d^0=id$) and we equip $C^\infty_{\Cref}(E,F)$ with the corresponding locally convex kernel topology $\mathrm{K}_{n,\geq 0}(d^n)^{-1}(C^\infty(E,E^{\otimes n}\multimap F))$ with the notation of \cite{Kothe} and the previous topology given on any $C^\infty(E,E^{\otimes k}\multimap F)$.\footnote{ This definition is quite similar to one definition (for the corresponding space of value $E^{\otimes k}\multimap F$ which can be interpreted as a space of hypocontinuous multilinear maps for an appropriate bornology) in \cite{Meise} except that instead of requiring continuity of all derivatives, we require their smoothness in the sense of Kriegl-Michor.}

We call $\Cref_{\infty\Cref}$ the category of $\mathscr{C}$-reflexive spaces with $C^\infty_{\Cref}$ as spaces of maps. Then from section \ref{sec:GeneralDiLL} we even have an induced $d:C^\infty_{\Cref}(E,F)\to C^\infty_{\Cref}(E,E^{\otimes_\mathscr{C} k}\multimap_\mathscr{C} F).$

Let us call 
$d_0(f)=df(0)$ so that $d_0:C^\infty_{\Cref}(E,F)\to C^\infty_{\Cref}(E,F)$ is continuous. Recall also that we introduced $\partial_E={\Cocontr E}\circ (\Tens{!E}{\Coder E})$ and dually $\overline{\partial}_E= (\Tens{!E}{\Der E}){\Contr E}:!E\to !E\o E.$
We conclude to our model:

\begin{theorem}\label{th:CRefDiLL}
Let $Fin\subset\mathscr{C}\subset \FDFS$
 as above. $\Cref$ is also a Seely category with biproducts with structure extended by the comonad $!_{\Cref}(\cdot)=(C_{\Cref}^\infty(\cdot))^*_\mathscr{C}$ associated to the adjunction with left adjoint $!_{\Cref}:\Cref_{\infty-\Cref}\to \Cref$ and right adjoint $U$. It gives a model of DiLL with codereliction $(d_0)^*_\mathscr{C}$. 
\end{theorem}
\begin{proof}
This is a combination of Theorem \ref{LLtoDiLL}, \ref{DifflambdaTensortoDiLL} and the previous lemma.
\end{proof}

\begin{remark}
One can check that
\begin{enumerate}
 \item for any $E\in \Cref,$ $\partial_E\overline{\partial}_E+id_E$ is invertible,
 \item The model is Taylor in the sense of \cite[3.1]{Ehrhard16}, i.e. for any $f_1,f_2:!_{\Cref}E\to F$ if $f_1\partial_E=f_2\partial_E$ then $f_1+f_2{\Coweak E}{\Weak E}=f_2+f_1{\Coweak E}{\Weak E}$.
 \end{enumerate}
Indeed, the Taylor property is obvious since $df_1=df_2$ in the convenient setting implies the same G\^ateaux derivatives, hence $f_1+f_2(0)=f_2+f_1(0)$ on each line hence everywhere.

For (1), we define the inverse by $(I_E)^*_\mathscr{C}$ with $I_E:\CinC(E)\to \CinC(E)$ as in \cite[3.2.1]{Ehrhard16} by $I_E(f)(x)=\int_0^1f(tx) dt$, which is a well-defined weak Riemann integral by Mackey-completeness of the space \cite{KrieglMichor}.

By \cite{Ehrhard16}, the two conditions reformulate the two fundamental theorems of calculus. See also \cite{CockettLemay} for a further developments on the two conditions above. 
\end{remark}

\begin{remark}
Let us continue our comparison of subsection \ref{sec:ConvenientSetting}. Let us see that if $E,F\in \Conv$, $C^\infty_{\Cref}(E,F)= \CinC(E,F)$ so that we didn't introduce a new class of smooth maps for convenient vector spaces. Our notion of smoothness turning our model into a model of DiLL is only crucial on the extra-spaces we added to get a $*$-autonomous category in $\rRef$. For, it suffices to see that $f\in \CinC(E,F)$ is $\rho$-smooth. But
\eqref{ConvSettingClosedFunctor} gives that the derivative automatically smooth with value $L_\beta(E,F)$ by convenient smoothness is also smooth by composition with value $E^*_\rho\varepsilon F$ as expected. Since this equation only depends on the source space $E$ to be bornological, it extends to spaces for higher derivatives, hence the conclusion.

Hence we have a functor $\mathscr{S}:\Conv_\infty\to \Cref_{\infty-\Cref}$ for any $\mathscr{C}$ as above. We don't think this is an equivalence of category any more, as was the corresponding functor in \ref{sec:ConvenientSetting}. But finding a counterexample to essential surjectivity may be difficult, even thought we didn't really try.
\end{remark}

\subsection{$k$-smooth maps as model of DiLL}

We now turn to improve the $*$-autonomous category $\kref$ of section \ref{sec:kref} into a model of DiLL using the much stronger notiom of $k$-smooth map considered in subsection \ref{sec:ksmooth}.
For $X,Y\in \kref$, $C^\infty_{co}(X,Y)\subset \Cin(X,Y)$, hence there is a differential map $d:C^\infty_{co}(X,Y)\to C^\infty_{co}(X,L_\beta(X,Y))$ but it is by definition valued in $C^0_{co}(X,L_{co}(X,Y))$. But actually since the derivatives of these map are also known, it is easy to use the universal property of projective limits to induce a continuous map:  $d:C^\infty_{co}(X,Y)\to C^\infty_{co}(X,L_{co}(X,Y))$. Finally, note that 
 $L_{co}(X,Y)=X^*_k\varepsilon Y$, hence the space of value is the one expected for the dialogue category $Kc^{op}$ from Theorem \ref{kRef}. 
 
 For simplicity, in this section we slightly change $\kref$ to be the category of $k$-reflexive spaces of density character  smaller than a fixed inaccessible cardinal $\kappa$, in order to have a small category $\mathscr{C}=\kref$ and in order to define without change $\CinC(X,Y)$
 
 We call $\kref_\infty$ the category of $k$-reflexive spaces with maps $C^0_{co}(X,Y)$ as obtained in subsection \label{ksmooth}. We call $\Kc_\infty$ the category of $k$-quasi-complete spaces (with density character smaller than the same $\kappa$) with maps $\CinC(X,Y)$. This is easy to see that this forms a category by definition of $\CinC$.  We first check our assumptions to produce models of LL. We call $\CinC:\Kc_\infty\to \Kc^{op}$ the functor associating $\CinC(X)=\CinC(X,\R)$ to a space $X$.
 
 \begin{lemma}
 $(\Kc^{op},\varepsilon,\K,(\cdot)^*_\rho, \Kc_\infty,\times ,0,\CinC ,U)$ is a Seely linear model of $\lambda$-tensor logic.
 \end{lemma}
 
 \begin{proof}
 We checked in Theorem \ref{kRef} that $\mathcal{C}=(\Kc^{op},\varepsilon,\K,(\cdot)^*_\rho)$ is a  dialogue category. Completeness and cocompleteness are obvious using the $k$-quasicompletion functor to complete colimits in $\LCS$. Lemma \ref{thm:CartesianClosedBasic} gives the maps $\Xi,\Lambda$ and taking the first diagram as definition of $\Lambda^{\mathcal{M}}$ one gets Cartesian closedness of $\mathcal{M}= (\Kc_\infty,\times ,0,\CinC(\cdot,\cdot))$, and this result also gives the relative adjunction. The other compatibility diagrams are reduced to conveniently smooth maps $\Cin_{Fin}$ as in the proof of Theorem \ref{th:CRefSeely}.
 \end{proof}
 Note that since $\kref^{op}$ is already $*$-autonomous and isomorphic to its continuation category
 
 \begin{lemma}
  The categories of  the previous lemma satisfy the assumptions of subsection \ref{sec:GeneralDiLL}.
 \end{lemma}
 
 \begin{proof}
 The differential $\lambda$-category part reduces to convenient smoothness case. The above construction of $d$ make everything else easy.
 \end{proof}

\begin{theorem}\label{th:kRefDiLL}
 $\kref$ is also a complete Seely category with biproducts with $*$-autonomous structure extended by the comonad $!_{co}(\cdot)=(C_{co}^\infty(\cdot))^*_k$ associated to the adjunction with left adjoint $!_{co}:\kref_{co}\to \kref$ and right adjoint $U$. It gives a model of DiLL with codereliction $(d_0)^*_k$. 
\end{theorem}

\begin{proof}
Note that on $\kref$ which corresponds to $\mathcal{D}$ in the setting of  subsection \ref{sec:GeneralDiLL}, we know that $C_{co}^\infty=\CinC$ by the last statement in lemma \ref{thm:CartesianClosedBasic}. But our previous construction of $d$ implies that the new class of smooth maps obtained by the construction of subsection \ref{sec:GeneralDiLL} is again $C_{co}^\infty$. 
The result is a combination of Theorem \ref{LLtoDiLL}, \ref{DifflambdaTensortoDiLL} and the previous lemmas.
\end{proof}

\section{Conclusion}
This work is a strong point for the validity of the classical setting of Differential Linear Logic. Indeed, if the proof-theory of Differential Linear Logic is classical, we present here the first smooth models of Differential Linear Logic which comprehend the classical structure. Our axiomatization of the rules for differential categories within the setting of Dialogue categories can be seen as a first step towards a computational classical understanding of Differential Linear Logic. 
We plan to explore the categorical content of  our construction for new models of Smooth Linear Logic, and the diversity of models which can be constructed this way.
Our results also argue for an exploration of a classical differential term calculus, as initiated by Vaux \cite{Vaux}, and inspired by works on the computational signification of classical logic \cite{CurienHerbelin} and involutive  linear negation \cite{Munch}. 

The clarification of a natural way to obtain $*$-autonomous categories in an analytic setting suggests to reconsider known models such a \cite{Gir96} from a more analytic viewpoint, and should lead the way to exploit the flourishing operator space theory in logic, following the inspiration of the tract \cite{Gir04}. An obvious notion of coherent operator space should enable this.

This interplay between functional analysis, physics and logic is also strongly needed as seen the more and more extensive use of convenient analysis in some algebraic quantum field theory approaches to quantum gravity  \cite{RejznerBF}. Here the main need would be to improve the infinite dimensional manifold theory of diffeomorphism groups on non-compact manifolds. From that geometric viewpoint, differential linear logic went only half the way in considering smooth maps on linear spaces, rather that smooth maps on a kind of smooth manifold. By providing nice $?$-monads, our work suggests to try using $?$-algebras for instance in $k$-reflexive or $\rho$-reflexive spaces as a starting point (giving a base site of a Grothendieck topos) to capture better infinite dimensional features than the usual Cahier topos. Logically, this probably means getting a better interplay between intuitionist dependent type theory and linear logic. Physically, this would be useful to compare recent homotopic approaches \cite{Schenkel} with applications of the BV formalism \cite{RejznerYangMills,RejznerYangMills2}. Mathematically this probably means merging recent advances in derived geometry (see e.g. \cite{Toen}) with infinite dimensional analysis. Since we tried to advocate the way linear logic nicely captures (for instance with two different tensor products) infinite dimensional features, this finally strongly suggests for an interplay of parametrized analysis in homotopy theory and  parametrized versions of linear logic \cite{CurienFioreMunch}.

\section{Appendix}
We conclude with two technical lemmas only used to show we have built two different examples of $!$ on the same category $\rho$-Ref. 
{
\begin{lemma}\label{MSchKernel}
For any ultrabornological spaces $E_i$, any topological locally convex hull $E=\mathrm{\Sigma}_{i\in I}A_i(E_i)$, then we have the topological identity:
 $$\mathscr{S}( E'_\mu)=\mathrm{K}_{i\in I}(A_i^t)^{-1}(\mathscr{S}( (E_i)'_\mu).$$
\end{lemma}
\begin{proof}
We start with case where $E_i$ are Banach spaces. By functoriality one gets a map between two topologies on the same space (see for Mackey duals \cite[p 293]{Kothe}):
$$\mathscr{S}( E'_\mu)\to\mathrm{K}_{i\in I}(A_i^t)^{-1}(\mathscr{S}( (E_i)'_\mu))=:F.$$
In order to identify the topologies, it suffices to identify the duals and the equicontinuous sets on them.
From \cite[\S 22.6.(3)]{Kothe}, the dual of the right hand side is $F'=\mathrm{\Sigma}_{i\in I}(A_i^{tt})(\mathscr{S}( (E_i)'_\mu)'=\mathrm{\Sigma}_{i\in I}(B_i)(E_i)\to E$ where the injective continuous map to $E$ is obtained by duality of the previous surjective map (and the maps called $B_i$ again are in fact compositions of $A_i^{tt}$ and the isomorphism between $[\mathscr{S}( (E_i)'_\mu)]'=E_i$). From the description of $E$ the map above is surjective and thus we  must have $F'=E$ as vector spaces.

Let us now identify equicontinuous sets. From continuity of $\mathscr{S}( E'_\mu)\to F$ every equicontinuous set in $F'$ is also equicontinuous in $E=(\mathscr{S}( E'_\mu))'$. Conversely an equicontinuous set in $E=(\mathscr{S}( E'_\mu))'$ is contained in the absolutely convex cover of a null-sequence $(x_n)_{n\geq 0}$ for the bornology of absolutely convex weakly-compact sets, (thus also for the bornology of Banach disks \cite[Th 8.4.4 b]{Jarchow}). By a standard argument, there is $(y_n)_{n\geq 0}$ null sequence of the same type such that $(x_n)_{n\geq 0}$ is a null sequence for the bornology of absolutely convex compact sets in a Banach space $E_B$ with $B$ the closed absolutely convex cover of $(y_n)_{n\geq 0}$.

Of course $(y_n)_{n\leq m}$ can be seen inside a minimal finite sum $G_m=\mathrm{\Sigma}_{i\in I_m}(B_i)(E_i)'$ and $G_m$ is increasing in $F$ so that one gets a continuous map $I:\mathrm{ind}\ \lim_{m\in\N} G_m\to F'$. Moreover each $G_m$ being a finite hull of Banach space, it is again a Banach space thus one gets a linear map $j:E_B\to \mathrm{ind}\ \lim_{m\in\N} G_m=G$. Since $I\circ j$ is continuous, $j$ is a sequentially closed map, $E_B$ is Banach space, $G $ a (LB) space therefore a webbed space, by De Wilde's closed graph theorem \cite[\S 35.2.(1)]{Kothe2}, one deduces $j$ is continuous. Therefore by Grothendieck's Theorem \cite[\S 19.6.(4)]{Kothe}, there is a $G_m$ such that $j$ is valued in $G_m$ and continuous again with value in $G_m$. 
 Therefore $(j(x_n))_{n\geq 0}$ is a null sequence for the bornology of absolutely convex compact sets in $G_m$. We want to note it is equicontinuous there, which means it is contained in a sum of equicontinuous sets.  

By \cite[\S 19.2.(3)]{Kothe}, $G_m$ is topologically a quotient by a closed linear subspace $\bigoplus_{i\in I_m}(B_i)(E_i)'/H$.
By \cite[\S 22.2.(7)]{Kothe} every compact subset of the quotient space $\bigoplus_{i\in I_m}(B_i)(E_i)'/H$ of a Banach space by a closed subspace $H$ is a canonical image of a compact subset of the direct sum, which can be taken a product of absolutely convex covers of null sequences.
 Therefore our sequence $(j(x_n))_{n\geq 0}$ is contained in such a product which is exactly an equicontinuous set in $G_m=\Big(\mathrm{K}_{i\in I_m}(A_i^t)^{-1}(\mathscr{S}( (E_i)'_\mu)\Big)'$ \cite[\S 22.7.(5)]{Kothe} (recall also that for a Banach space $\mathscr{S}( (E_i)'_\mu)=(E_i)'_c$). Therefore it is also equicontinuous in  $F'$ (by continuity of $F\to G_m'$). This concludes to the Banach space case.

For the ultrabornological case decompose $E_i$ as an inductive limit of Banach spaces. Get in this way a three terms sequence of continuous maps with middle term $\mathrm{K}_{i\in I}(A_i^t)^{-1}(\mathscr{S}( (E_i)'_\mu)$ and end point the corresponding iterated kernel coming from duals of Banach spaces by transitivity of Kernels/hulls. Conclude by the previous case of equality of topologies between the first and third term of the sequence, and this concludes to the topological equality with the middle term too.
\end{proof}

\begin{lemma}\label{WhyNotAP}For any \lcs $E$,
$((C^\infty_{Fin}(E))^*_\rho)^*_\rho$ is  Hilbertianizable, hence it has the approximation property.
\end{lemma}
\begin{proof}
We actually show that $F=((C^\infty_{Fin}(E))^*_\rho)^*_\rho=\mathscr{S}[\Big(\mathscr{C}_M\Big[(C^\infty_{Fin}(E))'_\mu\Big]\Big)'_\mu]$ is Hilbertianizable (also called a (gH)-space) \cite[Rmk 1.5.(4)]{hollstein}.

Note that $G=C^\infty_{Fin}(E)$ is a complete nuclear space. It suffices to show that for any complete { nuclear} space $G$, $\mathscr{S}[\Big(\mathscr{C}_M\Big[G'_\mu\Big]\Big)'_\mu]$ is a complete (gH) space. { Of course, we use lemma \ref{ordinalCompletion} but we need another description of the Mackey completion $\mathscr{C}_M^\lambda(G'_\mu)$ . We let $E_0=G'_\mu, E_{\lambda+1}=\cup_{\{x_n\}\in RMC(E_\lambda)}\overline{\gamma(x_n,n\in \N)} E_\lambda=\cup_{\mu<\lambda}E_\mu$ for limit ordinals.

Here $RMC(E_\lambda)$ is the set of sequences $(x_n)\in E_\lambda^{\N}$ which are rapidly Mackey-Cauchy in the sense that if $x$ is their limit in the completion there is a bounded disk $B\subset E_{\lambda+1}$ such that for all $k$, $(x_n-x)\in n^{-k}B$ for $n$ large enough. For $\lambda_0$ large enough, $E_{\lambda_0+1}=E_{\lambda_0}$ and any Mackey-Cauchy sequence $x_n$ in $E_{\lambda_0}$, let us take its limit $x$ in the completion and $B$ a closed bounded disk in $E_{\lambda_0}$ such that $||x_n-x||_B\to 0$ one can extract $x_{n_k}$ such that $||x_{n_k}-x||_B\leq k^{-k}$ so that $(x_{n_k}-x)\in k^{-l}B$ for $k$ large enough (for any $l$) thus $(x_{n_k})\in RMC(E_{\lambda_0})$ thus its limit is in  $E_{\lambda_0+1}=E_{\lambda_0}$ which is thus Mackey-complete. To apply lemma \ref{ordinalCompletion} with $D=\mathscr{N}((\cdot)'_\mu))$ one needs to see that $\{x_n,n\in \N \}$ is equicontinuous in $D(E_{\lambda_0})'$. But since $E_{\lambda_0}$ is Mackey-complete, one can assume the bounded disk $B$ is a Banach disk and $||x_n-x||_B=O(n^{-k})$ so that $x_n$ is rapidly convergent. From \cite[Prop 21.9.1]{Jarchow} $\{(x_n-x),n\in \N\}$ is equicontinuous for the strongly nuclear topology associated to the topology of convergence on Banach disks and a fortiori equicontinuous for $D(E_{\lambda_0})'$. By translation, so is $\{x_n,n\in \N\}$ as expected. From application of lemma \ref{ordinalCompletion}, 
$H^{\lambda_0}:=\mathscr{N}[(\mathscr{C}_M(G'_\mu))'_\mu]$ is complete since $\mathscr{N}[(G'_\mu)'_\mu]$ is already complete ($G$ is complete nuclear so that $\mathscr{N}[(G'_\mu)'_\mu]\to G$ continuous and use again \cite[IV.5 Rmq 2]{Bourbaki}).
}

$H^{\lambda_0}$ is  nuclear thus a (gH)-space. Since $H^{\lambda_0}$ is a complete (gH) space, it is a reduced projective limit of Hilbert spaces \cite[Prop 1.4]{hollstein} and semi-reflexive \cite[Rmk 1.5 (5)]{hollstein}. Therefore its Mackey=strong dual \cite[\S 22.7.(9) ]{Kothe} is an inductive limit of the Mackey duals, thus Hilbert spaces.

One can apply lemma \ref{MSchKernel} to get $\mathscr{S}([H^{\lambda_0}]_\mu)$ as a projective kernel of $\mathscr{S}(H)$ with $H$ Hilbert spaces. But from \cite[Thm 4.2]{Bellenot} this is the universal generator of Schwartz (gH) spaces, therefore the projective kernel is still of (gH) space. 

For $\lambda_0$ as above, this concludes to $\mathscr{S}[(\mathscr{C}_M((C^\infty_{Fin}(E))'_\mu))'_\mu]$ (gH) space, as expected.
\end{proof}

\end{document}